\documentclass[12pt,letterpaper]{article}

\usepackage[T1]{fontenc}
\usepackage{lmodern}
\usepackage[utf8]{inputenc}
\usepackage{geometry}
\usepackage{nicefrac}
\usepackage{booktabs}
\usepackage{amsmath}
\usepackage{amsthm} 
\usepackage{amsfonts}
\usepackage{dsfont}
\usepackage{cancel}
\usepackage{mathtools}
\usepackage{datetime}
\usepackage{bm}
\usepackage{paralist}
\usepackage{changepage}
\usepackage{amsthm}
\usepackage{verbatim}
\usepackage{amsmath,amssymb}
\usepackage{graphicx}
\usepackage{emptypage}
\usepackage{mathabx}
\usepackage{newlfont}
\usepackage{tikz-cd}
\usepackage{tikz}
\usepackage[all,cmtip]{xy}
\usepackage[bottom]{footmisc}
\usepackage[titletoc,title]{appendix}
\usepackage{chngcntr}
\usepackage{apptools}
\usepackage{listings}
\usepackage{color}
\usepackage{sgame, tikz} 

\usepackage{kpfonts}  
\usepackage{natbib}
\AtAppendix{\counterwithin{lem}{section}}
\usepackage{caption}
\usepackage{subcaption}
\usepackage{float}
\usepackage{thmtools,thm-restate}
\usepackage{epigraph}
\usepackage{xr-hyper}

\usepackage[colorlinks=true,allcolors=blue]{hyperref}

\makeatletter
\renewcommand\@makefnmark{\hbox{\@textsuperscript{\normalfont\color{black}\@thefnmark}}}
\makeatother

\usepackage{wasysym}
\hypersetup{colorlinks,linkcolor=blue,allcolors=blue}
\usepackage{sgamevar}
\hypersetup{urlcolor=blue}

\theoremstyle{plain} 
\newtheorem{cor}{Corollary} 
\newtheorem{prop}{Proposition}
 
\newtheorem{theorem}{Theorem}
\newtheorem{lemma}{Lemma}
\theoremstyle{definition} 
\newtheorem{ex}{Example}

 \newtheorem{taggedtheoremx}{Theorem}
\newenvironment{taggedtheorem}[1]
{\renewcommand\thetaggedtheoremx{#1}\taggedtheoremx}
 {\endtaggedtheoremx}
\newtheorem{taggedcorollaryx}{Corollary}
 \newenvironment{taggedcorollary}[1]
{\renewcommand\thetaggedcorollaryx{#1}\taggedcorollaryx}
 {\endtaggedcorollaryx}
\allowdisplaybreaks
\makeatletter
\renewenvironment{proof}[1][\proofname]{%
  \par\pushQED{\qed}\normalfont%
  \topsep6\p@\@plus6\p@\relax
  \trivlist\item[\hskip\labelsep\bfseries#1\@addpunct{.}]%
  \ignorespaces
}{%
  \popQED\endtrivlist\@endpefalse
}
\makeatother
\theoremstyle{remark} 
\newtheorem{rmk}{Remark}

\let\emptyset\varnothing

\DeclareMathOperator{\E}{\mathds{E}}
\renewcommand{\P}{\mathds{P}}

\newcommand{\F}{\mathcal{F}}

\newcommand{\R}{\mathds{R}}

\newcommand{\cQ}{\mathcal{Q}}
\newcommand{\N}{\mathds{N}}

\newcommand{\M}{\mathcal{M}}

\newcommand{\lb}{\underline}
\newcommand{\ub}{\overline}

\renewcommand{\1}{\mathds{1}}

\renewcommand{\epsilon}{\varepsilon}
\renewcommand*\d{\mathop{}\!\mathrm{d}}

\usepackage{setspace}

\theoremstyle{definition}

\reversemarginpar

\newdate{draftdate}{02}{10}{2023}
\usdate

\usepackage[nameinlink]{cleveref}
\crefname{manualasm}{assumption}{assumptions}
\crefname{taggedtheoremx}{theorem}{theorems}
\crefname{taggedcorollaryx}{corollary}{corollaries}
\crefname{cor}{corollary}{corollaries}
\crefalias{prop}{proposition}
\crefname{claim}{claim}{claims}
\crefname{ex}{example}{examples}
\crefname{defn}{definition}{definitions}
\crefname{rmk}{remark}{remarks}
\crefname{alg}{algorithm}{algorithms}

\begin{document}
\title{\textbf{Multidimensional Monotonicity and Economic Applications}\thanks{We thank Dirk Bergemann, Ben Brooks, Piotr Dworczak, Thomas Gresik, Marina Halac, Emir Kamenica, Andreas Kleiner, Maciej Kotowski, Ilia Krasikov, Elliot Lipnowski, Doron Ravid, Phil Reny, Fedor Sandormirskiy, Philipp Strack and Jidong Zhou, as well as seminar and conference audiences at University of Chicago, Yale, EC'25, Notre Dame, Sciences Po, NYU Shanghai for their valuable comments and suggestions. We also thank Tsong-Hong Tenn for his valuable research assistance.}
}
\author{Frank Yang\thanks{Department of Economics, Harvard University. Email: fyang@fas.harvard.edu} \and Kai Hao Yang\thanks{School of Management, Yale University. Email: kaihao.yang@yale.edu}}
\date{\today
}
\maketitle
\begin{abstract}
We characterize the extreme points of multidimensional monotone functions from $[0,1]^n$ to $[0,1]$, as well as the extreme points of the set of one-dimensional marginals of these functions. These characterizations lead to new results for various mechanism design and information design problems, including public good provision with interdependent values; interim efficient bilateral trade mechanisms; mechanism (anti) equivalence; asymmetric reduced form auctions; and optimal private private information structure. \\


\end{abstract}
\setcounter{page}{1}
\newpage

\addtocontents{toc}{\protect\setcounter{tocdepth}{2}} 
\newpage

\section{Introduction}

Many important economic problems involve the design of allocation rules or information structures. In an inspiring recent contribution, \citet*{kleiner2021extreme} show that many mechanism and information design problems involving a single agent, or symmetric agents, reduce to finding the extreme points of one-dimensional monotone functions subject to a majorization constraint. This extreme-point perspective is very useful as it uncovers the structural properties of these design problems and immediately leads to characterizations of optimal mechanisms.

In many settings, however, the agents are naturally asymmetric (e.g., bilateral trade), and their payoffs depend on the entire ex-post allocation rule (e.g., public good provision with interdependent values). 
In this paper, we show that a common structure behind these problems involves selecting a multidimensional monotone function that satisfies certain properties. 

We distinguish between these two cases: \textit{(i)} the payoff-relevant information depends on the entire ex-post allocation rule, and \textit{(ii)} the payoff-relevant information depends only on the interim allocation rules, but they can be asymmetric. The structural properties of these problems can be understood via the extreme points of \textit{(i)} the set of multidimensional monotone functions from $[0, 1]^n$ to $[0, 1]$, and \textit{(ii)} the set of their one-dimensional marginals. We characterize the extreme points of these sets, and show that these characterizations immediately lead to various new results in well-known mechanism design problems---including public good provision, bilateral trade, reduced form auctions, and mechanism equivalence---as well as in recent information design problems.

\paragraph{Abstract Results.}\hspace{-2mm}We start by describing our abstract results before discussing the economic applications in detail. A subset $A \subseteq [0,1]^n$ is an \textit{\textbf{upward-closed set}} (henceforth, an \textit{\textbf{up-set}}) if for any $x \leq y$, $x \in A$ implies $y \in A$ (see \Cref{fig:up-set}), where the partial order $\leq$ on $[0,1]^n$ is the usual component-wise order. A function $f:[0, 1]^n \rightarrow [0, 1]$ is \textit{\textbf{monotone}} if $f(x) \leq f(y)$ for all $x \leq y$.  The set of monotone functions is a compact and convex set, and thus admits extreme points. An elegant result of \citet{choquet1954theory} identifies the extreme points of this set. These extreme points are given by indicator functions $\1_{A}$ defined on an up-set $A \subseteq [0, 1]^n$. In the one-dimensional case, these extreme points are exactly the one-jump step functions $\{\1_{[k,1]}\}_{k \in [0,1]}$. Combined with Choquet's integral representation theorem, this implies that any monotone function can be represented as a mixture over such up-set functions.
However, when $n \geq 2$, such a representation is not unique in general, and the complexity of the mixture naturally increases with $n$. 

Our first result (\Cref{thm:ordered-up-sets}) refines Choquet's representation and reduces its complexity: We show that any monotone function can, in fact, be represented as a mixture over a collection of \textit{\textbf{nested}} up-set functions (i.e., these up-sets are totally ordered by set inclusion). Moreover, such a nesting representation is unique. This representation is useful as it allows us to equivalently view a multidimensional monotone function as a totally ordered chain of up-sets coupled with a one-dimensional probability distribution. Our second result (\Cref{thm:nested-up-sets}) exploits this structure and well-known results regarding one-dimensional distributions to characterize the extreme points of multidimensional monotone functions subject to $m$ linear constraints: they can always be represented as a mixture of $m+1$ nested up-sets (see \Cref{fig:nested-up-set}).

\begin{figure}[t]

  \begin{subfigure}[b]{0.45\textwidth}
    \centering
    \includegraphics[scale=0.25]{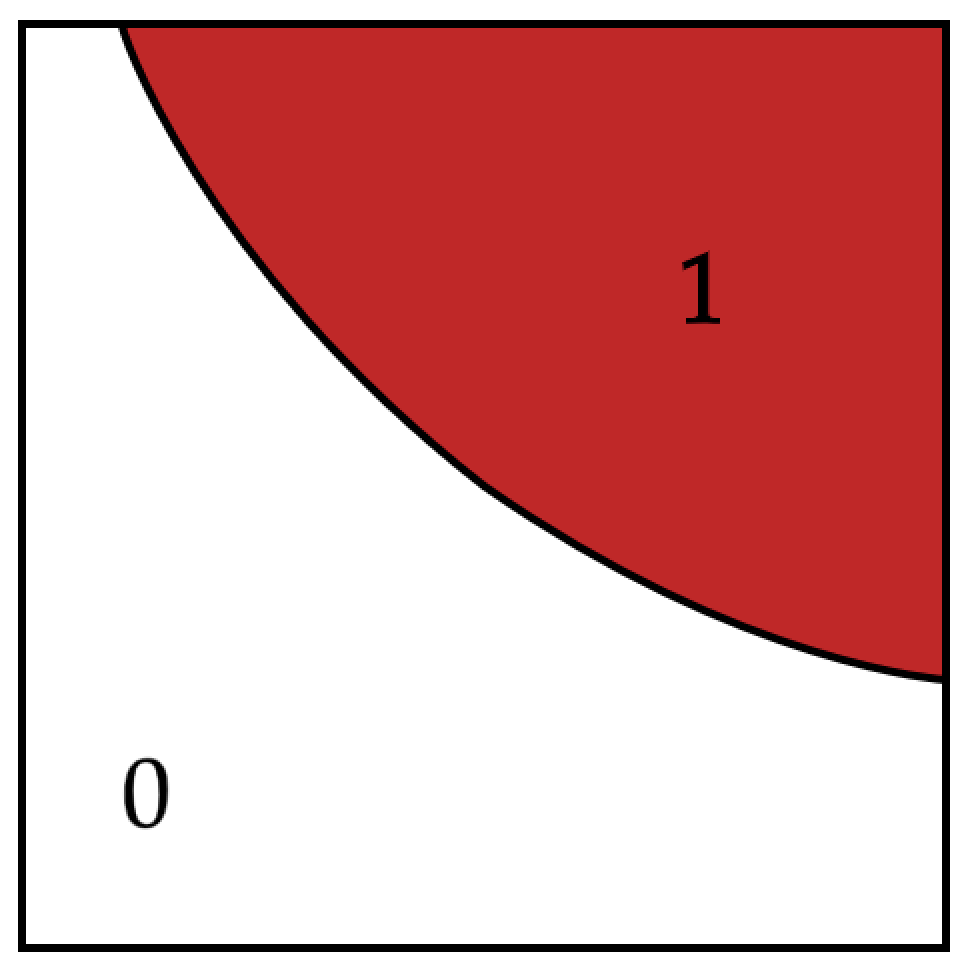} 
    \caption{An up-set function $\1_{A}$    \label{fig:up-set}}
  \end{subfigure}
  \quad 
  \begin{subfigure}[b]{0.45\textwidth}
    \centering
    \includegraphics[scale=0.25]{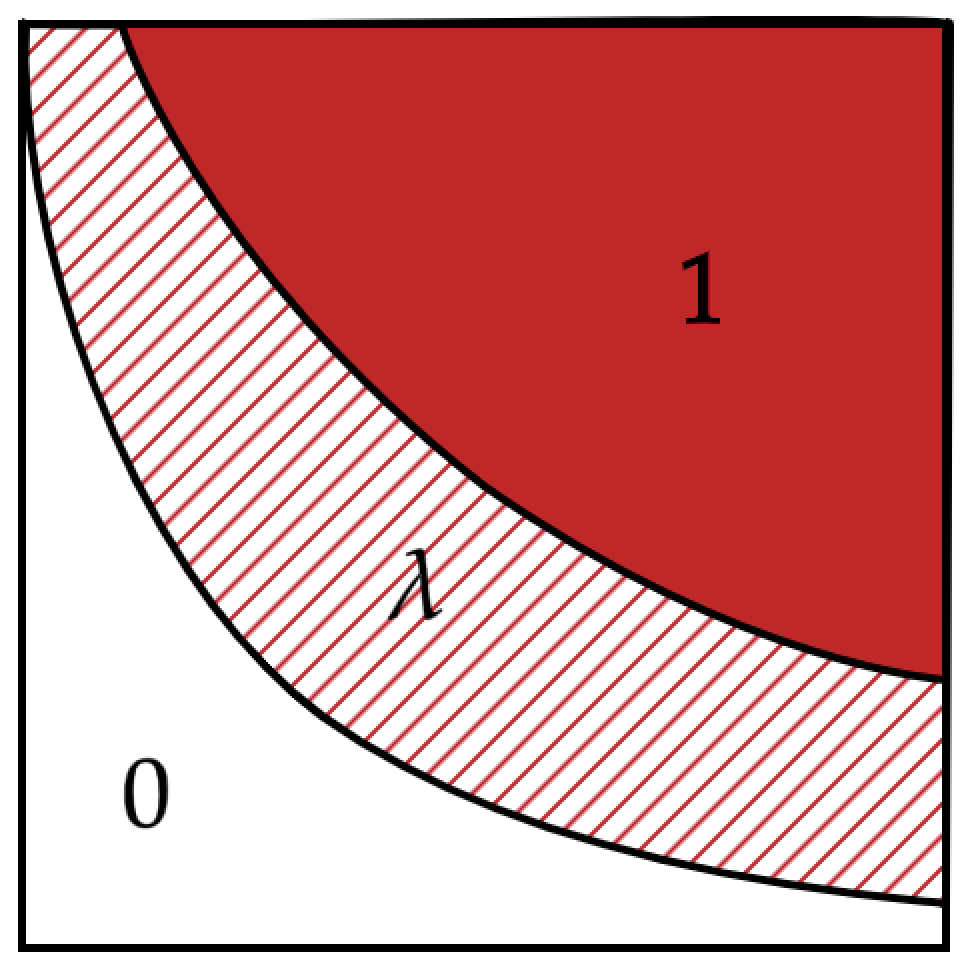}
    \caption{Mixture of two nested up-set functions\label{fig:nested-up-set}}
  \end{subfigure}
  
  \vspace{0.5cm} 
  \begin{subfigure}[b]{0.45\textwidth}
    \centering
    \includegraphics[scale=0.25]{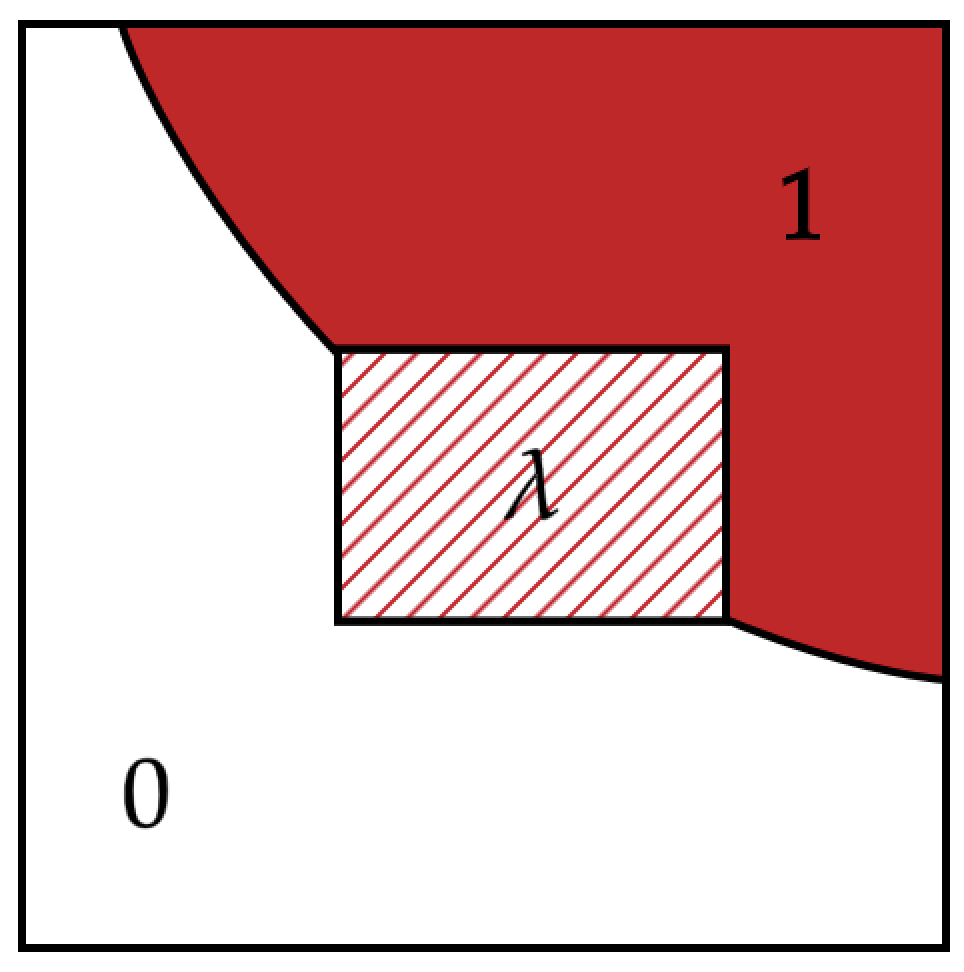} 
    \caption{Nested up-sets differ by a rectangle\label{fig:rectangle}}
 
  \end{subfigure}
  \quad 
  \begin{subfigure}[b]{0.45\textwidth}
    \centering
    \includegraphics[scale=0.25]{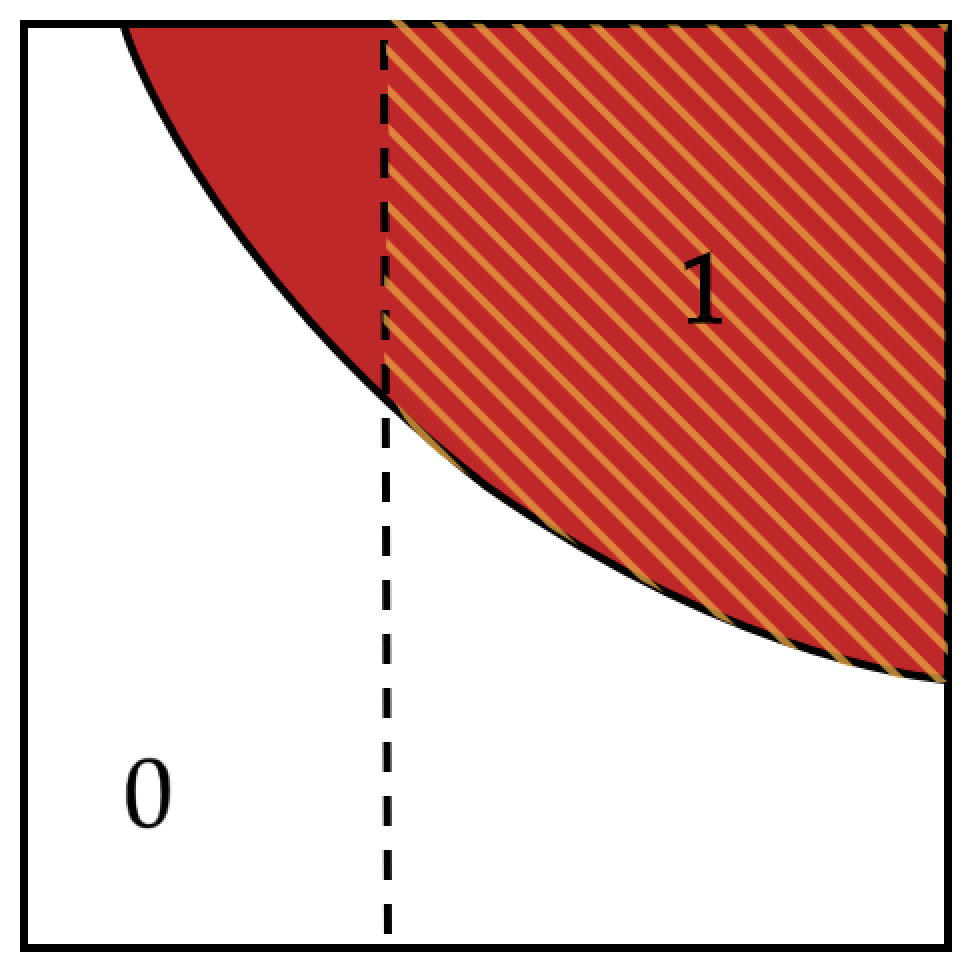} 
    \caption{Truncation of up-set for one marginal\label{fig:truncation}}
  \end{subfigure}
  \caption{Illustration of extreme points for different sets of functions}
  \label{fig:extreme-points}
\end{figure}

Building on these characterizations, our next set of results studies what we call \textit{\textbf{rationalizable monotone functions}}. 
For any integrable function $f: [0, 1]^n \rightarrow [0, 1]$, we say that the marginal of $f$ in dimension $i$ is the function $x_i \mapsto \int f(x) \d x_{-i}$. A collection of one-dimensional monotone functions $q:=(q_i)_{i=1,\dots n}$ is \textit{\textbf{rationalizable}} if there exists some function $f:[0, 1]^n \rightarrow [0, 1]$ such that $q_i$'s are exactly the marginals of $f$. Note that this definition does not require the function $f$ to be monotone. However, an elegant result by \citet{gutmann1991existence} shows that any such $q$ that is rationalizable by some function must also be rationalizable by a \textit{monotone} function. We exploit this connection and our characterizations of multidimensional monotone functions to characterize the extreme points of the set of rationalizable monotone functions $q$. In particular, we show that  (\Cref{thm:extreme-are-up-sets}) every extreme point of the set of rationalizable monotone functions must be the one-dimensional marginals of an up-set function $\1_{A}$, and every extreme point of rationalizable monotone functions subject to $m$ linear constraints must also be the one-dimensional marginals of a mixture of at most $m+1$ such $\1_{A_k}$, where the sets $\{A_k\}_k$ are nested up-sets.

In the case of rationalizable pairs $q:= (q_1, q_2)$, we further sharpen \Cref{thm:extreme-are-up-sets} to a necessary and sufficient characterization of the extreme points. Without any linear constraint, we show that (\Cref{thm:rationalized-upsets}) the one-dimensional marginals of the up-set functions $\1_A$ are exactly the extreme points of the rationalizable pairs, and these extreme points must be  \textit{\textbf{uniquely}} rationalized. With one linear constraint, we show that (\Cref{thm:rectangle}) the extreme points of rationalizable pairs are characterized by a mixture of $\1_{A_1}$ and $\1_{A_2}$, where $A_1$ and $A_2$ are nested up-sets that differ by a single \textit{\textbf{rectangle}} (see \Cref{fig:rectangle}). Moreover, these extreme points must be \textit{\textbf{uniquely}} rationalized among all monotone functions.

In the case of rationalizable pairs, it is known that a pair of monotone functions $(q_1, q_2)$ is rationalizable if and only if $q_1$ is majorized by the conjugate of $q_2$ (\citealt{gutmann1991existence}).\footnote{That is, $q_1(z)$ is majorized by $\hat{q}_2(z):=1-q_2^{-1}(1-z)$.} Building on this connection, our characterization immediately gives the extreme points of this joint majorization set: They are exactly the pairs where the majorization constraint binds everywhere (\Cref{prop:reverse-majorization}). This structure is very different from the one-dimensional case (\citealt*{kleiner2021extreme}), precisely because these majorization constraints are imposed on a pair of monotone functions, linking them between the original space $q_1$ and the inverse space $q^{-1}_2$. Using this logic, we also characterize the extreme points of the joint majorization set under a linear constraint (\Cref{prop:square-majorization}), and the extreme points of the weak joint majorization set (\Cref{prop:weak}) where the majorization relation is replaced by weak majorization. For weak joint majorization, any extreme point $(q_1, q_2)$ has the structure that $q_2$ is the marginal of an up-set function $\1_A$, and $q_1$ is the marginal of another up-set function $\1_{\tilde{A}}$ where $\tilde{A}$ is a \textit{\textbf{truncation}} of $A$ (see \Cref{fig:truncation}).

\paragraph{Economic Applications.}\hspace{-2mm}We now discuss the economic applications of our results (see \Cref{sec:application} for details). We apply our abstract results to obtain new results in four well-known mechanism design problems and one recent information design problem.  

Our first application studies the classical problem of public good provision, where a designer decides whether to implement a public project subject to an ex-ante budget constraint for financing the project. Almost all existing studies of this problem focus on the case of independent private values (see e.g., \citealt{guth1986private}). In practice, the values of the agents for a public good seem likely to be interdependent and correlated. However, under Bayesian incentive compatibility, it is well known that the designer can generically achieve the first best by constructing bets \'{a} la \citet{cremer1988full}, as long as there are correlated signals. This may not be satisfying, due to the fragility of such mechanisms and calls for other approaches to this problem (\citealt{brooks2023robust}). A natural candidate is to study the design of \textit{\textbf{ex-post}} incentive-compatible (IC) and individually rational (IR) mechanisms. By a standard argument, an allocation rule is ex-post IC if and only if it is monotone in each agent's signal. Because the objective is a linear functional, our extreme point results (recall \Cref{fig:nested-up-set}) immediately imply that the optimal ex-post IC mechanism is simply a \textit{\textbf{two-threshold policy}}, where the designer aggregates the signals of each agent into a score, and if the score is higher than the high threshold, the project is implemented for sure, and if it is lower than the high threshold but higher than the low threshold, the project is implemented with some positive probability (and otherwise abandoned). The specifics of the scoring rule and the two thresholds will depend on the correlation structure of the agents' signals, but the optimal mechanism turns out to always have this simple structure, regardless of the correlation structure. Even though the Myersonian ironing approach does not apply due to the multidimensional monotonicity constraint, the extreme point structure allows us to fully characterize the optimal mechanism in specific settings. In particular, we illustrate our approach in a setting with limited negative externalities and derive a new mechanism (\textit{\textbf{externality-refund mechanism}}) for implementing the second-best allocation.  

Our second application studies another well-known mechanism design problem, bilateral trade \'{a} la \citet{myerson1983efficient}. An important question in bilateral trade is to understand its \textit{\textbf{interim-efficient frontier}} (\citealt{myerson1984two}). It has been shown that in some special cases depending on the welfare weights and type distributions, the interim-efficient mechanism takes a familiar form of trading if and only if the suitably defined virtual value is above the virtual cost (\citealt{ledyard1999characterization,ledyard2007general}). However, to understand the frontier, we have to characterize the optimal mechanism for arbitrary welfare weights and type distributions. Because any interim-efficient mechanism must maximize a linear functional over the pair of interim allocation rules for the buyer and seller, and the ex-post budget constraint can be represented as a linear constraint, our extreme point results (recall \Cref{fig:rectangle}) immediately yield a new class of trading mechanisms that can attain the optimal welfare for any welfare weights and type distributions. We call these trading mechanisms \textit{\textbf{markup-pooling mechanisms}}: In such a mechanism, trade is implemented if and only if the value $v$ is above a monotone transformation of the cost $\phi(c)$ (the \textit{\textbf{markup function}}), with the exception that when $c$ falls into a single interval $I$ (the \textit{\textbf{pooling interval}}), we \textit{resample} the cost from the two ends of the interval, and execute the trade if the buyer's value $v$ is above the marked up, resampled cost $\phi(\tilde{c})$. 

Our third application studies the question of interim equivalence between mechanisms initiated by \citet{manelli2010bayesian}. Two mechanisms are \textit{\textbf{payoff-equivalent}} if they yield the same ex-ante expected surplus and the same interim expected utilities for all agents. From a series of fundamental contributions (\citealt{manelli2010bayesian}; \citealt{gershkov2013equivalence}; \citealt{chen2019equivalence}), it is known that \textit{(i)} for any Bayesian incentive compatible (BIC) mechanism, there exists a payoff-equivalent dominant strategy incentive compatible (DIC) mechanism (\textit{\textbf{BIC-DIC equivalence}}); and \textit{(ii)} for any stochastic mechanism, there exists a payoff-equivalent deterministic mechanism (\textit{\textbf{stochastic-deterministic equivalence}}). These mechanism equivalence results are very helpful, because DIC mechanisms are robust to changes in the agents' beliefs and deterministic mechanisms do not require a credible randomization device. However, as another application of our results, we show that there is a strong conflict between asking for DIC mechanisms and asking for deterministic mechanisms. In particular, we present a \textit{\textbf{mechanism anti-equivalence theorem}} for social choice problems with two agents and two alternatives: A BIC mechanism is payoff-equivalent to a deterministic DIC mechanism if and only if they are ex-post equivalent. It turns out that, although every BIC mechanism has a payoff-equivalent DIC counterpart and every stochastic mechanism has a payoff-equivalent deterministic counterpart, the interim payoffs under a mechanism that is both deterministic and DIC cannot be replicated by \textit{any} other mechanism. Consequently, among DIC mechanisms, a stochastic mechanism and a deterministic mechanism must yield different interim payoffs; among deterministic mechanisms, a BIC mechanism and a DIC mechanism must yield different interim payoffs, unless they are ex-post equivalent. Thus, our result implies that the BIC-DIC equivalence necessarily relies on randomization, and the stochastic-deterministic equivalence necessarily relies on non-DIC mechanisms.

Our fourth application studies asymmetric reduced form auctions. A \textit{\textbf{reduced form auction}} is a collection of interim allocation probabilities that is implementable by an ex-post auction allocation rule (\citealt{matthews1984implementability}; \citealt{border1991implementation}).  \citet*{kleiner2021extreme} characterize the extreme points of symmetric reduced form auctions, by representing the symmetric Border's condition as a majorization constraint on a one-dimensional monotone function. For asymmetric auctions, suppose that there are two bidders. Let $q_i$ be the allocation probability to bidder $i$ as a function of bidder $i$'s type in the quantile space. Exploiting the extreme-point structure of the weak joint majorization set, we show that $(q_1, q_2)$ is implementable by an auction if and only if $q_1$ is weakly majorized by the inverse of $q_2$, with every extreme point of that set characterized by $q_1 = q_2^{-1} \1_{[k_1, 1]}, q_2 = q_1^{-1} \1_{[k_2, 1]}$ for some $k_1, k_2$.\footnote{Using a different approach, \citet*{he2024private} also connect the two-bidder asymmetric reduced form auctions to majorization, though they do not characterize the extreme points of this set.} These extreme points correspond to \Cref{fig:truncation} (up to relabeling), where the left truncation of the up-set $A$ is the region where the seller keeps the object, and the remaining up-set is the region of types that we allocate to bidder $1$ (similar truncation for the complement of $A$ describes the allocation to bidder $2$). These extreme mechanisms include the standard Myersonian auction (\citealt{Myerson1981}) as well as non-auction mechanisms such as sequential posted prices (\citealt{gershkov2021theory}). We apply this extreme point characterization to auctions with endogenous values studied in \citet{gershkov2021theory}, where revenue maximization becomes a convex objective, and show how our characterization helps illuminate their examples of optimal asymmetric mechanisms.

Our last application studies private private information structures introduced by \citet*{he2024private}. Consider a binary state and $n$ agents with a common prior about the state. A \textit{\textbf{private}} \textit{\textbf{private}} information structure sends an independent signal to each agent so that they update their belief about the state but not about other agents' signals. \citet*{he2024private} characterize the Blackwell-Pareto undominated information structures and use a garbling argument to characterize the set of \textit{\textbf{feasible belief distributions}} that can be induced by some private private information structure. Complementary to their analysis, we use our extreme point results to provide a convex-hull characterization of the \textit{\textbf{feasible belief quantiles}} under private private information. Unlike the feasible belief distributions, the set of feasible belief quantiles is convex and admits extreme points, which we also characterize using our abstract results. In particular, we show that any extreme feasible belief quantiles can be implemented by a \textit{\textbf{nested bi-upset}} signal structure, where we randomize between two nested up-sets in $[0, 1]^n$, and uniformly draw signals from there. This is because a collection of belief quantile functions is feasible if and only if it is rationalizable in our sense and satisfies one linear constraint (i.e., the Bayes plausibility constraint). As an immediate consequence, for many information design problems subject to the privacy constraint, a nested bi-upset signal would be optimal. Moreover, as we show, a pair $(q_1, q_2)$ is an extreme point of the feasible belief quantiles if and only if it can be implemented using a nested bi-upset signal where the nested up-sets differ by at most a rectangle. Combining with the characterization of Blackwell-Pareto frontier from \citet*{he2024private}, our result then implies that in the case of two agents, any such optimal private private information---even though it may not be Blackwell-Pareto undominated---can be attained by \textit{(i)} selecting a Blackwell-Pareto undominated signal and then \textit{(ii)} applying a single \textit{\textbf{interval-pooling}} for one of the two agents.

\subsection{Related Literature}

We study the convex set of multidimensional monotone functions and their marginals. In particular, we characterize the extreme points of the set of multidimensional monotone functions subject to finitely many linear constraints, as well as the extreme points of the rationalizable monotone functions. \citet{choquet1954theory} provides a well-known characterization of extreme points of monotone functions, stating that a monotone function is an extreme point if and only if it is an indicator function defined on an up-set. We use Choquet's characterization to further characterize the extreme points of monotone functions subject to finitely many linear constraints. Our characterization sharpens the set of candidate extreme points derived from \citet{Winkler1988} to finite mixtures of indicator functions defined on \textit{nested} up-sets.  

Our results on rationalizable monotone functions build on and contribute to the mathematics literature on the existence of a joint density for given marginals (\citealt{lorenz1949}; 
 \citealt{kellerer1961funktionen}; \citealt{Strassen1965}). In particular, \citet{gutmann1991existence} provide conditions on when a collection of monotone functions is rationalizable, and show the equivalence between the rationalizability by an arbitrary function, by a monotone function, and by an indicator function.\footnote{Specifically, \citet{gutmann1991existence} show these three results using complementary but separate approaches: the rationalizability characterization follows from \citet{Strassen1965}, which in turn is built upon supporting hyperplanes, instead of extreme points, of convex sets; rationalizability by monotone functions follows from a convex minimization program, which is later  generalized by \citet{gershkov2013equivalence}; whereas the proof of rationalizability by indicator functions relies on the existence and the characterization of extreme points of a different convex set compared to our focus---the set of (not necessarily monotone) functions on $[0,1]^n$ with \textit{fixed} marginals.} Exploiting this equivalence and our results on multidimensional monotone functions, we characterize the extreme points of the set of rationalizable monotone functions, with and without linear constraints. The unique rationalizability property of our extreme points in the two-dimensional case also connects to and builds on results from the mathematical tomography literature on ``sets of uniqueness'' (\citealt{lorenz1949,fishburn1990sets,kemperman1991sets,kellerer1993uniqueness}), which studies necessary and sufficient conditions for a function to be uniquely determined by its marginals.

In economics, several recent papers characterize extreme points of one-dimensional monotone functions under various constraints and derive economic applications.
\citet*{kleiner2021extreme} characterize the extreme points of the set of nondecreasing functions that majorize, or are majorized by, some fixed nondecreasing function. They then derive applications on reduced-form implementation, BIC-DIC equivalence, delegation, and persuasion.\footnote{See also \citet{arieli2023optimal}.} \citet{nikzad2023constrained} and \citet{candogan2023disclosure} further characterize the extreme points of these convex sets subject to finitely many linear constraints. Meanwhile, \citet{yang2024monotone} characterize the extreme points of the set of nondecreasing functions that are pointwise bounded by two other nondecreasing functions, and derive applications to gerrymandering, quantile-based persuasion, overconfidence, and security design. Our results both complement these recent characterizations and provide new applications in various settings as well. 

A one-dimensional monotone function can be equivalently viewed as a probability distribution. However, this equivalence breaks down for multidimensional monotone functions.\footnote{Indeed, recall that a bivariate CDF $F$ must also satisfy $F(x'_1,x'_2) - F(x_1,x'_2) - F(x'_1,x_2) + F(x_1,x_2) \ge 0$ for all $x_1 < x'_1$ and $x_2 < x'_2$.} Complementary to our analysis of multidimensional monotone functions, \citet{kleiner2024extreme} characterize the extreme points of multidimensional probability distributions that are dominated by a fixed distribution in the convex stochastic order. Also complementary to our extreme point approach, \citet*{bedard2023multivariate} study a problem of maximizing a concave functional over multidimensional monotone functions and develop a computationally tractable method using a notion of multivariate majorization and a sweeping operator \'{a} la  \citet{Rochet1998}. 

Our analysis also connects to the literature on multidimensional screening. \citet{Manelli2007} introduce the extreme point approach to study multidimensional screening. They show that the well-known result of a posted price being optimal for selling one good (\citealt{Myerson1981}; \citealt{riley1983optimal}) is an immediate consequence of the extreme points of one-dimensional monotone functions. They then show that with multiple goods, the optimal mechanism is drastically more complex, as it generally involves many lotteries. The key difference between our setting and theirs is that with multiple dimensions, monotonicity no longer characterizes the implementable allocations.\footnote{Indeed, the implementable allocations are characterized by \textit{\textbf{cyclic monotonicity}} (\citealt{rochet1987necessary}). However, in certain screening problems without monetary transfers, the implementable allocations can be characterized by multidimensional monotonicity; see e.g. \citet{vravosinos2024bidimensional}.} Indeed, the indirect utility functions in multidimensional screening are monotone \textit{convex} functions, which turn out to have drastically different structures and are known to admit a dense set of extreme points (\citealt{johansen1974extremal}; \citealt{bronshtein1978extremal}; \citealt{lahr2024extreme}).\footnote{These results are also the reason why recent studies of multidimensional screening do not adopt an extreme point approach but provide conditions on primitives such that a simple mechanism is optimal (see e.g. \citealt{haghpanah2021pure}; \citealt{yang2022costly}).} In comparison, we relax the convexity constraint and use multidimensional monotone functions to study design problems with multiple agents.\footnote{We focus on Bayesian Nash or ex post Nash equilibria as solution concepts. See \citet*{rudov2025extreme} for conditions under which such equilibria are extreme points of the \textit{correlated} equilibria.} 

Several recent economics papers also build on results from mathematical tomography including \citet{gershkov2013equivalence} who study BIC-DIC equivalence, \citet{chen2019equivalence} who study stochastic-deterministic equivalence, and \citet*{he2024private} who study information structures with independent signals.\footnote{See also \citet{arieli2021feasible} who study feasible posterior distributions when signals need \textit{not} be independent. Their question corresponds to the study of coherent distributions (see \citealt{dawid1995coherent}), for which the characterization of extreme points is an open question (\citealt{burdzy2020bounds}).} As our applications to these settings show, our extreme point characterizations of rationalizable monotone functions lead to a variety of new results.

The remainder of the paper proceeds as follows. \Cref{sec:multidimensional} presents our results on multidimensional monotone functions. \Cref{sec:rationalizable} presents our results on rationalizable monotone functions. \Cref{sec:application} presents our economic applications. \Cref{sec:conclusion} concludes. \Cref{app:proof} provides omitted proofs. \Cref{app:additional} provides additional results.

\section{Multidimensional Monotone Functions}\label{sec:multidimensional}
Endow the set of integrable functions from $[0,1]^n \to [0,1]$ with the $L^1$ norm. Let $\mathcal{F}$ be the set of monotone functions from $[0,1]^n$ to $[0,1]$. That is, a function $f:[0,1]^n \to [0,1]$ is in $\mathcal{F}$ if $f(x) \leq f(y)$ for any $x,y \in [0,1]^n$ such that $x \leq y$.  Clearly, $\mathcal{F} \subseteq L^1([0,1]^n)$ is a convex set. Moreover, by a suitable version of Helly's selection theorem and the dominated convergence theorem, $\F$ is compact.\footnote{See, e.g., \citet{leonov1996total} for a multidimensional version of Helly's selection theorem.} 

We say that a set $A \subseteq [0,1]^n$ is an \emph{\textbf{upward-closed set}} (or simply an \emph{\textbf{up-set}}) if, for any $x \in A$, $y \geq x$ implies $y \in A$. Clearly, an indicator function $\1_A$ is monotone if and only if $A$ is an up-set. Given a totally ordered index set $\mathcal{I}$, a family of up-sets $\{A_i\}_{i \in \mathcal{I}}$ is \emph{\textbf{nested}} if $A_i \subseteq A_{i'}$ for all $i,i' \in \mathcal{I}$ such that $i<i'$. Our results pertain to the properties of the extreme points of sets related to $\F$. To begin with, we first note the following well-known observation.

\begin{lemma}[\citealt{choquet1954theory}, Theorem 40.1]\label{thm:choquet}
$f \in \mathcal{F}$ is an extreme point of $\mathcal{F}$ if and only if $f=\1_A$ for some up-set $A \subseteq [0,1]^n$.   
\end{lemma}

An immediate consequence of \Cref{thm:choquet} is that \emph{any} monotone function $f \in \F$ can be written as a mixture of indicator functions defined by up-sets (by Choquet's integral representation theorem). Nonetheless, since up-sets in $[0,1]^n$ are only partially ordered under set-inclusion when $n \geq 2$, these indicator functions may not be ordered. As a result, mixtures implied by \Cref{thm:choquet} could potentially involve mixing up-sets that are related in complicated ways. However, as \Cref{thm:ordered-up-sets} below shows, there is always a way to represent a monotone function $f \in \F$ as a mixture of indicator functions defined by a family of \emph{nested} up-sets. 

\begin{prop}[Nesting representation]\label{thm:ordered-up-sets}
A function $f:[0, 1]^n \rightarrow [0, 1]$ is monotone if and only if there exist a collection of nested up-sets $\{A_r\}_{r\in[0, 1]}$ and a probability measure $\mu \in \Delta([0, 1])$ such that for all $x \in [0,1]^n$,
\[
f(x) = \int_0^1 \1_{A_r}(x) \mu(\d r) \,.
\]
Moreover, if $\big(\{A_r\}_{r\in[0, 1]}, \mu\big)$ and $\big(\{A'_r\}_{r\in[0, 1]}, \mu'\big)$ represent the same monotone function in the above sense, then they must induce the same distribution over the same family of nested up-sets. 
\end{prop}

Although Choquet's extreme point theorem (\Cref{thm:choquet}) implies that any monotone function $f \in \F$ can be represented as a mixture of indicator functions defined by up-sets, \Cref{thm:ordered-up-sets} further refines the representation and ensures that there must be such a mixture that is supported on a family of nested up-sets. Moreover, \Cref{thm:ordered-up-sets} asserts that such a nesting representation must be unique, even though Choquet's integral representation is in general not unique. This shows that a multidimensional monotone function can be equivalently viewed as a \emph{one}-dimensional probability distribution on a family of nested up-sets. 

An immediate consequence of this single-dimensional representation of a multidimensional monotone function is a characterization of extreme points of multidimensional monotone functions subject to finitely many affine constraints. Fix any finite collection $\{\phi^j\}_{j=1}^m$ of essentially bounded functions on $[0,1]^n$, and a finite collection $\{\eta^j\}_{j=1}^m$ of real numbers. Consider the following subset $\overline{\F}$ of $\F$: 
\[
\overline{\F}:=\left\{f \in \F: \int_{[0,1]^n} f(x) \phi^j(x) \d x \leq \eta^j\,, \forall j \in \{1,\ldots,m\}\right\}\,.
\]
From \Cref{thm:choquet}, it is well known that, by Proposition 2.1 of \citet{Winkler1988}, any extreme point of $\F$ is a mixture of at most $m+1$ indicator functions on up-sets. However, there are many different up-sets in $[0,1]^n$, and mixtures of them, even if there are only finitely many, could be complex. Nonetheless, \Cref{thm:ordered-up-sets} allows us to further sharpen the characterization, as stated by \Cref{thm:nested-up-sets} below.

\begin{theorem}\label{thm:nested-up-sets}
Every extreme point of $\overline{\F}$ is a mixture of at most $m+1$ indicator functions $\{\1_{A_j}\}_{j=1}^{m+1}$ where $\{A_j\}_{j=1}^{m+1}$ are nested up-sets. 
\end{theorem}

\Cref{thm:nested-up-sets} further simplifies the structure of extreme points of $\overline{\F}$ implied by \Cref{thm:choquet} and the result of \citet{Winkler1988}. In addition to being a mixture of at most $m+1$ indicator functions defined on up-sets, \Cref{thm:nested-up-sets} states that these up-sets must be nested.  An immediate consequence is that any extreme point $f$ of $\overline{\F}$ can be written as 
\[
f(x)=\sum_{j=1}^{m+1}p_j\1\{g(x) \in [k_{j-1},k_{j})\}\,,
\]
where $0 \leq p_1 \leq p_2 \leq \cdots \leq p_{m+1}=1$, $0 \leq k_0 \leq k_1 \leq \cdots \leq k_{m+1}=1$, and $g:[0,1]^n \to [0,1]$ is a monotone function. That is, $f$ increases from $0$, through $\{p_1,\ldots,p_{m}\}$, to $1$ as a monotone aggregator $g \in \F$ increases from $0$ to $1$. Note that without the nesting structure identified by \Cref{thm:nested-up-sets}, $m+1$ many mixtures of arbitrary up-sets would in general induce a monotone function that has $2^{m+1}$ many distinct values. Thus, \Cref{thm:nested-up-sets} provides an exponential reduction in the complexity of the extreme points by exploiting the nesting structure from \Cref{thm:ordered-up-sets}. 

Clearly, not all mixtures of $m+1$ indicator functions defined by nested up-sets are extreme points of $\overline{\F}$. Indeed, if $\phi^j \equiv 0$ and $\eta^j=0$ for all $j$, then the only extreme points of $\overline{\F}$ are single indicator functions defined on up-sets, according to \Cref{thm:choquet}. However, if the constraints are tight and have enough independence, then every mixture of at most $m+1$ indicator functions defined on nested up-sets is an extreme point of $\overline{\F}$. Specifically, consider any family $\{\phi^j\}_{j=1}^m$ of essentially bounded functions on $[0,1]^n$ such that for any nonempty disjoint subsets $\{B_j\}_{j=1}^k$, $B_j \subseteq [0,1]^n$, $k \leq m$, the vectors 
\[
\left\{\begin{pmatrix}
\int_{B_1}\phi^1(x) \d x\\
\vdots\\
\int_{B_1}\phi^m(x) \d x
\end{pmatrix},\cdots,\begin{pmatrix}
\int_{B_k}\phi^1(x) \d x\\
\vdots\\
\int_{B_k}\phi^m(x) \d x
\end{pmatrix}\right\}
\]
are linearly independent. Let 
\[
\overline{\F}^\star:=\left\{f \in \F: \int_{[0,1]^n} f(x) \phi^j(x) \d x = \eta^j\,, \forall j \in \{1,\ldots,m\}\right\}\,.
\]
We then have the following converse of \Cref{thm:nested-up-sets}:
\begin{prop}\label{prop:necessity-joint}
$f \in \overline{\F}^\star$ is an extreme point if and only if $f=\sum_{j=1}^k \lambda_j\1_{A_j}$ for some nested up-sets $\{A_j\}_{j=1}^k$, $k \leq m$, and for some $\{\lambda_j\}_{j=1}^k$ such that $\lambda_j \geq 0$ and $\sum_{j=1}^k \lambda_j=1$.  
\end{prop}

\begin{figure}
\centering
\tikzset{
solid node/.style={circle,draw,inner sep=1.25,fill=black},
hollow node/.style={circle,draw,inner sep=1.25}
}
\begin{subfigure}[b]{0.4\linewidth}
\centering
    \includegraphics[scale=0.25]{b.png}
\caption{An Extreme Point}
\label{fig1a}
\end{subfigure}%
\begin{subfigure}[b]{0.4\linewidth}
\centering
 \includegraphics[scale=0.25]{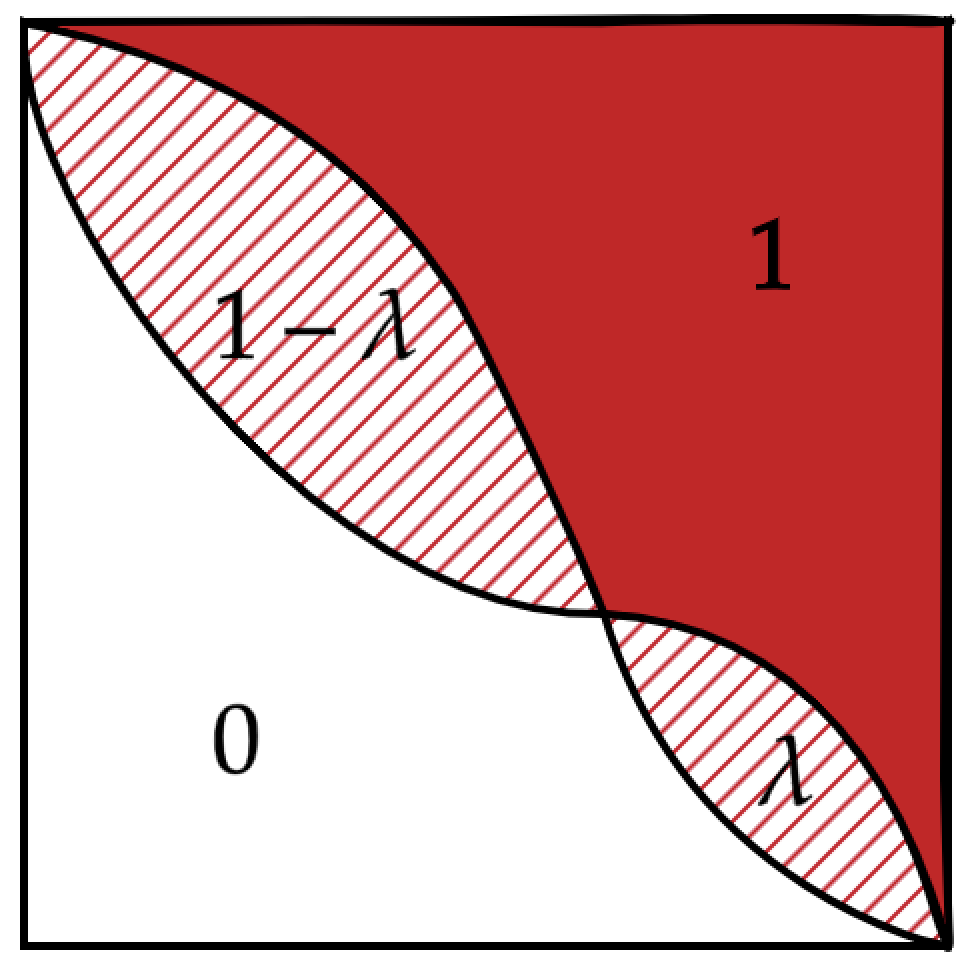}
\caption{Not an Extreme Point}
\label{fig1b}
\end{subfigure}
\caption{Extreme Points of $\overline{\F}$}
\label{fig1}
\end{figure}

\Cref{fig1} illustrates the structure of extreme points of $\overline{\F}$, in the case of $m=1$. \Cref{fig1a} depicts a monotone function $f=\lambda\1_A +(1-\lambda) \1_{A'}$, where $A' \subseteq A$ are nested up-sets. According to \Cref{thm:nested-up-sets}, any extreme point of $\overline{\F}$ must exhibit the structure given by \Cref{fig1a}. \Cref{prop:necessity-joint} further ensures that any such function must be an extreme point of $\overline{\F}^\star$ when the constraints have enough independence. On the other hand, \Cref{fig1b} depicts a monotone function $f=\lambda \1_A+(1-\lambda) \1_{A'}$, where $A$ and $A'$ are up-sets but are not nested. According to \Cref{thm:nested-up-sets}, such functions can never be an extreme point of $\overline{\F}$, even if they are mixtures of two extreme points of $\F$. 

Let $\overline{\F}_{\emph{sym}} \subset \overline{\F}$ be the set of symmetric monotone functions satisfying the linear constraints. A set $A$ is \textit{\textbf{symmetric}} if $\1_{A}$ is a symmetric function. It turns out that \Cref{thm:nested-up-sets} continues to hold when, in addition, the symmetry constraint is imposed:

\begin{prop}\label{prop:sym-nested-up-sets}
Every extreme point of $\overline{\F}_{\emph{sym}}$ is a mixture of at most $m+1$ indicator functions $\{\1_{A_j}\}_{j=1}^{m+1}$ where $\{A_j\}_{j=1}^{m+1}$ are nested symmetric up-sets. 
\end{prop}

\section{Rationalizable Monotone Functions}\label{sec:rationalizable}
Next, we characterize extreme points of monotone functions that can be ``rationalized'' by some (not necessarily monotone) function $f:[0,1]^n \to [0,1]$. For any tuple $q=(q_1,\ldots,q_n)$ of nondecreasing, left-continuous functions from $[0,1]$ to $[0,1]$, we say that $q$ is \textit{\textbf{rationalizable}} if there exists a function $f:[0,1]^n \to [0,1]$ (not necessarily monotone) such that 
\[
q_i(x_i)=\int_{[0,1]^{n-1}}f(x)\d x_{-i}\,,
\]
for all $i$ and all $x_i \in [0,1]$. For any such $f$ and $q$, we also say that $q$ is \emph{\textbf{rationalized}} by $f$. Let $\cQ$ be the set of all rationalizable tuples of nondecreasing left-continuous functions. Fix a family $\{\psi_i^j\}_{1\leq i\leq n}^{1\leq j \leq m}$ of essentially bounded functions on $[0,1]$ and a family $\{\eta^j\}_{j=1}^m$ of real numbers, let $\overline{\cQ}$ be the set of all tuples of rationalizable monotone functions subject to $m$ affine constraints: 
\[
\overline{\cQ}:=\left\{q \in \cQ: \sum_{i=1}^n \int_0^1 q_i(x_i)\psi^j_i(x_i)\d x_i \leq \eta^j\,, \forall j \in \{1,\ldots,m\}\right\}\,.
\]

While we do not impose monotonicity on the joint function $f$, an elegant result of \citet{gutmann1991existence} shows that any rationalizable monotone tuple $q$ can be rationalized by a multidimensional monotone $f \in \F$:
\begin{lemma}[\citealt{gutmann1991existence}, Theorem 5; \citealt{gershkov2013equivalence}, Theorem 1]\label{lem:gutmann-monotone}
A tuple of one-dimensional monotone functions $q=(q_1,\dots, q_n)$ is rationalizable if and only if it can be rationalized by a multidimensional monotone function $f \in \F$. 
\end{lemma}

As a consequence of \Cref{lem:gutmann-monotone}, the set $\cQ$ is exactly a linear projection of the set $\F$. \Cref{lem:extreme-projection} below connects the extreme points of a convex set to the extreme points of its linear projection.
\begin{lemma}[Affine mapping lemma]\label{lem:extreme-projection}
Let $X$ be a compact convex subset of a locally convex topological vector space, and let $Y$ be a topological vector space. For any continuous affine map $L:X \to Y$, we have
\[
\mathrm{ext}(L(X)) \subseteq L(\mathrm{ext}(X)) \,.
\]
\end{lemma}

Now, using \Cref{lem:gutmann-monotone} and \Cref{lem:extreme-projection}, we exploit \Cref{thm:choquet} and \Cref{thm:nested-up-sets} to characterize the extreme points of $\cQ$ and $\overline{\cQ}$:  

\begin{theorem}\label{thm:extreme-are-up-sets}
\mbox{}
\begin{compactenum}[(i)]
\item Every extreme point of $\cQ$ is rationalized by $\1_{A}$ for some up-set $A \subseteq [0,1]^n$.  
\item Every extreme point of $\overline{\cQ}$ is rationalized by a mixture of $\{\1_{A_j}\}_{j=1}^{m+1}$ for some nested up-sets $\{A_j\}_{j=1}^{m+1}$.
\end{compactenum}
\end{theorem}

One may also impose symmetry on $\mathcal{Q}$ and $\overline{\cQ}$, since \Cref{lem:gutmann-monotone} holds with the symmetry restriction (\citealt{gershkov2013equivalence}), and \Cref{thm:nested-up-sets} holds under symmetry (\Cref{prop:sym-nested-up-sets}), \Cref{thm:extreme-are-up-sets} continues to hold on $\cQ_{\text{sym}}$ and $\overline{\cQ}_{\text{sym}}$ with  nested symmetric up-sets. 

While \Cref{thm:extreme-are-up-sets} provides necessary conditions for the extreme points of $\cQ$ and $\overline{\cQ}$, these conditions may not be sufficient in general. There might be (mixtures of) indicator functions defined on (nested) up-sets that are not extreme points of $\cQ$ and $\overline{\cQ}$. Nonetheless, when $n=2$, these necessary conditions can be greatly sharpened. Below, we explore the cases where $n=2$ and $m \in \{0, 1\}$. 

\subsection{Rationalizable Monotone Pairs}
Our next set of results characterizes the extreme points of rationalizable functions in the case of $n=2$ with and without a linear constraint. Namely, $\cQ$ is the set of pairs $q=(q_1,q_2)$ of rationalizable monotone functions on $[0,1]$; and 
\[
\overline{\cQ}=\left\{q \in \cQ: \int_0^1 q_1(x_1)\psi_1(x_1)\d x_1+\int_0^1 q_2(x_2) \psi_2(x_2) \d x_2 \leq \eta\right\}\,,
\]
for some essentially bounded functions $\psi_1$ and $\psi_2$ and some real number $\eta$. 

When $n=2$, part $(i)$ of \Cref{thm:extreme-are-up-sets} becomes both necessary and sufficient. The sufficient condition also exhibits a notion of uniqueness,  as \Cref{thm:rationalized-upsets} below shows. To state the theorem, we say that $q \in \cQ$ is \textit{\textbf{uniquely rationalized}} by a function $f:[0,1]^2 \to [0,1]$ if for any function $\tilde{f}:[0,1]^2 \to [0,1]$ that also rationalizes $q$, it must be that $f\equiv \tilde{f}$ (up to measure zero).  

\begin{theorem}\label{thm:rationalized-upsets}
$q=(q_1,q_2)$ is an extreme point of $\cQ$ if and only if $q$ is rationalized by $\1_A$ for some up-set $A \subseteq [0,1]^2$. Moreover, every extreme point of $\cQ$ is uniquely rationalized. 
\end{theorem}

For a general $n \in \N$, \Cref{thm:extreme-are-up-sets} states that any extreme point $q$ of $\cQ$ can be rationalized by an indicator function $\1_A$ defined on an up-set. It does not imply that any $q \in \cQ$ rationalized by indicator functions defined on an up-set must be an extreme point. However, when $n=2$, being rationalized by an indicator function defined on an up-set is both necessary and sufficient for $q$ to be an extreme point of $\cQ$, according to \Cref{thm:rationalized-upsets}. Perhaps surprisingly, the projection map that maps $\mathcal{F}$ to $\mathcal{Q}$ actually does not ``destroy'' extreme points in dimension $2$. 
Furthermore, any extreme point must be \textit{uniquely} rationalized by such a function---thus, any $q \in \cQ$ rationalized by some function $f$ that is \textit{not} an indicator function $\1_A$ defined on an up-set \textit{cannot} be an extreme point.

The proof of \Cref{thm:rationalized-upsets} proceeds as follows. It actually shows that if $q$ is rationalized by an indicator function defined on an up-set, then $q$ must be an \textit{\textbf{exposed point}} (the unique maximizer for some linear functional). Since every exposed point is an extreme point, the sufficiency part of \Cref{thm:rationalized-upsets} follows. The argument exploits the fact that when $n=2$, an up-set is an \textit{\textbf{additive set}}, in the sense of \citet{fishburn1990sets}.
Using this property, we can define an additively separable function $\phi$ on $[0, 1]^2$ such that $\1_{A}$ is the unique maximizer for the linear functional $f \mapsto \int \phi(x) f(x) \d x$, which in turn implies that the one-dimensional marginals of $\1_A$ must be an exposed point of $\mathcal{Q}$, by the additive separability of $\phi$. Moreover, for the unique rationalizability part of \Cref{thm:rationalized-upsets}, we observe that $\1_{A}$ is actually the unique maximizer of $f \mapsto \int \phi(x) f(x) \d x$ among \textit{all} functions (not necessarily monotone). Thus, any function $f$ sharing the same marginals with $\1_A$ must also maximize $f \mapsto \int \phi(x) f(x) \d x$, and hence must be identical to $\1_A$.

\begin{rmk}\label{rmk:exposed}
In fact, as a consequence of the proof, we have also shown that every extreme point in $\mathcal{Q}$ is exposed when $n =2$, and that every extreme point in $\mathcal{F}$ is exposed for any $n$. The proof can be readily extended to show that, in any dimension $n$, every exposed point of $\mathcal{Q}$ must be the projection of some additive set, and the projection of ``almost'' every additive set is an exposed point of $\mathcal{Q}$ (see \Cref{thm:exposedQ} in the appendix for the formal statement and proof). However, since it is known that up-sets do not coincide with additive sets when $n > 2$, there would still be a gap between projections of up-sets and extreme points of $\cQ$ when $n > 2$. Regarding the unique rationalizability of our extreme points in dimension $n = 2$, an alternative proof would be to exploit the fact that every additive set is a ``strong set of uniqueness'' (\citealt{lorenz1949}; \citealt{fishburn1990sets}; \citealt{kemperman1991sets}). However, that argument would not generalize to prove the unique rationalizability of the extreme points of $\overline{\mathcal{Q}}$ in \Cref{thm:rectangle}.
\end{rmk}

We further characterize the extreme points of $\overline{\cQ}$ in the case of $n=2$ and $m=1$. To this end, recall that for two nondecreasing functions $g_1:[0,1] \to \R$ and $g_2:[0,1] \to \R$, $g_2$ is said to \emph{\textbf{weakly majorize}} $g_1$ (denoted by $g_1 \preceq_w g_2)$ if 
\[
\int_x^1 g_1(z) \d z \leq \int_x^1 g_2(z)\d z 
\]
for all $x \in [0,1]$. Moreover, $g_2$ is said to \emph{\textbf{majorize}} $g_1$ (denoted by $g_1 \preceq g_2$) if, in addition to the above inequalities, the equality holds at $x=0$. Note that if $g_1$ and $g_2$ are CDFs, then $g_2$ majorizes $g_1$ if and only if $g_1$ is a mean-preserving spread of $g_2$, whereas $g_2$ weakly majorizes $g_1$ if and only if $g_2$ second-order stochastic dominates $g_1$. 

Another result of \citet{gutmann1991existence} shows that rationalizable monotone functions can be characterized by a joint majorization relation: 

\begin{lemma}[\citealt{gutmann1991existence}, Theorem 4]\label{lem:gutmann}
For any pair $q=(q_1,q_2)$ of nondecreasing functions, $q \in \cQ$ if and only if $q_1 \preceq \hat{q}_2$, where $\hat{q}_2(z):=1-q_2^{-1}(1-z)$ is the conjugate of $q_2$.\footnote{For any nondecreasing left-continuous function $g:[0,1] \to [0,1]$, $g^{-1}$ is defined as $g^{-1}(z):=\inf\{x \in [0,1]:g(x) > z\}$. Note that $g^{-1}$ is right-continuous, and hence $\hat{g}$ is left-continuous.} 
\end{lemma}

We say that a subset $B \subseteq [0,1]^2$ is a \emph{\textbf{rectangle}} if $\mathrm{cl}(B)=[\underline{x}_1,\overline{x}_1] \times [\underline{x}_2,\overline{x}_2]$, where $\underline{x}_1 \leq \overline{x}_1$ and $\underline{x}_2 \leq \overline{x}_2$. Two nested up-sets $A' \subseteq A$ are said to \emph{\textbf{differ by}} $B \subseteq [0,1]^2$ if $A \backslash A'=B$.

\begin{theorem}\label{thm:rectangle}
Every extreme point of $\overline{\cQ}$ is rationalized by a mixture of $\1_A$ and $\1_{A'}$, where $A' \subseteq A \subseteq [0,1]^2$ are nested up-sets that differ by at most a rectangle. Moreover, every extreme point of $\overline{\cQ}$ is uniquely rationalized among all monotone functions. 
\end{theorem}

Notably, \Cref{thm:rectangle} greatly reduces the set of possible extreme points of $\overline{\cQ}$ when $n=2$. Indeed, part $(ii)$ of \Cref{thm:extreme-are-up-sets} indicates that, even when $m=1$, an extreme point of $\overline{\cQ}$ is a mixture of two indicator functions defined on nested up-sets.  When $n=2$, \Cref{thm:rectangle} states that these two up-sets can only differ by a single rectangle, which sharpens the characterization of extreme points of $\overline{\cQ}$.

\begin{figure}[t]   \hspace{0.8cm}\includegraphics[scale=0.5]{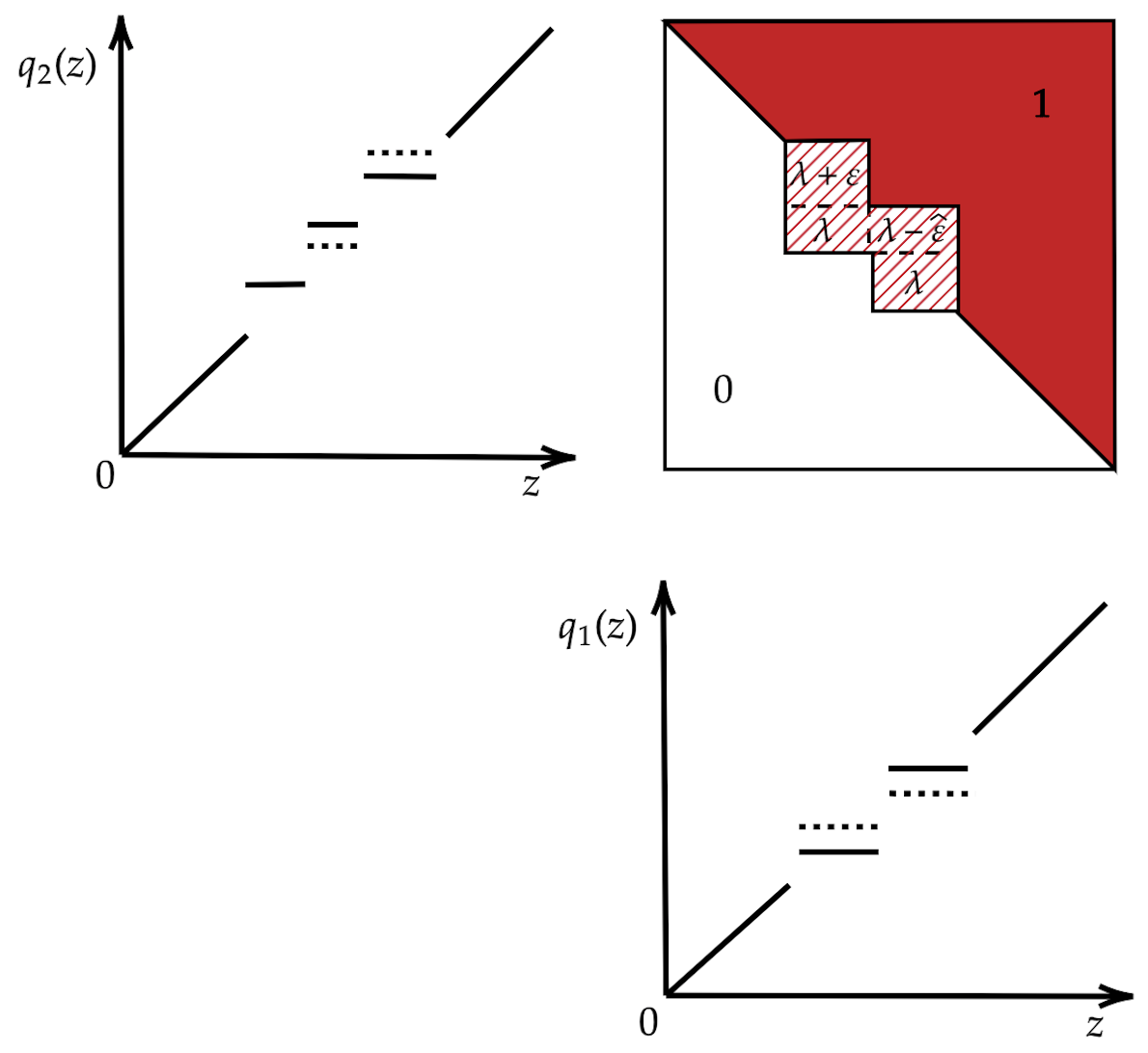}
    \caption{Perturbation of Non-Rectangular Regions, with Original Marginals (Solid) and Perturbed Marginals (Dashed)}
    \label{fig:proof}
\end{figure}

The proof of \Cref{thm:rectangle} can be found in the Appendix. The argument is involved, so we sketch the intuition here. We first sketch the intuition behind the first part of \Cref{thm:rectangle}. By part $(ii)$ of \Cref{thm:extreme-are-up-sets}, we know that any extreme point $q \in \overline{\cQ}$ must be rationalized by a mixture of indicator functions defined on two nested up-sets $A_1 \subseteq A_2$. Suppose for contradiction that the two nested up-sets differ not by a rectangle but by a ``flipped Z'' as in the dashed region of \Cref{fig:proof}.  A key observation is that by choosing the correct ratio of (small enough) $\varepsilon, \hat{\varepsilon} >0$ and perturbing the mixture in the dashed region using the way illustrated by \Cref{fig:proof}, there exist two distinct monotone $q',q'' \in \overline{\cQ}$ such that $\nicefrac{1}{2} q'+\nicefrac{1}{2} q''=q$, as illustrated by the marginals in \Cref{fig:proof}. This contradicts $q$ being an extreme point of $\overline{\cQ}$.

The actual proof shows that for two arbitrary nested up-sets, if they differ by \textit{any} non-rectangular region, then there must exist such a perturbation. It consists of three steps. First, we observe that for any extreme point $q \in \overline{\cQ}$, $q_1$ must be an extreme point of a convex set of nondecreasing functions majorized by $\hat{q}_2$. Combining with Theorem 1 of \citet{nikzad2023constrained}, it follows that there exist countably many disjoint intervals $\{(\underline{z}_k,\overline{z}_k]\}_{k=1}^\infty$ such that $q_1 \equiv \hat{q}_2$ on $[0,1] \backslash \cup_{k=1}^\infty (\underline{z}_k,\overline{z}_k]$, and equals a step function with at most two steps on each $(\underline{z}_k,\overline{z}_k]$. We then switch the role of $q_1$ and $q_2$ and apply the symmetric argument to $q_2$ fixing $\hat{q}_1$. 
Second, we use these properties and the fact that $(q_1, q_2)$ must be rationalizable by a mixture of two nested up-sets according to \Cref{thm:extreme-are-up-sets} to derive a set of necessary conditions that $(q_1, q_2)$ must satisfy. In particular, we show that on each interval $(\underline{z}_k,\overline{z}_k]$, $q_1$ must be a constant and $\hat{q}_2$ must be a step function with only one jump (symmetrically, $q_2$ and $\hat{q}_1$ have the same feature). Moreover, we show that the intervals on which $q_1$ and $q_2$ are constants form a countable collection of disjoint ``off-diagonal'' rectangles on $[0,1]^2$. Third, for any such $(q_1, q_2)$, we explicitly construct a monotone $f$ that rationalizes it, and show that the monotone function $f$ allows for a perturbation similar to the one we construct in \Cref{fig:proof}, unless it coincides with a mixture of two nested up-sets differing by at most a rectangle as depicted in \Cref{fig2a}.

For the unique rationalizability in \Cref{thm:rectangle}, we first show that any extreme 
point $(q_1, q_2)$ must be uniquely rationalized by the mixture of two nested up-sets that differ by a rectangle among all mixtures of two nested up-sets, using a similar argument as before, by constructing an additively separable function $\phi(x_1, x_2)$ such that our rationalization using a rectangle is the unique maximizer for the linear functional $\Phi: f \mapsto \int \phi f \d x$ among the mixtures of two nested up-sets fixing the marginals. Now, for any monotone $f$ that rationalizes $(q_1, q_2)$, we can apply \Cref{thm:nested-up-sets} and Choquet's integral representation theorem to represent $f$ as an integral over the mixtures of two nested up-sets. Projecting that representation into marginals also gives an integral representation of $(q_1, q_2)$. But since $(q_1, q_2)$ is an extreme point, and hence has a unique integral representation, all such mixtures of two nested up-sets must have the same marginals, and hence they all rationalize $(q_1, q_2)$. But then all such mixtures, and hence $f$, must be identical to our construction using the rectangle.  

Similar to \Cref{thm:nested-up-sets}, the conditions identified by \Cref{thm:rectangle} are also sufficient when the affine constraints have enough independence. Specifically, consider any pair $\psi_1,\psi_2$ of essentially bounded functions such that for any rectangle $D \subseteq [0,1]^2$ with a positive measure,
\[
\int_{D}(\psi_1(x_1)+\psi_2(x_2))\d x \neq 0\,.
\]
Fix any $\eta \in \R$. Let
\[
\overline{\cQ}^\star:= \left\{q \in \cQ: \int_0^1 q_1(x_1)\psi_1(x_1)\d x_1+\int_0^1 q_2(x_2) \psi_2(x_2) \d x_2 = \eta\right\}\,.
\]
We have the following characterization:

\begin{prop}\label{prop:marginal-sufficiency}
$q=(q_1,q_2) \in \overline{\cQ}^\star$ is an extreme point if and only if $q$ is rationalized by a mixture of $\1_A$ and $\1_{A'}$, where $A' \subseteq A \subseteq [0,1]^2$ are nested up-sets that differ by at most a rectangle. Moreover, every extreme point of $\overline{\cQ}^\star$ is uniquely rationalized among all monotone functions. 
\end{prop}

As a consequence of \Cref{prop:marginal-sufficiency}, since exposed points are dense in extreme points by Straszewicz's theorem (\citealt{Kleejr1958}), the set of \textit{exposed} mixtures of two nested up-sets differing by a non-trivial rectangle must be dense in the set of all such mixtures.\footnote{To see this, note that any such mixture, since it is an extreme point of $\overline{\cQ}^\star$, must be approximated by a sequence of exposed points of $\overline{\cQ}^\star$, but this sequence must contain a subsequence of mixtures of nested up-sets differing by a non-trivial rectangle since otherwise the majorization constraint would be binding everywhere in the limit.} Thus, for linear optimization problems, \Cref{prop:marginal-sufficiency} shows that the rectangular difference between the two nested up-sets cannot be ignored in general. 

\begin{figure}[t]
\centering
\begin{subfigure}[b]{0.4\linewidth}
\centering
    \includegraphics[scale=0.25]{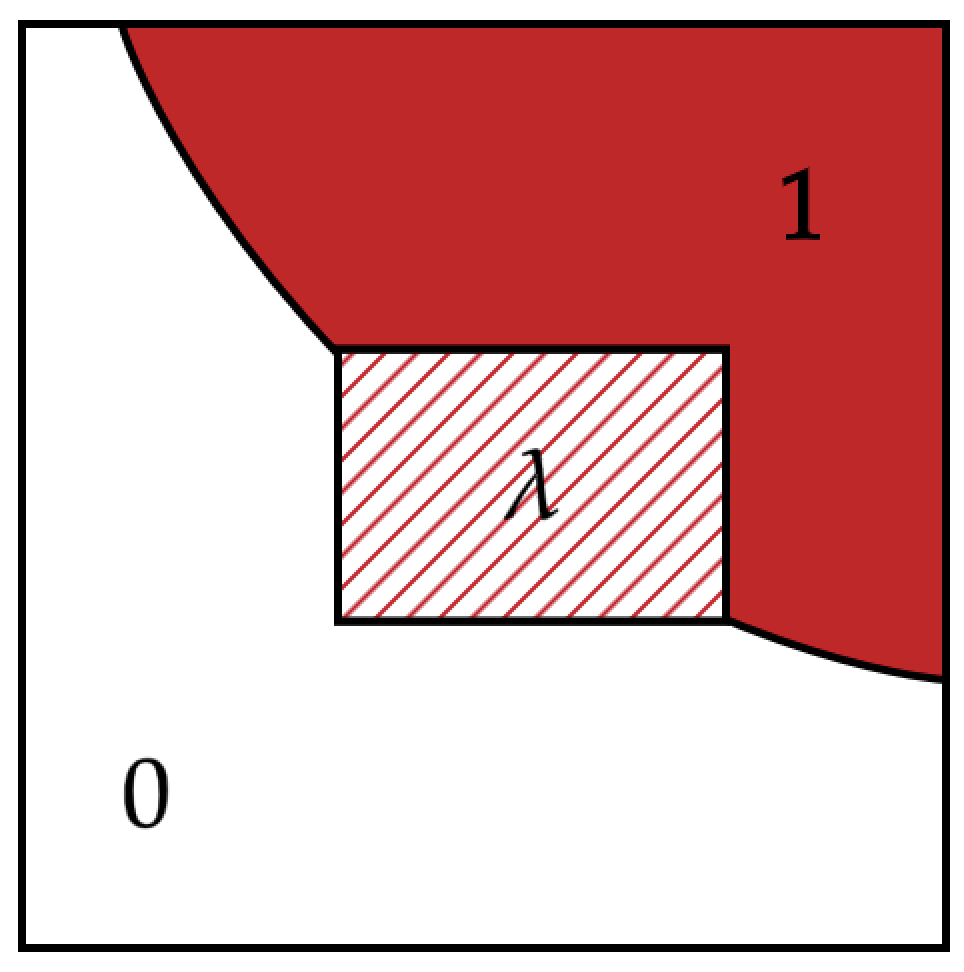}
\caption{An Extreme Point}
\label{fig2a}
\end{subfigure}%
\begin{subfigure}[b]{0.4\linewidth}
\centering
 \includegraphics[scale=0.25]{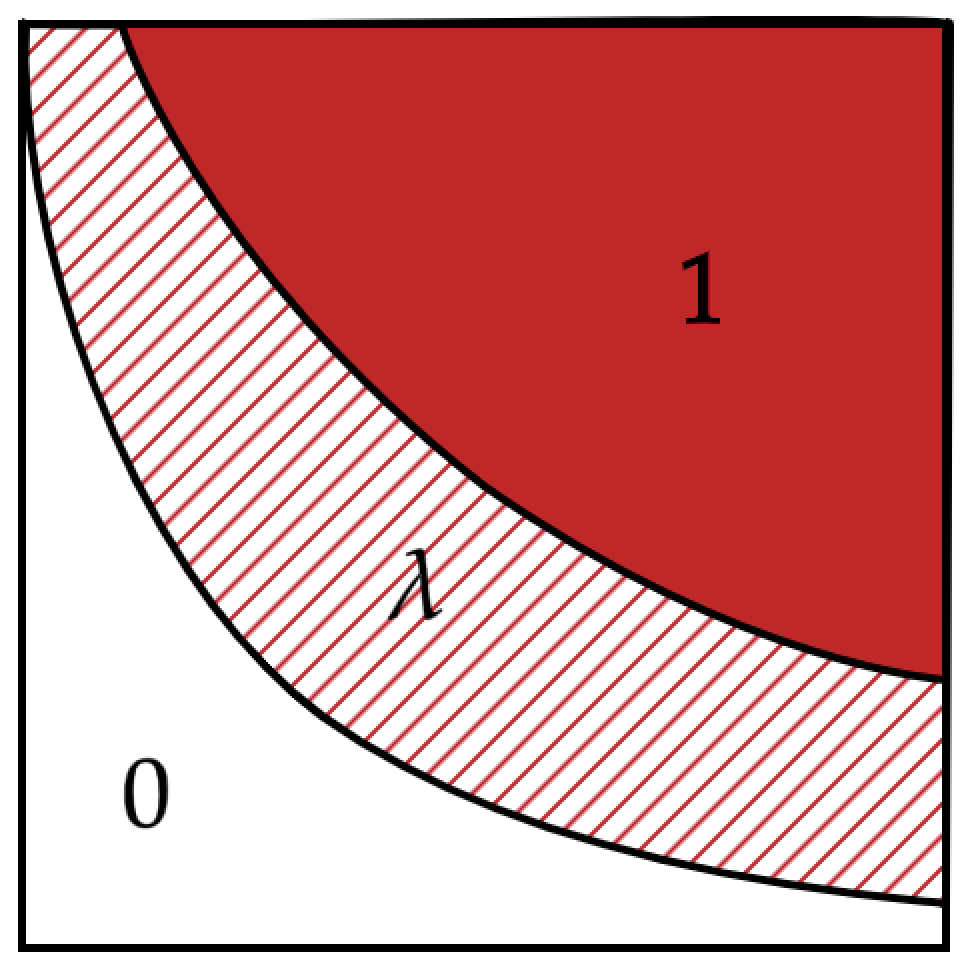}
\caption{Not an Extreme Point}
\label{fig2b}
\end{subfigure}
\caption{Extreme Points of $\overline{\cQ}$}
\label{fig2}
\end{figure}
\Cref{fig2} illustrates the extreme points of $\overline{\cQ}$. \Cref{fig2a} depicts a monotone function $f=\lambda \1_A+(1-\lambda) \1_{A'}$, where $A' \subseteq A$ are nested up-sets that differ by a rectangle. According to \Cref{thm:rectangle}, any extreme point of $\overline{\cQ}$ must be rationalized by some $f \in \F$ of this shape. Moreover, \Cref{prop:marginal-sufficiency} states that any $q \in \overline{\cQ}^\star$ rationalized by such a monotone function must be an extreme point of $\overline{\cQ}^\star$. Meanwhile, \Cref{fig2b} depicts a monotone function $f=\lambda \1_A+(1-\lambda) \1_{A'}$, where $A' \subseteq A$ are nested up-sets but do not differ by a rectangle. According to \Cref{thm:rectangle}, because any extreme point of $\overline{\cQ}$ is uniquely rationalized among monotone functions, this function can never rationalize an extreme point of $\overline{\cQ}$, even though it might be an extreme point of $\overline{\F}$.

\begin{rmk}[Projections under Different Measures]\label{rmk:projection}
While we use the Lebesgue measure when defining rationalizability, all of our results on rationalizable monotone functions can be readily extended to settings where rationalizability is defined with respect to integrals under any \textit{\textbf{product measures}}. See \Cref{subsec:measure} for details. 
\end{rmk}

\subsection{Extreme Points of Joint Majorization Sets} 
\Cref{thm:rationalized-upsets} and \Cref{thm:rectangle}, together with the connection between rationalizability and majorization (\Cref{lem:gutmann}), provide a characterization of extreme points of the set of pairs of monotone functions $q=(q_1,q_2)$ where $q_1$ is majorized by the conjugate of $q_2$. 

\begin{prop}\label{prop:reverse-majorization}
$q$ is an extreme point of the convex set 
\[
\Big\{q\,:\, q_1 \preceq \hat{q}_2\,;\, q_1,q_2 \mbox{ are nondecreasing, left-continuous} \Big\}
\]
if and only if $q_1= \hat{q}_2$. 
\end{prop}

The structure of the extreme points characterized by \Cref{prop:square-majorization} differs considerably from the ones characterized by \citet*{kleiner2021extreme}, who consider the set of nondecreasing functions that majorize, or are majorized by, a \textit{fixed} nondecreasing function. In particular, when the majorization upper and lower bounds are endogenous and linked via the conjugate/inverse relation, the extreme points must be the ones where the majorization constraints bind everywhere, rather than those given by a combination of binding monotonicity and majorization constraints.

\begin{prop}\label{prop:square-majorization}
For any pair of essentially bounded functions $\{\psi_1,\psi_2\}$ on $[0,1]$ and any $\eta \in \R$, if $q$ is an extreme point of the convex set 
\begin{align*}
\bigg\{q\,:&\, q_1 \preceq \hat{q}_2\,;\, q_1,q_2 \mbox{ are nondecreasing, left-continuous}\,;\,\\ &\int_{0}^1 q_1(x_1)\psi_1(x_1)\d x_1+\int_0^1 q_2(x_2)\psi_2(x_2)\d x_2 \leq \eta \bigg\}\,,
\end{align*}
then there exist $\underline{z}\leq \overline{z}$, $ q_1(\underline{z})\leq \underline{\gamma} \leq \overline{\gamma} \leq  q_1(\overline{z}^{+})$, and $\lambda \in (0,1)$, such that $q_1(z)=\hat{q}_2(z)$ for all $z \in [0, 1]\backslash (\underline{z},\overline{z}]$ and otherwise 
\[
q_1(z)=\lambda \overline{\gamma}+(1-\lambda)\underline{\gamma}\,,\quad \mbox{ and } \hat{q}_2(z)=\begin{cases}
\underline{\gamma},&\mbox{if } z \in (\underline{z},\,\,\,(1-\lambda)\overline{z}+\lambda\underline{z}]\\
\overline{\gamma},&\mbox{if } z \in ((1-\lambda)\overline{z}+\lambda\underline{z},\,\,\,\overline{z}]
\end{cases}\,.
\]
\end{prop}

Similar to \Cref{prop:reverse-majorization}, the structure of the extreme points given by \Cref{prop:square-majorization} is considerably different from those identified by \citet{nikzad2023constrained} and \citet{candogan2023disclosure}, who consider the set of nondecreasing functions that majorize, or are majorized by, a \textit{fixed} nondecreasing function, subject to finitely many linear constraints. In particular, rather than possibly having multiple intervals on which the monotonicity constraints bind, the extreme points given by \Cref{prop:square-majorization} have only a single interval on which the monotonicity constraints bind, whereas the majorization constraints bind everywhere else.   

From \Cref{prop:reverse-majorization}, another immediate consequence is a necessary condition for the extreme points of the weak joint majorization set where $q_1$ is weakly majorized by the conjugate of $q_2$. 

\begin{prop}\label{prop:weak} 
Every extreme point  of the convex set 
\[
\Big \{q \,:\,  q_1 \preceq_w  \hat{q}_2\,;\, q_1,q_2 \mbox{ are nondecreasing, left-continuous} \Big\} 
\]
satisfies $q_1 = \hat{q}_2\1_{[k,1]}$ for some $k \in [0,1]$.
\end{prop}

\section{Economic Applications}\label{sec:application}

In this section, we apply our abstract results in \Cref{sec:multidimensional} and \Cref{sec:rationalizable} to a series of economic applications. 
\Cref{subsec:public} studies public good provision; \Cref{subsec:bilateral} studies bilateral trade; \Cref{subsec:anti-equivalence} studies the limits of mechanism equivalence; \Cref{subsec:auction} studies reduced form auctions; \Cref{subsec:private} studies private private information.   

\subsection{Public Good with Interdependent Values and Correlated Types}\label{subsec:public}

Consider a classic public good provision problem. There are $n$ agents. A designer wants to decide whether to implement a public project. 

We allow for interdependent values and correlated types. Each agent $i$ has a willingness to pay $v_i(s_i, s_{-i})$ for the project being implemented, where $s_i \in S_i \subset \R$ is the private signal received by agent $i$, and $s_{-i}$ are the private signals received by the other agents. We assume that $S_i$ is a compact interval. We allow the joint distribution of signals $s:= (s_1,\dots, s_n)$ to be correlated and only assume that it admits a positive continuous density $g$ on $S_1 \times \cdots \times S_n$. We assume that the agent $i$'s value function $v_i(s_i, s_{-i})$ is bounded, continuously differentiable, and increasing in the agent $i$'s signal $s_i$ for almost all $s_{-i}$. Agents have quasi-linear preferences in transfers. Thus, given a probability $\alpha$ for implementing the project, and transfer $t_i$, agent $i$'s ex-post payoff is 
\[\alpha \cdot v_i(s_i, s_{-i}) - t_i\,.\]

The designer wants to maximize the expected surplus subject to ex-ante budget balance. In particular, implementing the project would cost the designer $c \in \R$, and hence the designer wants to raise enough money to cover the cost in expectation: 
\[\E\Big[\sum_i t_i(s)\Big] \geq \E\Big[c \cdot \alpha (s)\Big]\,. \tag{Budget Balance}\]
The designer can design any mechanism subject to ex-post incentive compatibility (ex-post IC) and ex-post individual rationality (ex-post IR) constraints. By the revelation principle, it is without loss of generality to focus on \textit{\textbf{(direct, incentive-compatible) mechanisms}} $(\alpha, t): S_1 \times \cdots \times S_n \rightarrow [0, 1]\times \R^n$ that satisfy: for all $i$
\begin{align*}
\alpha(s) v_i(s) - t_i(s) &\geq \alpha(\hat{s}_i, s_{-i}) v_i(s_i, s_{-i}) - t_i(\hat{s}_i, s_{-i}) &&\text{ for all $s_i, \hat{s}_i$ and $s_{-i}$\,;} \tag{ex-post IC}\\
\alpha(s) v_i(s) - t_i(s) &\geq 0 &&\text{ for all $s$\,. \tag{ex-post IR}}
\end{align*}
Therefore, the designer's problem is given by 
\[\max_{(\alpha, t)} \E\Bigg[ \alpha(s) \Big(\sum_i v_i(s) - c \Big)\Bigg]\]
subject to the ex-ante budget balance, and ex-post IC and IR constraints. By the logic of \citet{cremer1985}, it is well-known that if we relax the ex-post incentive constraints to be interim incentive constraints, then the first best is generically implementable via full surplus extraction.\footnote{In fact, \citet{cremer1985} show that under stronger conditions, full surplus extraction is possible even with ex-post IC and interim IR mechanisms.} However, such mechanisms are generally fragile and depend on the exact knowledge of the designer (\citealt{brooks2023robust}). 

By the standard argument, ex-post IC implies that the allocation rule $\alpha(s)$ must be a multidimensional monotone function, and transfers can be written as a linear function of the allocation rule $\alpha$, up to a constant. However, unlike standard Myersonian problems where the ironing approach can be applied, the designer needs to choose a \emph{multidimensional} monotone function $\alpha:[0,1]^n \to [0,1]$ to maximize the expected surplus subject to the budget balance constraint, and thus it is not immediately clear how the optimal mechanism can be characterized.  

Nonetheless, a direct application of our main results shows that with ex-post IC and IR constraints, the optimal mechanism in fact has a simple structure. We say that a mechanism $(\alpha, t)$ is a \textit{\textbf{two-threshold policy}} if there exist a monotone \textit{\textbf{aggregation function}} $\phi:[0, 1]^n \rightarrow [0, 1]$ and three constants $p, k_L, k_H \in [0, 1]$ where $k_L \leq k_H$ such that 
\[\alpha(s) = \begin{cases} 1 & \text{ if $\phi(s) \geq k_H$}\;, \\
p & \text{ if $k_L \leq \phi(s) < k_H$}\;, \\
0    & \text{ otherwise}\,.
\end{cases}\]
Such a mechanism can be implemented by aggregating the signals of each agent into a score $\phi(s)$ and comparing the score to two thresholds $k_L$ and $k_H$: if the score is higher than the high threshold, then the project is implemented for sure; if it is lower than the high threshold but higher than the low threshold, the project is implemented with some probability $p$; otherwise, the project is abandoned.

We say that the agents are \textit{\textbf{symmetric}} if for any permutation $\pi$ of $\{1,\dots, n\}$ and any agent $i$, we have $v_i(s_{\pi(1)}, \dots, s_{\pi(n)}) = v_{\pi(i)}(s_1, \dots, s_n)$.\footnote{This condition means that the agents' ``names'' $i \in \{1,\ldots,n\}$ are irrelevant in terms of their value functions, and is the same as the symmetry condition for mechanisms in the reduced-form implementation literature (see, e.g., \citealt{hart2015implementation}).}

\begin{prop}\label{prop:public}
There exists an optimal mechanism that is a two-threshold policy. Moreover, if the agents are symmetric and the signals follow an exchangeable distribution, then there exists a symmetric optimal two-threshold policy. 
\end{prop}
\begin{proof}[Proof of \Cref{prop:public}]
For any agent $i$, since $v_i(s_i, s_{-i})$ is increasing in $s_i$, by a standard argument, we know that $\alpha(s_i, s_{-i})$ is ex-post IC implementable for agent $i$ if and only if $\alpha(\,\cdot\,, s_{-i})$ is nondecreasing in $s_i$. Moreover, it is without loss of generality to normalize $S_i = [0, 1]$. Therefore, $\alpha: [0, 1]^n \rightarrow [0, 1]$ is ex-post IC implementable if and only if it is a multidimensional monotone function. Moreover, by ex-post IC and the Envelope theorem, we also have that the indirect utility functions $U_i(s) := \alpha(s) v_i(s) - t_i(s)$ must satisfy: 
\[U_i(s) = U_i(0, s_{-i}) + \int^{s_i}_0 \alpha(z, s_{-i}) \frac{\partial}{\partial s_i} v_i(z, s_{-i}) \d z\,. \]
Thus, 
\[t_i(s) = \alpha(s) v_i(s)  - \int^{s_i}_0 \alpha(z, s_{-i}) \frac{\partial}{\partial s_i} v_i(z, s_{-i}) \d z - U_i(0, s_{-i})\,.\]
By ex-post IR, we have $U_i(0, s_{-i})\geq 0$. We may assume without loss of generality that $U_i(0, s_{-i}) = 0$, since we can always define $\tilde{t}_i(s) = t_i (s) + U_i(0, s_{-i})$ and the resulting mechanism $(\alpha, \tilde{t})$ would continue to satisfy budget balance, ex-post IR and IC. Therefore, the designer's problem is equivalent to 
\[\max_{\alpha \in \overline{\mathcal{F}}} \E\Bigg[ \alpha(s) \Big(\sum_i v_i(s) - c \Big)\Bigg]\]
where $\overline{\mathcal{F}}$ is the set of monotone functions $\alpha: [0, 1]^n \rightarrow[0, 1]$ satisfying the constraint
\[
\mathbb{E}\Bigg[\sum_i \left(\alpha(s) v_i(s)  - \int^{s_i}_0 \alpha(z, s_{-i}) \frac{\partial}{\partial s_i} v_i(z, s_{-i}) \d z\right) - c \cdot \alpha (s)\Bigg]\geq 0 \,.\]
By Bauer's maximum principle, we immediately know that the above problem admits a solution that is an extreme point of $\overline{\mathcal{F}}$. Let $\alpha^\star$ be such a solution. By \Cref{thm:nested-up-sets}, we know that $\alpha^\star = \1_{A_1} + p \1_{A_2\backslash A_1}$ for two nested up-sets $A_1 \subseteq A_2$. Clearly, this mechanism can be represented via a two-threshold policy (e.g. set $\phi = \alpha^\star$, $k_L = p$, and $k_H = 1$). 

If, in addition, the agents are symmetric and $(s_1,\ldots,s_n)$ is exchangeable, then whenever $\alpha \in \ub{\F}$ is a solution to the designer's problem, $\alpha_\pi(s_1,\dots,s_n):=\alpha(s_{\pi(1)},\dots, s_{\pi(n)})$, where $\pi$ is any permutation, must also be a solution by symmetry. Thus, $\tilde{\alpha}$ defined by averaging $\alpha_\pi$ over all permutations $\pi$ is also a solution. Therefore, there must be a symmetric optimal mechanism $\alpha(\,\cdot\,)$. By \Cref{prop:sym-nested-up-sets}, every extreme point of the set of symmetric monotone functions $\alpha:[0,1]^n \to [0,1]$ satisfying a linear constraint must be a mixture of indicator functions defined on two symmetric nested up-sets. As a result, there must be a symmetric optimal two-threshold policy.  
\end{proof}

The specifics of the two-threshold policy in \Cref{prop:public}, including the aggregation function $\phi$, the cutoffs $0 \leq k_L \leq k_H \leq 1$, and the randomization weight $p \in [0,1]$, of course depend on the details of the environment, but \Cref{prop:public} shows the optimal mechanism always has this simple feature regardless of how the signals are correlated and how the preferences are interdependent. 

To illustrate \Cref{prop:public}, consider the following example of positive externalities:

\begin{figure}[t]
    \centering
    \includegraphics[scale=0.45]{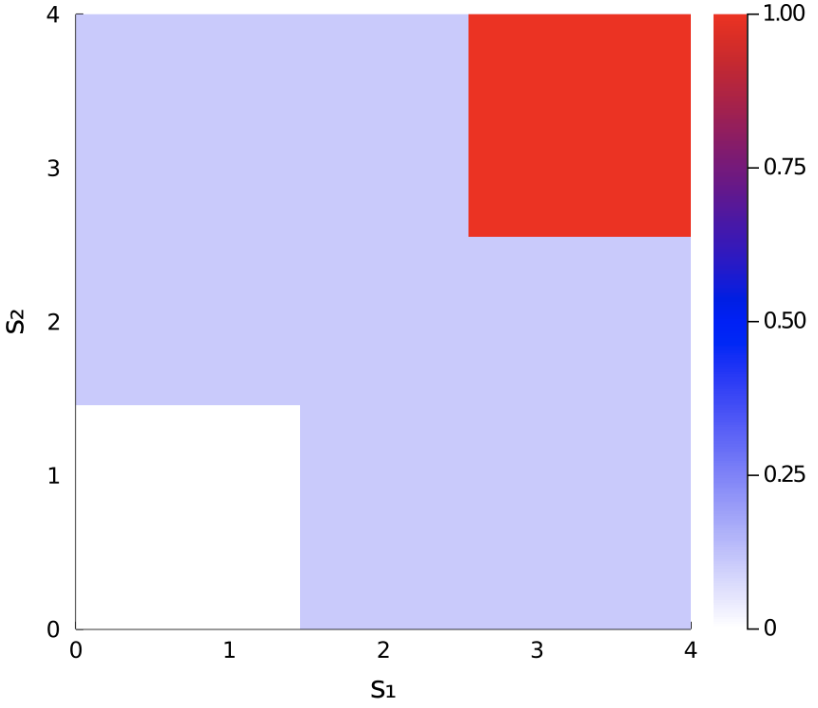}
    \caption{Illustration of optimal allocation rule $\alpha(s_1, s_2)$ for \Cref{ex:public}}
   \label{fig:public}
\end{figure}

\begin{ex}[Positive externalities]\label{ex:public}
Suppose that $n=2$, and that 
\[v_1(s_1, s_{2}) = s_1 + w \cdot s_{2} \quad \mbox{ and } \quad  v_2(s_1,s_2)=s_2+w\cdot s_1\,.\]
where $w > 0$. Suppose that the signals $(s_1, s_2)$ are also positively correlated and distributed according to a truncated, symmetric, joint log-normal distribution. \Cref{prop:public} immediately applies and shows that a symmetric two-threshold policy would be optimal. To illustrate, suppose the parameters are such that $w = 0.1$, and the log-normal distribution is truncated on $[0, 4]^2$, centered at $(2, 2)$, with $(\log(s_1), \log(s_2))$ having a standard deviation $0.4$ and correlation $0.5$. The cost $c = 3$. \Cref{fig:public} illustrates the optimal mechanism computed numerically (by discretizing the type space into a $30$-by-$30$ grid). In this example, the optimal mechanism asks each agent to submit their signals $s_i$, implements the project if both signals are very high, abandons the project if both signals are very low, and otherwise implements the project with a constant, interior probability. \qed
\end{ex}

It is noteworthy that even though the ironing technique \`{a} la \citet{Myerson1981} does not apply for our problem due to its nature of multidimensional monotonicity, \Cref{prop:public} can be applied to fully characterize the optimal mechanism in specific settings.\footnote{Indeed, \citet*{bedard2023multivariate} establish a version of the sweeping operator \`{a} la \citet{Rochet2003} for solving concave maximization problems under multidimensional monotonicity constraints---however, the sweeping operator is considerably less analytically tractable than the ironing procedure of \citet{Myerson1981}.} 

To illustrate, we consider a model of \textit{\textbf{limited negative externalities}}. There are two agents. They cause negative externalities to each other from their consumption of the public good, but they can take actions to secure themselves a payoff lower bound. Specifically, the value functions are given by 
\[v_1(s_1,s_2) = \max\{s_1 - s_2, 0\},\quad  v_2(s_1,s_2) = \max\{s_2 - s_1, 0\}\,.\]
Unlike the positive externalities example (\Cref{ex:public}), the ex-post efficient allocation rule here (implementing the project if and only if $|s_1-s_2|$ is above $c$) is \textit{not} monotone. Thus, the first-best allocation is not implementable even if we ignore the budget balance constraint. Nevertheless, \Cref{prop:public} can be used to derive the second-best mechanism. Consider the following \textit{\textbf{externality-refund mechanism}}: The designer posts a single price $k$ and implements the project if and only if at least one of the agents pays the price, in which case any agent $i$ who pays $k$ can get a refund of the negative externalities caused by the other agent up to the price (i.e., $\min\{k, s_{-i}\}$).

\begin{prop}\label{prop:example}
Suppose that there are two agents and limited negative externalities. If the joint distribution of $(s_1,s_2)$ is exchangeable and has a conditional hazard rate function $h(s_1|s_2):=g(s_1|s_2)/(1-G(s_1|s_2))$ that is strictly increasing in $s_1$ and decreasing in $s_2$, then an optimal mechanism can be implemented by randomizing over two externality-refund mechanisms. 
\end{prop}

\Cref{prop:example} says that with limited negative externalities, under the hazard-rate assumption, the optimal mechanism simply randomizes over two possible posted prices and offers a refund for the negative externalities up to the paid price. The proof of \Cref{prop:example} is in the appendix. It leverages \Cref{prop:public} and shows that the optimal two-threshold policy must have an aggregation function $\max\{s_1, s_2\}$. Indeed, the externality-refund mechanism implements the allocation rule $\1\{\max\{s_1, s_2\} \geq k\}$ in an ex-post Nash equilibrium---each agent $i$ pays if and only if $s_i \geq k$.\footnote{To see that this is an ex-post equilibrium, consider agent $i$'s problem, assuming that the other agent plays the proposed strategy. If $s_{-i} < k$, then agent $i$ faces an effective price of $k - s_{-i}$ given the refund against a benefit of $\max\{s_i - s_{-i}, 0\}$---the net is positive if and only if $s_i \geq k$; if the other agent has $s_{-i} \geq k$, then paying or not paying gives the same payoff to agent $i$ (the project always gets implemented and agent $i$ always gets the payment $k$ back if paying).} Intuitively, the mechanism tries to raise the chance of implementing the project when the surplus is high (i.e., when $s_1$ and $s_2$ differ a lot) while at the same time keeping it incentive compatible. 

As an example, suppose that $(s_1,s_2)$ are drawn from a symmetric, truncated normal distribution on $[0,1]^2$, with means $\mu \in [0,1]$,  variance $\sigma^2>0$, and correlation $\rho \in (0,1)$. By the properties of (truncated) normal distributions, $h(s_1|s_2)$ is indeed strictly increasing in $s_1$ and decreasing in $s_2$. Thus, \Cref{prop:example} applies and shows that a randomized externality-refund mechanism is optimal. In fact, holding other parameters fixed, one can show that there exists $\Bar{\rho} > 0$ such that for all $\rho \in (0, \Bar{\rho})$, a single threshold suffices in the two-threshold policy---the optimal mechanism posts a deterministic price and offers a refund for the negative externalities up to this price.

\subsection{Interim Efficient Frontier of Bilateral Trade}\label{subsec:bilateral}

Consider a classic bilateral trade problem as in \citet{myerson1983efficient}. There is a single good. A buyer privately observes their value $v \in [\underline{v}, \overline{v}]$, drawn from a continuous distribution $G_B$. A seller privately observes their cost $c \in [\underline{c}, \overline{c}]$, drawn from a continuous distribution $G_S$. For this application, we assume that $v$ and $c$ are independent. 

By the revelation principle, it is without loss of generality to focus on \textit{\textbf{(direct, incentive-compatible) mechanisms}} $(p, t):[\underline{v},\overline{v}]\times[\underline{c},\overline{c}] \rightarrow [0, 1] \times \R$, specifying the probabilities of trade and the expected payments made from buyer to seller, satisfying the Bayesian incentive compatibility and Bayesian individual rationality constraints: 
\begin{align*}
\E_c\Big[p(v, c) v  - t(v, c)\Big] &\geq  \E_c\Big[p(\hat{v}, c) v  - t(\hat{v}, c)\Big]   &&\text{ for all $v, \hat{v}$\,;} \tag{Buyer IC}\\
\E_v\Big[t(v, c) -  p(v, c)c   \Big] &\geq  \E_v\Big[t(v, \hat{c}) - p(v, \hat{c})c  \Big]   &&\text{ for all $c, \hat{c}$\,;} \tag{Seller IC}\\
\E_c\Big[p(v, c) v  - t(v, c)\Big] &\geq 0  &&\text{ for all $v$\,;} \tag{Buyer IR}\\
\E_v\Big[t(v, c) -  p(v, c)c   \Big] &\geq 0  &&\text{ for all $c$\,.} \tag{Seller IR}
\end{align*}

For a mechanism $M$, let $U_B(v)$ and $U_S(c)$ be the interim utility functions for the buyer and the seller. We are interested in characterizing the \textit{\textbf{interim efficient mechanisms}}. These mechanisms can be identified as solutions to weighted welfare maximization
\[
\max_{(p, t)} \left[\int_0^1 U_B(v) \Lambda_B(\d v)+\int_0^1 U_S(c) \Lambda_S(\d c)\right]
\]
under some \textit{\textbf{welfare weights}} $\Lambda_B$ and $\Lambda_S$, which are Borel measures on $[0,1]$, subject to the IC and IR constraints (\citealt{holmstrom1983efficient}).

We say that a mechanism $(p, t)$ is a \textit{\textbf{markup-pooling}} mechanism if there exists a nondecreasing function $\phi$, an interval $I=[c_L, c_H]$, and a constant $k \in [0, 1]$ such that 
\begin{itemize}
    \item If $c \not\in I$, then trade happens if and only if $v \geq \phi(c)$. 
    \item Otherwise, let $\tilde{c} = c_L$ with probability $k$ and $\tilde{c} = c_H$ with probability $1-k$. Trade happens if and only if $v \geq \phi(\tilde{c})$.  
\end{itemize}
A markup-pooling mechanism is illustrated by \Cref{fig:MPM}. In such a mechanism, trade is implemented if and only if the value $v$ is above a monotone transformation of the cost $\phi(c)$ (the \textit{\textbf{markup function}}), with the exception that when $c$ falls into a specific interval $I$ (the \textit{\textbf{pooling interval}}), we simply \textit{resample} the cost from the two ends of the interval and execute the trade if the buyer's value $v$ is above the marked up, resampled cost $\phi(\tilde{c})$.

\begin{figure}[t]
    \centering
    \includegraphics[scale=0.3]{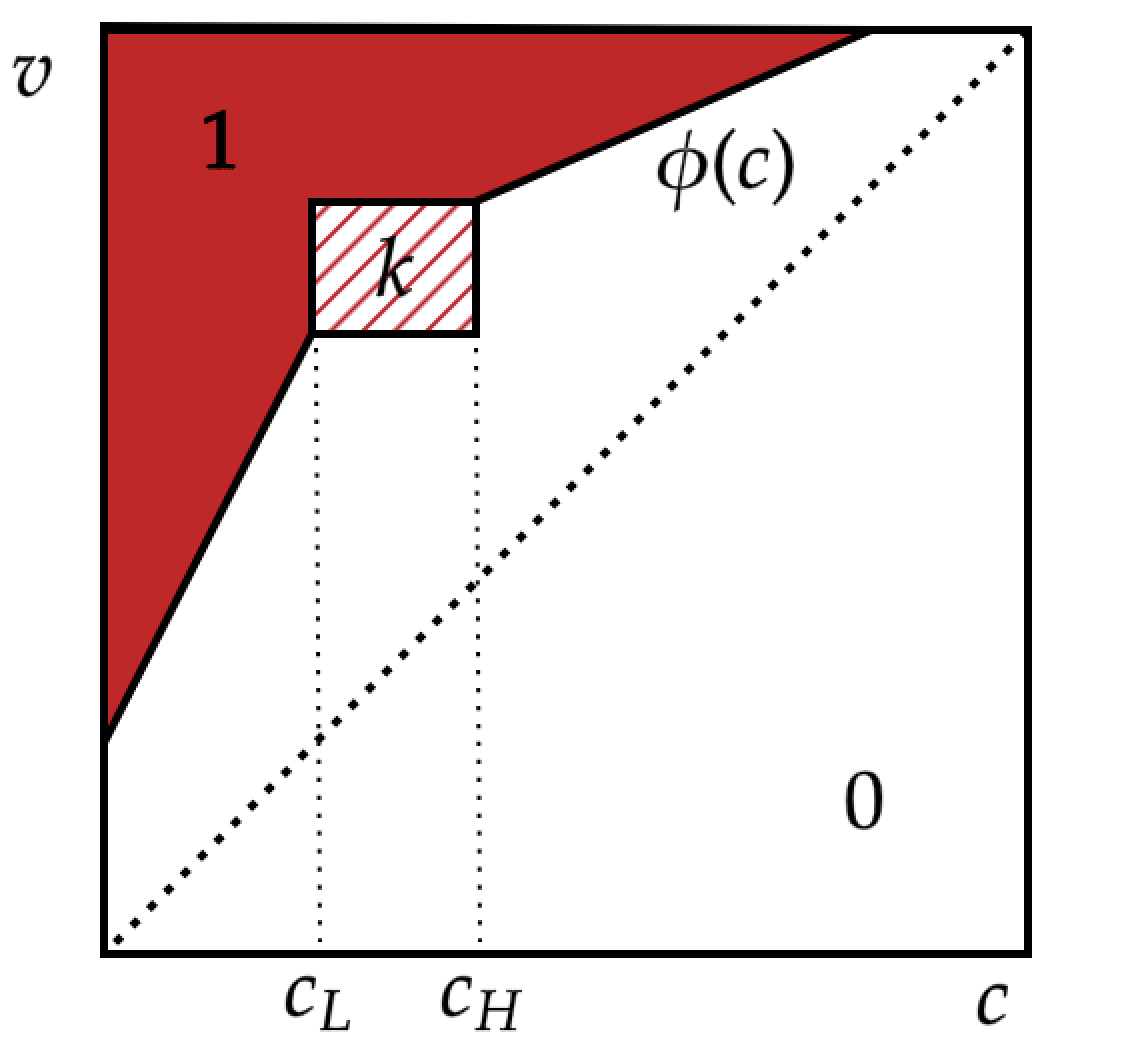}
    \caption{A Markup-Pooling Mechanism}
   \label{fig:MPM}
\end{figure}

\begin{prop}\label{prop:bilateral}
For any welfare weights $\Lambda_B$ and $\Lambda_S$, there exists a markup-pooling mechanism that maximizes the expected welfare. 
\end{prop}

\begin{proof}[Proof of \Cref{prop:bilateral}]
Let $\mathcal{M}$ be the set of IC and IR mechanisms. Let 
\[\mathcal{U}:= \Big\{(U^M_B, U^M_S): M \in \mathcal{M} \Big\}\]
be the set of implementable interim utility functions. We claim that any extreme point of $\mathcal{U}$ must be implementable by a markup-pooling mechanism. The result follows immediately from this claim and Bauer's maximum principle, since the expected welfare is a continuous linear functional on $\mathcal{U}$ by definition.

Clearly, any extreme point in $\mathcal{U}$ must be an extreme point in $\mathcal{U}_z:= \big\{(U^M_B, U^M_S): M \in \mathcal{M} \text{ and $U^M_B(\underline{v}) = z$} \big\}$ for some $z \in \R_+$. We show that for any $z \in \R_+$, the extreme points of $\mathcal{U}_z$ can be implemented with a markup-pooling mechanism. By \citet{myerson1983efficient}, we may relax the ex-post budget constraint to be ex-ante budget constraint without affecting $\mathcal{U}_z$. In particular, let
\[q_1(v) = \int p(v, c) G_S(\d c), \quad \mbox{ and } \quad q_2 (c) = \int p(v, c) G_B(\d v)\,.\]
The ex-ante budget constraint is simply that 
\[\Pi(q_1, q_2) := \int \text{MR}(v) q_1(v) G_B(\d v) -  \int \text{MC}(c) q_2 (c) G_S(\d c) = U_B(\underline{v}) +  U_S(\overline{c})\,,\]
where $\text{MR}(v):= v - \frac{1-G_B(v)}{g_B(v)}$ and $\text{MC}(c):= c + \frac{G_S(c)}{g_S(c)}$, with $g_i$ being the density of $G_i$. By \citet{myerson1983efficient}, we can write 
\[U_B(v) = \int_{\underline{v}}^{v} q_1(s) \d s + \underbrace{U_B(\underline{v})}_{=z} \,\qquad U_S(c) = \int^{\overline{c}}_c q_2(s) \d s +\underbrace{U_S(\overline{c})}_{= \Pi(q_1, q_2) - z}\,.\]
Now consider 
\[\mathcal{Q}:= \Big\{(q_1, q_2): (G_B,G_S)\text{-rationalizable; } \text{$q_1$ is nondecreasing and $q_2$ is nonincreasing}  \Big\}\]
and 
\[\overline{\mathcal{Q}}_z := \big\{(q_1, q_2)\in \mathcal{Q}: \Pi(q_1, q_2) \geq z\big\}\,.\]
Note that $\Pi(q_1, q_2)$ is a continuous linear functional of $(q_1, q_2)$. Therefore, for any $z \in \R_+$, we have that the map between $\overline{\mathcal{Q}}_z$ and $\mathcal{U}_z$ defined by 
\[L_z[q_1, q_2](v,c) := \Big(\int_{\underline{v}}^{v} q_1(s) \d s + z, \,\, \int^{\overline{c}}_c q_2(s) \d s + \Pi(q_1, q_2) - z\Big)\]
is a continuous linear map. Moreover, $\mathcal{U}_z = L_z(\overline{\cQ}_z)$. It follows from \Cref{lem:extreme-projection} that any extreme point of $\mathcal{U}_z$ must be implementable by an extreme point of $\overline{\mathcal{Q}}_z$. By \Cref{thm:rectangle} (and \Cref{cor:rectangle}), the extreme points must be $\1_{A_1} + \lambda \1_{A_2 \backslash A_1}$, where $A_1\subseteq A_2$ are two nested up-sets in the space of $(v, -c)$, and $A_2 \backslash A_1$ is a rectangle. Note that any such mechanism can be implemented by a markup-pooling mechanism, concluding the proof. 
\end{proof}

As the proof of \Cref{prop:bilateral} shows, the class of markup-pooling mechanisms can in fact attain any extreme point of the set of interim utility functions. An immediate consequence is that \textit{any} trading mechanism is payoff-equivalent to a randomization over markup-pooling mechanisms. Moreover, by our unique rationalizability result (see \Cref{thm:rectangle}), it also follows that if a \textit{\textbf{dominant-strategy incentive compatible (DIC)}} mechanism can attain any such extreme point, then it must be a markup-pooling mechanism. 

\begin{prop}\label{prop:payoff-equivalent}
Any mechanism is payoff-equivalent to a randomization over markup-pooling mechanisms. Moreover, if a DIC mechanism attains an extreme point of the set of feasible interim payoffs, then it must be a markup-pooling mechanism.     
\end{prop}

\begin{proof}[Proof of \Cref{prop:payoff-equivalent}]
By the proof of \Cref{prop:bilateral}, we know that any extreme point of the set of interim utility functions $\mathcal{U}$ must be implementable with a markup-pooling mechanism. Now fix any mechanism $M$. Let $V \in \mathcal{U}$ be the interim utility pairs generated by $M$. By Choquet's theorem, there exists $\mu \in \Delta(\text{ext}(\mathcal{U}))$ such that 
\[V = \int_{\text{ext}(\mathcal{U})} U \mu (\d U)\,.\]
Thus, randomizing over the markup-pooling mechanisms that implement $\text{ext}(\mathcal{U})$ according to $\mu$ generates a payoff-equivalent mechanism. 

Now we prove the second part. Fix any DIC mechanism $(p, t)$ that attains any extreme point of $\mathcal{U}$. Let $z = U_B(\underline{v})$.  Clearly, the mechanism $(p, t)$ must also attain an extreme point of $\mathcal{U}_z$ where $\mathcal{U}_z$ is defined as the set of interim payoffs where the buyer of type  $\underline{v}$ obtains payoff $z$, as in the proof of \Cref{prop:bilateral}. Now, by the proof of \Cref{prop:bilateral}, any extreme point of $\mathcal{U}_z$ must be implementable by some extreme point $(q_1, q_2)$ of the set $\overline{\mathcal{Q}}_z$ defined there. Since $(q_1, q_2)$ are pinned down by $(U_B, U_S)$ as their derivatives, the interim allocation rules of the mechanism $(p,t)$ must be an extreme point of $\overline{\mathcal{Q}}_z$. Moreover, we know that they must be $(G_B, G_S)$-rationalized by $p(v, c)$ which is monotone in $(v, -c)$, since $(p, t)$ is DIC. By \Cref{thm:rectangle} (and \Cref{cor:rectangle}), any extreme point of $\overline{\mathcal{Q}}_z$ must be uniquely $(G_B, G_S)$-rationalized among all monotone functions, and hence $p(v, c)$ must be identical to a mixture of two indicator functions defined on two nested up-sets that differ by a rectangle. Thus, $(p, t)$ is a markup-pooling mechanism. 
\end{proof}

\begin{figure}[t]
    \centering
    \includegraphics[scale=0.5]{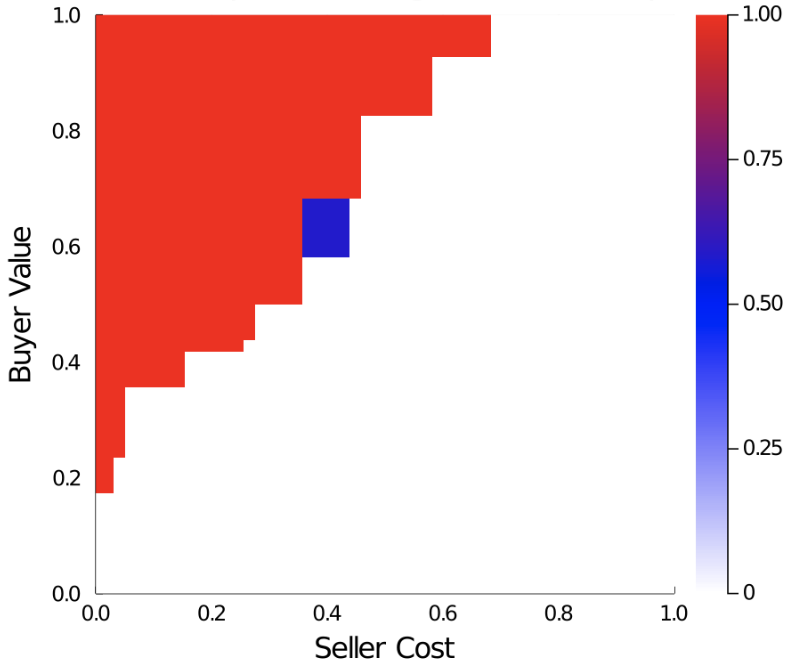}
    \caption{Illustration of interim-efficient allocation rule $p(v, c)$ for \Cref{ex:trade}}
    \label{fig:trade}
\end{figure}

\begin{ex}\label{ex:trade}
To illustrate the markup-pooling mechanisms and \Cref{prop:bilateral}, we provide a numerical example. Consider the case where $v, c \in [0, 1]$. \Cref{fig:trade} depicts a numerically computed optimal mechanism in a $50$-by-$50$ discretized grid, where we uniformly sample a generic probability density and a generic welfare weight for both the buyer and the seller. The blue rectangle in \Cref{fig:trade} is the region where the trade happens with a constant, interior probability. The allocation rule can be implemented using a markup-pooling mechanism, where the markup function is the lower boundary of the larger up-set (the union of the red up-set and the blue rectangle). The pooling interval is the interval of the seller's types that fall into the blue rectangle. \qed
\end{ex}

Unlike in \Cref{subsec:public}, the monotonicity constraint here is the standard one-dimensional monotonicity resulting from Bayesian IC. But the difficulty arises when dealing with the budget balance constraint. To the best of our knowledge, the existing literature on this problem (e.g. \citealt{myerson1983efficient}; \citealt{ledyard1999characterization}) assumes joint conditions on welfare weights and type distributions such that \textit{(i)} the monotonicity condition never binds, and \textit{(ii)} the Lagrangian relaxation has a unique solution regardless of the multiplier on the budget balance constraint. Under such assumptions, the optimal mechanism never involves randomization. The difficulty without these assumptions is that the budget balance constraint is a type-dependent linear constraint that would interact with the monotonicity constraint. In particular, without these assumptions, there would generally be many optimal solutions to the ironed Lagrangian, and one needs to carefully select an optimal solution that is feasible with respect to the budget constraint. One can use ironing to show that the optimal mechanism may assign arbitrary values to a countable collection of rectangles. Our extreme point characterization refines the ironing approach by showing that the trading probabilities are all equal to a \textit{\textbf{constant}} on these rectangles, and in fact, there is only a \textit{\textbf{single}} rectangle needed in the optimal allocation rule. Together, these two features characterize our markup-pooling mechanisms.

\subsection{Mechanism Anti-Equivalence}\label{subsec:anti-equivalence}

Consider a standard social choice environment with linear utilities and independent, one-dimensional, private types as in \citet{gershkov2013equivalence}. Suppose that we have two agents $\mathcal{I}=\{1, 2\}$ and two alternatives $\mathcal{K}= \{1,2\}$. Each agent $i$ privately observes their type $\theta_i \in \Theta_i \subset \R$, where $\Theta_i$ is a compact interval and $\theta_i$ follows some continuous distribution $G_i$. Agent $i$'s utility for alternative $k$ is given by 
\[u^k_i(\theta_i, t_i) = a^k_i \theta_i + c^k_i - t_i\]
where $a^k_i, c^k_i \in \R$ are constants with $a^1_i \neq a^2_i$, and $t_i \in \R$ is a monetary transfer. 
Examples of such an environment include the classic public good provision and bilateral trade problems. By the revelation principle, it is without loss of generality to consider direct mechanisms
\[(p, t): \Theta_1 \times \Theta_2 \rightarrow \Delta(\mathcal{K}) \times \R^2\]
that specify the probabilities of implementing each alternative and the transfers from each agent for a given profile of type report.

A mechanism $(p,t)$ is \textit{\textbf{Bayesian incentive compatible (BIC)}} if truthful reporting by all agents is a Bayes-Nash equilibrium. A mechanism $(p,t)$ is \textit{\textbf{dominant strategy incentive compatible (DIC)}} if truthful reporting is a
 dominant strategy equilibrium. A mechanism $(p,t)$ is \textit{\textbf{deterministic}} if $p(\theta)$ is a deterministic allocation for all $\theta$. 

We study the question of interim equivalence between mechanisms initiated by \citet{manelli2010bayesian}. Two BIC mechanisms $(p, t)$ and $(\tilde{p}, \tilde{t})$ are \textit{\textbf{payoff equivalent}} if they yield the same ex-ante expected surplus and the same interim expected utilities for all agents.\footnote{Because we have two alternatives, this equivalence notion is the same as asking the two mechanisms to have the same interim expected allocation probabilities.} From a series of fundamental contributions (\citealt{manelli2010bayesian}; \citealt{gershkov2013equivalence}; \citealt{chen2019equivalence}), we now know that for any BIC mechanism, there exists a payoff-equivalent DIC mechanism (\textit{\textbf{BIC-DIC equivalence}}); for any stochastic mechanism, there exists a payoff-equivalent deterministic mechanism (\textit{\textbf{stochastic-deterministic equivalence}}). These two equivalence results suggest that for a wide range of objectives, it is without loss to consider DIC mechanisms which are strategically simpler and robust to changes in the agents' beliefs; similarly, it is without loss to consider deterministic mechanisms which do not require a verifiable randomization device. 

However, as an application of our abstract results, we show that there is a strong conflict between asking for strategically robust implementations and asking for verifiable deterministic mechanisms. 
Say that two payoff-equivalent mechanisms are \textit{\textbf{ex-post equivalent}} if they implement the same ex-post allocation probabilities for all type profiles (up to measure zero). 

\begin{prop}\label{prop:anti}
A BIC mechanism is payoff-equivalent to a deterministic DIC mechanism if and only if they are ex-post equivalent. 
\end{prop}
\begin{proof}[Proof of \Cref{prop:anti}]For any BIC mechanism $(p, t)$, by the Envelope theorem, 
\[U_i(\theta_i) = U_i(\underline{\theta}_i) + \int_{\underline{\theta}_i}^{\theta_i} \sum_k a^k_i q^k_i(s_i) \d s_i\,, \]
where $U_i(\theta_i)$ denotes the interim payoff for type $\theta_i$, and $q^k_i(\theta_i)$ denotes the interim allocation probability of alternative $k$ for type $\theta_i$. Thus, the payoff equivalence between two mechanisms $(q, t)$ and $(\tilde{q}, \tilde{t})$ implies that for all $i$ and $\theta_i$
\[ \sum_k a^k_i q^k_i(\theta_i) = \sum_k a^k_i \tilde{q}^k_i(\theta_i)\,.\]
Since $\sum_k q^k_i(\theta_i) = 1$, and since we have two alternatives, the above implies that $q^k_i(\theta_i) = \tilde{q}^k_i(\theta_i)$ for all $k$, $i$, and $\theta_i$. 

Now, note that for any BIC mechanism $(p, t)$, by a standard argument, we also have that $\sum_k a^k_i q^k_i(\theta_i)$ is nondecreasing in $\theta_i$. Because we have two alternatives, this is equivalent to $a^2_i + (a^1_i - a^2_i) q^1_i(\theta_i)$ being nondecreasing in $\theta_i$. Thus, for $a^1_i > a^2_i$, this is equivalent to $q^1_i(\theta_i)$ being nondecreasing, and for $a^1_i<a^2_i$, this is equivalent to $q^1_i(\theta_i)$ being nonincreasing. If $a^1_i > a^2_i$, let 
\[\widehat{q}^1_i(s_i) := q^1_i(G^{-1}_i(s_i))\,,\]
and otherwise let 
\[\widehat{q}^1_i(s_i) := q^1_i(G^{-1}_i(1-s_i))\,.\]
Then $\widehat{q}^1_i$ must be nondecreasing in $s_i$. Define analogously $\widehat{p}^1(s_1, s_2)$, where we apply the transformation $G^{-1}_i(\,\cdot\,)$ or $G^{-1}_i(1-\,\cdot\,)$ to each argument of $p^1(\theta_1, \theta_2)$ depending on whether $a^1_i > a^2_i$. For any BIC mechanism $(p,t)$, clearly the monotone pair $(\widehat{q}^1_1, \widehat{q}^1_2)$ must be rationalized by  $\widehat{p}^1$. 
Moreover, for any DIC mechanism $(p,t)$, by a similar argument, we also have that $\widehat{p}^1$ must be monotone. 

Fix two payoff-equivalent mechanisms $(p, t)$ and $(\tilde{p}, \tilde{t})$. Their corresponding transformed interim allocation rules $(\widehat{q}^1_1, \widehat{q}^1_2)$ and $(\widehat{\tilde{q}}^1_1, \widehat{\tilde{q}}^1_2)$ must coincide. Now, if $(p, t)$ is deterministic \textit{and} DIC, then by the above argument $\widehat{p}^1$ must be an indicator function defined on an up-set. Thus, by \Cref{thm:rationalized-upsets},  $(\widehat{q}^1_1, \widehat{q}^1_2)$ is an extreme point of the set of monotone rationalizable pairs $\mathcal{Q}$, and hence must be uniquely rationalized, which implies $\widehat{\tilde{p}}^1 \equiv \widehat{p}^1$. Therefore, the two mechanisms are ex-post equivalent. 
\end{proof}

Another way to state \Cref{prop:anti} is that every deterministic DIC mechanism must yield interim payoffs that differ from those of every mechanism that is non-deterministic or non-DIC (in the sense of ex-post allocations). \Cref{prop:anti} shows that, in a strong sense, the BIC-DIC equivalence results of \citet{manelli2010bayesian} and \citet{gershkov2013equivalence} necessarily rely on non-deterministic mechanisms. Similarly, the stochastic-deterministic equivalence result of \citet{chen2019equivalence} necessarily relies on non-DIC mechanisms.\footnote{\citet{chen2019equivalence} provide an example to show that not every stochastic BIC mechanism is payoff-equivalent to a deterministic DIC mechanism. \Cref{prop:anti} shows that this is in fact never possible in the two-agent, two-alternative setting.}

\subsubsection{Beyond two agents: Exposed mechanisms}

As the proof of \Cref{prop:anti} shows, the key intuition behind the mechanism anti-equivalence theorem is that the interim allocation rules of the deterministic DIC mechanisms are extreme points of the set of feasible interim allocation rules, and that these extreme points must be uniquely rationalized, by \Cref{thm:rationalized-upsets}, and hence pin down the ex-post allocation rules. While the fact that deterministic DIC mechanisms induce extremal interim allocation rules relies on the two-agent assumption, the property that certain ``optimal mechanisms'' must be uniquely rationalized turns out to generalize to an arbitrary number of agents. We say that a mechanism $(p, t)$ is \textit{\textbf{exposed}} if it is uniquely optimal for some linear objective of the interim allocation rules (e.g., efficiency or revenue).

\begin{prop}\label{prop:exposed}
Under $n$ agents and two alternatives, a mechanism is payoff-equivalent to an exposed mechanism if and only if they are ex-post equivalent. 
\end{prop}

\begin{proof}[Proof of \Cref{prop:exposed}]
By \Cref{rmk:exposed} and \Cref{thm:exposedQ}, if $(p, t)$ is an exposed mechanism, then its induced interim allocation rule $(q_1,\dots, q_n)$, when appropriately defined as in the proof of \Cref{prop:anti}, must be the marginals of $\1_A$ where $A$ is an additive set. By the proof of \Cref{thm:rationalized-upsets} (or by \citealt{fishburn1990sets}), we have that $\1_A$ is uniquely rationalizable. Thus, the result follows by the proof of \Cref{prop:anti}.
\end{proof}

\Cref{prop:exposed} shows that for any such optimal mechanism, there is no non-trivial mechanism equivalence theorem since there exists a unique implementation in terms of the ex-post allocation rule. Of course, this is not in contradiction to \citet{manelli2010bayesian}, \citet{gershkov2013equivalence} and \citet{chen2019equivalence}---the result shows that the unique ex-post allocation rule corresponding to any such optimal mechanism must be implementable via a deterministic DIC mechanism.

\subsection{Asymmetric Reduced Form Auctions}\label{subsec:auction}

Consider a classic auction setting \'{a} la \citet{Myerson1981} with two bidders. There is one seller and one good. Each bidder $i$ privately observes their value $v_i$ for the object, drawn from a distribution $G_i$. We assume that $v_1$ and $v_2$ are independent, but can be drawn from two different distributions. 

By the revelation principle, it is without loss of generality to focus on \textit{\textbf{(direct, incentive-compatible) mechanisms}} $(p_1,p_2, t_1, t_2):[\underline{v}_1,\overline{v}_1]\times[\underline{v}_2, \overline{v}_2] \rightarrow [0, 1]^2 \times \R^2$, which specifies the probabilities of allocating the good to bidder $1$ and bidder $2$ as well as the monetary transfers they make to the seller, satisfying the Bayesian incentive compatibility, Bayesian individual rationality constraints, and feasibility constraint: 
\begin{align*}
\E_{v_{-i}}\Big[p_i(v) v_i  - t_i(v)\Big] &\geq  \E_{v_{-i}}\Big[p_i(
\hat{v}_i, v_{-i}) v_i  - t_i(\hat{v}_i, v_{-i})\Big]  &&\text{ for all $i$, $v_i$, $\hat{v}_i$\,;} \tag{BIC} \\
\E_{v_{-i}}\Big[p_i(v) v_i  - t_i(v)\Big] &\geq 0 &&\text{ for all $i$, $v_i$\,;} \tag{BIR}\\
\sum_i p_i(v)  &\leq 1 &&\text{ for all $v$\,.} \tag{Feasibility} 
\end{align*}

A pair of monotone functions $(q_1,q_2)$, where $q_i: [\underline{v}_i, \overline{v}_i] \rightarrow [0, 1]$, is a \textit{\textbf{reduced-form auction}} if there exists a mechanism $(p, t)$ such that for all $i$
\[q_i(v_i) = \int p_i(v_i, v_{-i}) G_{-i}(\d v_{-i})\,.\]

An immediate consequence of our main result is a characterization of the reduced-form auctions with two bidders and their extreme points: 
\begin{prop}\label{prop:reduced}
$(q_1, q_2)$ is a reduced-form auction if and only if the associated quantile allocations $\tilde{q}_i(s):= q_i(G^{-1}_i(s))$ satisfy 
\[\tilde{q}_1 \preceq_w\tilde{q}^{-1}_2\,,\]
with every extreme point satisfying 
\[\tilde{q}_1 \equiv \tilde{q}^{-1}_2 \1_{[k_1,1]} \quad\text{ and }\quad \tilde{q}_2 \equiv \tilde{q}^{-1}_1 \1_{[k_2,1]}  \]
for some $k_1, k_2 \in [0, 1]$. 
\end{prop}
\begin{proof}[Proof of \Cref{prop:reduced}]
We start by showing that any reduced form auction $(q_1, q_2)$ must satisfy $\tilde{q}_1 \preceq_w \tilde{q}^{-1}_2$. Let 
\[\overline{q}_1(x_1) := \int_0^1 [1 - p_1(G^{-1}_1(x_1), G^{-1}_2(x_2))] \d x_2\,, \qquad \overline{q}_2(x_2) := \int_0^1 [1 - p_1(G^{-1}_1(x_1), G^{-1}_2(x_2))] \d x_1\,.\]
Let $\overline{q}^\uparrow_i$ be the monotone rearrangement of $\overline{q}_i$. By construction, $(\overline{q}^\uparrow_1, \overline{q}^\uparrow_2)$ is rationalizable and hence by \Cref{lem:gutmann}, $\overline{q}^\uparrow_1 \preceq (\overline{q}_2^\uparrow)^{-1}$. Thus,
\[ \int^s_0 \overline{q}^\uparrow_1(z) \d z \geq \int^s_0 1-(\overline{q}^\uparrow_2)^{-1}(1-z) \d z\,. \]
By definition of $\overline{q}^\uparrow_1$, and the monotonicity of $\tilde{q}_1$, we have $\overline{q}^\uparrow_1(z) = 1 - \tilde{q}_1(1-z)$ and hence 
\[ \int^s_0  [1 - \tilde{q}_1(1-z)] \d z \geq \int^s_0 [1-(\overline{q}^\uparrow_2)^{-1}(1-z)] \d z\,, \]
which is equivalent to 
\[ \int^1_{1-s}   \tilde{q}_1(z) \d z \leq \int^1_{1-s}(\overline{q}^\uparrow_2)^{-1}(z) \d z\,. \]
But since $p_1 + p_2 \leq 1$, we have $\overline{q}_2 \geq \tilde{q}_2$, and hence $\overline{q}^{\uparrow}_2 \geq \tilde{q}^{\uparrow}_2 = \tilde{q}_2$, where the equality is due to the monotonicity of $\tilde{q}_2$. Thus, $(\overline{q}^{\uparrow}_2)^{-1} \leq (\tilde{q}_2)^{-1}$, and hence
\[ \int^1_{1-s}   \tilde{q}_1(z) \d z \leq \int^1_{1-s}(\tilde{q}_2)^{-1}(z) \d z\,. \]
Since this holds for all $s$, we have $\tilde{q}_1 \preceq_w \tilde{q}^{-1}_2$.

Now, for the other direction, consider the set 
\[\tilde{\mathcal{Q}}:=\Big\{(\tilde{q}_1, \tilde{q}_2): \tilde{q}_1 \preceq_w \tilde{q}^{-1}_2\,; \tilde{q}_1, \tilde{q}_2 \text{ are nondecreasing and left-continuous} \Big\}\,.\]
By \Cref{prop:weak}, we know that $\tilde{\mathcal{Q}}$ is convex and every extreme point of $\tilde{\mathcal{Q}}$ satisfies  
\begin{equation}\label{eq:one-side}
\tilde{q}_1 \equiv \tilde{q}^{-1}_2  \1_{[k, 1]} 
\end{equation}
for some $k \in [0, 1]$.\footnote{Indeed, consider the convex set $\mathcal{Q}^w:= \big\{(q_1, q_2): q_1 \preceq_w \hat{q}_2\,; q_1, q_2 \text{ are nondecreasing, left-continuous} \big\}$ in \Cref{prop:weak}, and the linear map $L$ defined by $(q_1, q_2) \mapsto (q_1, 1 - q_2(1-\cdot))$. Let $\check{q}_2(x_2):= 1 - q_2(1-x_2)$. Note that $(\check{q}_2)^{-1} = \hat{q}_2$. Thus, $q_1 \preceq_w \hat{q}_2$ if and only if $q_1 \preceq_w \check{q}^{-1}_2$. Therefore, $L(\mathcal{Q}^w) = \{(q_1, q_2): q_1 \preceq_w q^{-1}_2\,;\, q_1, q_2 \text{ are nondecreasing, left-continuous}\} = \tilde{\cQ}$. By \Cref{lem:extreme-projection}, we have $\text{ext}(L(\cQ^w))\subseteq L(\text{ext}(\cQ^w))$. Thus, $\tilde{\cQ}$ is convex and every extreme point of $\tilde{\cQ}$ satisfies $q_1 \equiv q^{-1}_2\1_{[k, 1]}$ for some $k \in [0, 1]$ by \Cref{prop:weak}.}
We show that any $(q_1, q_2)$ satisfying the above is a reduced-form auction. In particular, let 
\[\qquad \qquad \tilde{p}_1(x_1, x_2) := \begin{cases}
 1 &\text{ $x_1 \geq \max\{k, \tilde{q}_2(x_2) \}$} \\
 0 & \text{ otherwise} \,;
\end{cases}\]
and 
\[\tilde{p}_2(x_1, x_2) := \begin{cases}
 1 &\text{ $x_1 < \tilde{q}_2(x_2)$} \\
 0 & \text{ otherwise} \,.
\end{cases}\]
Let $p_i(v_1, v_2) := \tilde{p}_i(G_1(v_1), G_2(v_2))$ be the allocation rules in the value space. One may verify that $(p_1, p_2)$ induces $(q_1, q_2)$, and hence $(q_1, q_2)$ is a reduced-form auction.

Now, take any $(q_1, q_2)$ such that $(\tilde{q}_1, \tilde{q}_2) \in \tilde{\mathcal{Q}}$. By Choquet's integral representation theorem, we know that there exists $\mu \in \Delta(\mathrm{ext}(\tilde{\mathcal{Q}}))$ such that 
\[\tilde{q} = \int_{\mathrm{ext}(\tilde{\mathcal{Q}})} q^\dagger \mu (\d q^\dagger)\,. \]
For any $q^\dagger \in \mathrm{ext}(\tilde{\mathcal{Q}})$, by the above construction, there exists $\tilde{p}(q^\dagger)$ that induces $q^\dagger$. It follows immediately that we can induce $(q_1, q_2)$ by setting $p(v_1, v_2) := \overline{p}(G(v_1), G(v_2))$ where 
\[ \overline{p} := \int_{\mathrm{ext}(\tilde{\mathcal{Q}})} \tilde{p}(q^\dagger) \mu (\d q^\dagger)\,, \]
proving that $(q_1, q_2)$ is a reduced-form auction. 

Lastly, we show that any extreme point of the set $\tilde{\mathcal{Q}}$ must satisfy
\[
\tilde{q}_1\equiv\tilde{q}_2^{-1}\1_{[k_1,1]} \quad \mbox{ and } \quad \tilde{q}_2 \equiv  \tilde{q}_1^{-1} \1_{[k_2,1]}
\]
for some $k_1,k_2 \in [0,1]$. As we have already shown, every extreme point of $\tilde{\cQ}$ must satisfy $\tilde{q}_1 \equiv \1_{[k_1,1]}\tilde{q}_2^{-1}$ for some $k_1 \in [0,1]$, according to \eqref{eq:one-side}. To show that every extreme point of $\tilde{\cQ}$ must also satisfy  $\tilde{q}_2 \equiv \1_{[k_2,1]}\tilde{q}_{1}^{-1}$ for some $k_2 \in [0,1]$, note that $\tilde{q}_1 \preceq_w \tilde{q}_2^{-1}$ if and only if $\tilde{q}_2 \preceq_w \tilde{q}_1^{-1}$.\footnote{To see this, by Theorem 4.A.6 of \citet{shaked2007stochastic}, $\tilde{q}_1 \preceq_w \tilde{q}_2^{-1}$ if and only if there exists a nondecreasing, left-continuous $\tilde{q}:[0,1]\to [0,1]$ such that $\tilde{q}_1 \preceq \tilde{q} \leq \tilde{q}_2^{-1}$. Therefore, $\tilde{q}_2 \leq \tilde{q}^{-1} \preceq \tilde{q}_1^{-1}$. Again, by Theorem 4.A.6 of \citet{shaked2007stochastic}, $\tilde{q}_2 \preceq_w \tilde{q}_1^{-1}$.}As a result, $\tilde{q} \in \tilde{Q}$ if and only if $\tilde{q}_2 \preceq_w \tilde{q}^{-1}$. By \Cref{prop:weak} again, it then follows that any extreme point of $\tilde{\mathcal{Q}}$ must also satisfy $\tilde{q}_2\equiv \tilde{q}_1^{-1}\1_{[k_2,1]}$ for some $k_2 \in [0,1]$.  This completes the proof. 
\end{proof}

The first part of \Cref{prop:reduced} recovers the reduced-form asymmetric auction characterization given by \citet*[][Corollary 8]{he2024private}. The proof of \citet*{he2024private} is based on their characterization of feasible posterior beliefs induced by a private private information structure. In contrast, \Cref{prop:reduced} results from the characterization of extreme points of weak joint majorization sets and Choquet's integral representation theorem. 

Moreover, \Cref{prop:reduced} effectively characterizes the extreme points of reduced-form asymmetric auctions in a two-bidder setting, which in turn pins down structures of optimal auctions for an arbitrary quasi-convex objective function that depends on the interim allocations.\footnote{This follows from a suitable version of Bauer's maximum principle; see e.g. \citet{ball2023bauer}.} Specifically, we have the following consequence: 

\begin{cor}\label{cor:opt}
Any upper-semicontinuous and quasi-convex functional of the reduced-form auctions is maximized by some $(q_1^\star,q_2^\star)$ such that
\[
\tilde{q}_1^\star \equiv \1_{[k_1,1]}(\tilde{q}_2^\star)^{-1} \quad \mbox{ and } \quad \tilde{q}_2^\star \equiv \1_{[k_2,1]}(\tilde{q}_1^\star)^{-1}\,,
\]
for some $k_1,k_2 \in [0,1]$. 
\end{cor}

\begin{figure}[t]
\centering
\begin{subfigure}[b]{0.3\linewidth}
\centering
\includegraphics[scale=0.22]{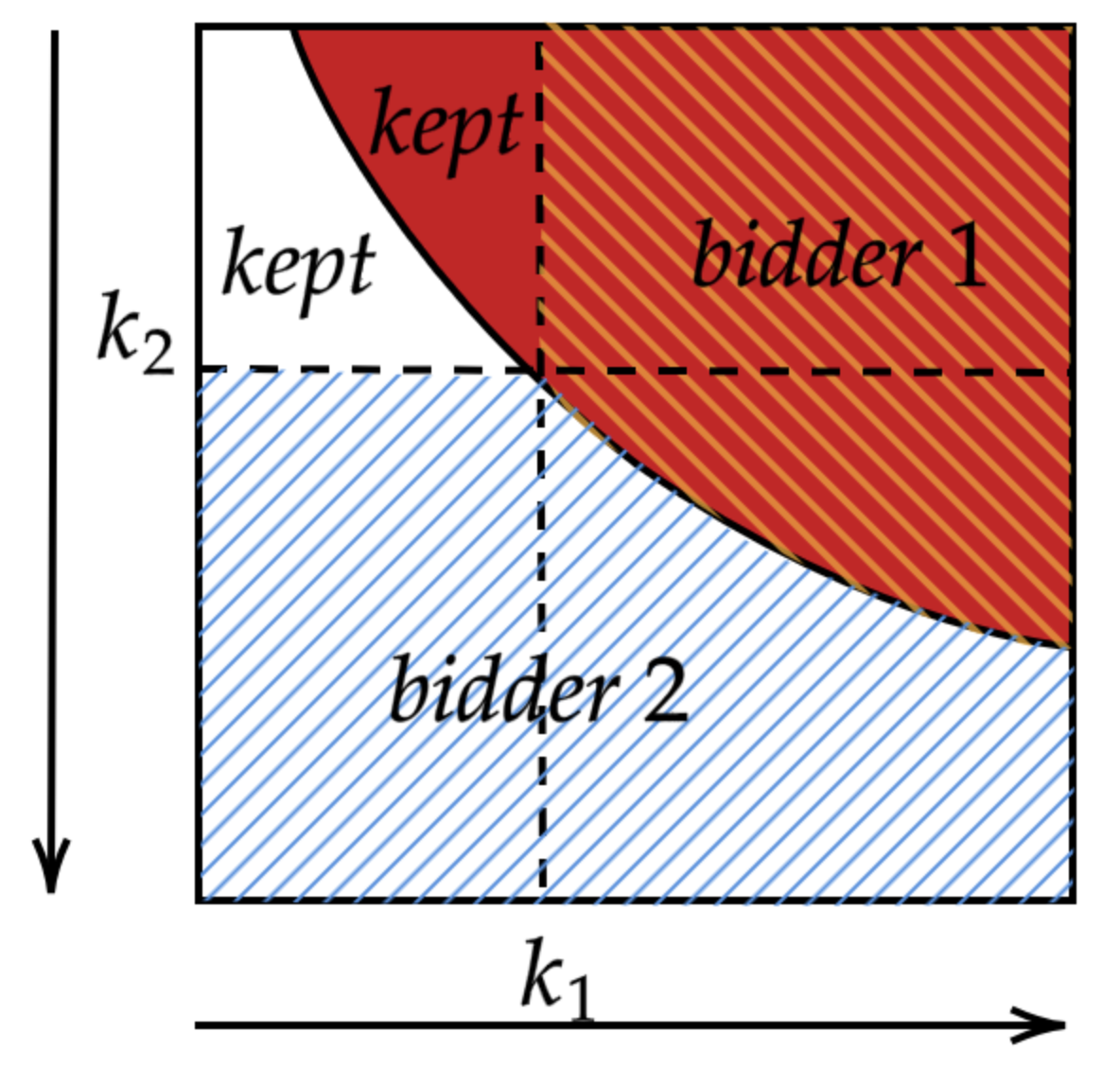}
\caption{An Extreme Point}
\label{fig:reduced-form}
\end{subfigure}%
\begin{subfigure}[b]{0.3\linewidth}
\centering
 \includegraphics[scale=0.22]{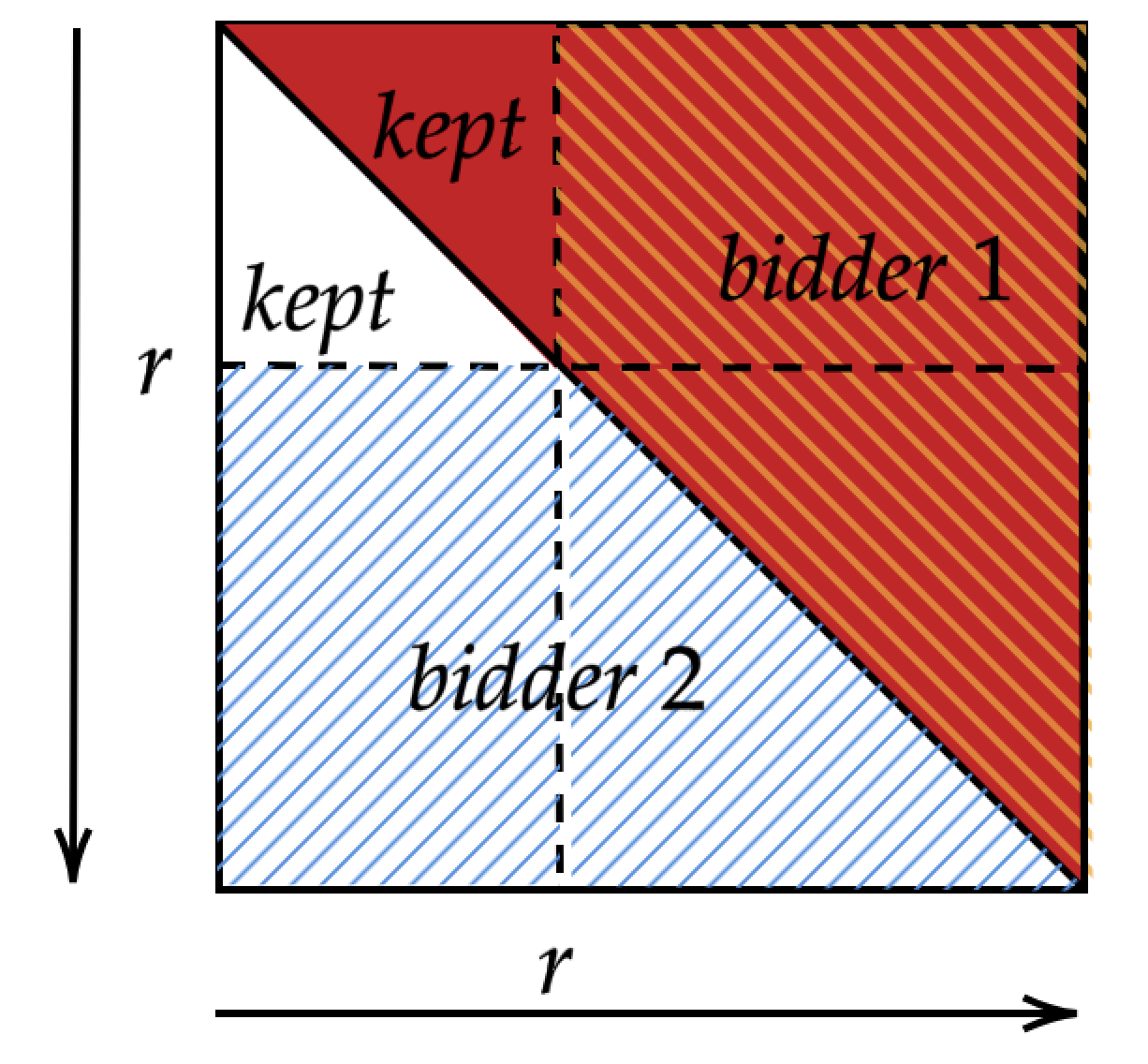}
\caption{SPA with Reserve}
\label{fig:reduced-form-b}
\end{subfigure}%
\begin{subfigure}[b]{0.3\linewidth}
\centering
\includegraphics[scale=0.22]{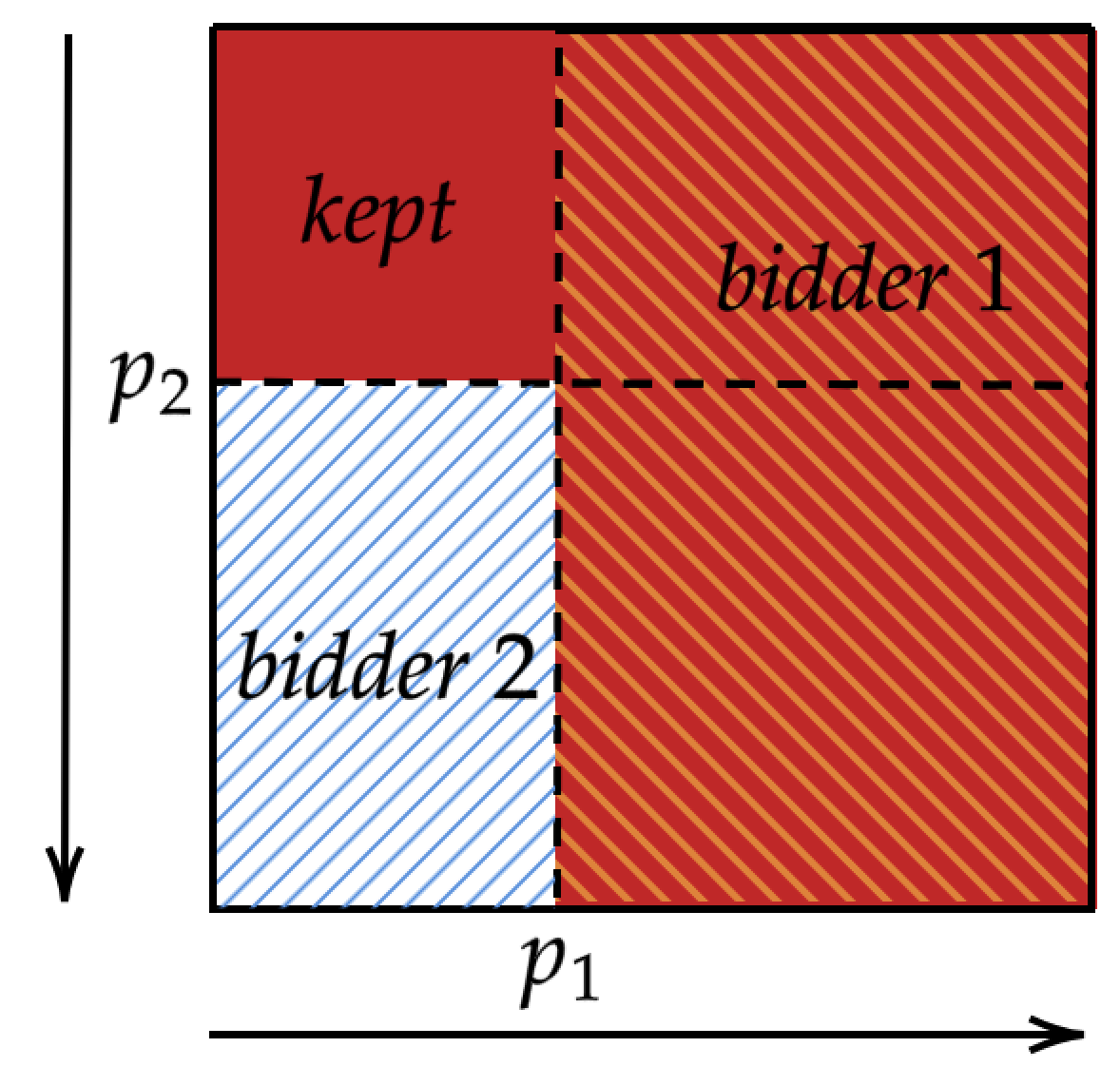}
\caption{Sequential Posted Prices}
\label{fig:reduced-form-c}
\end{subfigure}%
\caption{Extreme Points of $\tilde{\cQ}$}
\label{fig:auction}
\end{figure}

The optimal reduced-form mechanisms in \Cref{cor:opt} can be implemented by mechanisms that correspond to truncations of up-sets. \Cref{fig:reduced-form} illustrates a generic mechanism that implements an extremal reduced-form auction. In \Cref{fig:reduced-form}, the horizontal axis corresponds to bidder 1's type $\theta_1$ and the vertical axis corresponds to bidder 2's type $1-\theta_2$ (this axis is inverted so that the up-sets defined in this direction coincide with those given by the abstract results). The allocation rule is given by the following: Whenever both bidder $i$'s type is lower than $k_i$, the objective is not allocated. Otherwise, the object is allocated to bidder 1 if the type profile belongs to the up-set depicted in \Cref{fig:reduced-form}, and to bidder 2 if otherwise. For example, \Cref{fig:reduced-form-b} corresponds to a second-price auction with reserve price $r$, which is optimal in a symmetric, regular, and linear environment (\citealt{Myerson1981}). Likewise, \Cref{fig:reduced-form-c} corresponds to the \textit{\textbf{sequential posted price}} mechanism, which offers bidder 1 a take-it-or-leave-it price $p_1$, and then offers bidder 2 a take-it-or-leave-it price $p_2$ if bidder 1 decides not to buy, and is shown by \citet{gershkov2021theory} to be a solution of some optimal auction problems with endogenous values. 

\Cref{cor:opt} applies to any standard objectives such as weighted welfare and revenue maximization, as they are all linear in $(q_1, q_2)$. Moreover, for optimal auctions with endogenous valuations such as in \citet{gershkov2021theory}, the objective is a convex functional because of the investment decisions by the bidders, and hence \Cref{cor:opt} also applies. To illustrate, consider the following example from \citet{gershkov2021theory}:
\begin{ex}[Auction with investments]\label{ex:auction}
There are two symmetric bidders. Each bidder $i$ can take an investment decision $a_i \in \R_+$ at a cost $c(a_i)$ prior to participating in the mechanism. After making investment $a_i$, a type-$\theta_i$ bidder would have value $v(a_i, \theta_i)= a_i + \theta_i $ for the object. As shown in  \citet{gershkov2021theory}, we may treat each bidder as having a convex utility in reduced-form allocation probabilities: 
$h(q, \theta_i) = q\theta_i + w(q)$, where $w(q):= \max_{a}\{qa - c(a)\}$ is a convex function. The seller's expected revenue is then 
\[\sum_i  \int q_i(\theta_i)\psi(\theta_i) + w(q_i(\theta_i)) G(\d \theta_i)\,,\]
where $\psi(\theta_i)$ is the usual virtual value function. The objective is a convex functional in the reduced-form allocation rules $(q_1, q_2)$, and hence there always exists an optimal mechanism taking the form identified in \Cref{cor:opt}. \qed
\end{ex}

To derive optimal mechanisms, \citet{gershkov2021theory} restrict attention to symmetric mechanisms, which, as they show, can be with loss even in a completely symmetric environment such as \Cref{ex:auction}.\footnote{Indeed, for $c(a) = b \cdot \frac{a^2}{2}$ where $b < 0.5$, and uniform type distribution on $[0, 1]$, \citet{gershkov2021theory} show that the optimal mechanism for \Cref{ex:auction} must be \textit{asymmetric}.}  \citet{gershkov2021theory} then provide additional conditions in the case of two bidders to ensure the optimal mechanism is symmetric. Here, we make no symmetry restriction and show that the optimum can always be achieved via a mechanism that is described by an up-set and two truncation thresholds $k_1$ and $k_2$.  Indeed, both the mechanism given by \Cref{fig:reduced-form-b} (i.e., the efficient allocation above a symmetric threshold) and the mechanism given by \Cref{fig:reduced-form-c} (i.e., the sequential posted price mechanism) can be optimal in the setting of \citet{gershkov2021theory}. In fact, all examples of optimal mechanisms in \citet{gershkov2021theory}, symmetric or not, have this structure, which as \Cref{prop:reduced} and \Cref{cor:opt} show is not a coincidence.

\subsection{Optimal Private Private Information}\label{subsec:private}

Consider a binary state $\omega \in \Omega:=\{0, 1\}$. Let $p := \P(\omega = 1) \in (0,1)$ be the prior of the state being $1$. Consider a collection of $n$ signals $(s_1,\dots, s_n) \in S_1\times \cdots \times S_n$ about the state $\omega$, where $S_i$'s are some measurable spaces. Formally, $(\omega, s_1, \dots, s_n)$ are random variables defined on a common probability space. Following \citet*{he2024private}, we call $\mathcal{I}:= (\omega, s_1, \dots, s_n)$ an \textit{\textbf{information structure}}, and say that $\mathcal{I}$ is a \textit{\textbf{private private information structure}} if $(s_1, \dots, s_n)$ are independent random variables. 

We are interested in providing a convex hull characterization of private private information structures. As shown in \citet*{he2024private}, it is without loss of generality to assume that each $s_i$ is distributed uniformly on $[0, 1]$ and represent a private private information structure by a function $f:[0, 1]^n \rightarrow [0, 1]$ such that $\int f(s) \d s = p$. In particular, $f(s)$ represents $\P(\omega = 1 \mid s)$.

Let $p_i(s_i) := \P(\omega = 1 \mid s_i)$ be the belief of the state being $1$ given signal $s_i$. A \textit{\textbf{belief distribution}} $(\mu_1, \dots, \mu_n) \in \Delta([0, 1])\times \cdots \times \Delta([0, 1])$ is \textit{\textbf{feasible under the constraint of privacy}} if there exists a private private information structure such that $\mu_i$ is the distribution of $p_i(s_i)$ for all $i$.

Let $\mathcal{F}^p$ be the set of functions $f: [0, 1]^n \rightarrow [0, 1]$  such that $\int f(s) \d s= p$. We say that a signal structure $f\in \mathcal{F}^p$ is a \textit{\textbf{nested bi-upset}} signal if 
\[f=  \1_{A_1} + \lambda \1_{A_2\backslash A_1} \text{ for two nested up-sets $A_1 \subset A_2$ and some $\lambda \in [0, 1]$}\,.\]
Note that for any $A_1 \subset A_2$, there exists a unique $\lambda^\star$ such that $\int \1_{A_1}(s) + \lambda^\star \1_{A_2\backslash A_1}(s) \d s = p$. Thus, a nested bi-upset signal structure is pinned down by the two nested upsets $A_1 \subset A_2$. Any such signal structure can be implemented by \textit{(i)} randomizing between $A_1$ and $A_2$ with probabilities $1-\lambda^\star$ and $\lambda^\star$, and \textit{(ii)} drawing signals $s$ uniformly from the up-set $A_i$ if the state is $1$, and drawing signals $s$ uniformly from the complement $A^c_i$ if the state is $0$.

Let $\mathcal{F}^\star \subset \mathcal{F}^p$ be the set of nested bi-upset signals. Let 
\[\mathcal{Q}^\star:= \Big\{(q_1,\dots,q_n): \text{ there exists $f \in \mathcal{F}^\star$ such that } q_i(s_i) =  \int f(s_i, s_{-i}) \d s_{-i} \text{ for all $i$ } \Big\}\]
be the set of one-dimensional marginals generated by such signals. 

An immediate consequence of our results is the following  convex hull characterization of the feasible belief quantile functions under private private information structures: 
\begin{prop}\label{prop:convex}
A belief distribution $(\mu_1, \dots, \mu_n)$ is feasible under the constraint of privacy if and only if its associated quantile functions are in the closed convex hull of $\mathcal{Q}^\star$, i.e., 
$(q_1, \dots, q_n) \in \mathrm{cl}({\mathrm{co}}(\mathcal{Q}^\star))$, 
with the set of extreme points in $\mathcal{Q}^\star$.
\end{prop}

\begin{proof}[Proof of \Cref{prop:convex}]
Note that the set of feasible quantile functions $(q_1,\dots, q_n)$ are exactly the following ones:\footnote{To see this, note that the belief quantile function $q_i(s_i)$ is simply the monotone rearrangement of $p_i(s_i)$.} 
\[\widehat{\mathcal{Q}}:=\Big\{(q_1,\dots,q_n): \text{$q_i$'s are monotone and rationalizable} \text{, and } \int q_1(s_1)\d s_1 = p\Big\}\,,\]
where the additional constraint $\int q_1(s_1)\d s_1 = p$ coupled with rationalizability ensures that any joint function $f$ rationalizing $q$ satisfies $\int f(s)\d s= p$. By \Cref{thm:extreme-are-up-sets}, it follows immediately that $\mathrm{ext}(\widehat{\mathcal{Q}}) \subseteq \mathcal{Q}^\star$. It remains to show $\widehat{\mathcal{Q}} = \mathrm{cl}(\mathrm{co}(\mathcal{Q}^\star))$. By the Krein-Milman theorem, we immediately have  
$\widehat{\mathcal{Q}} =  \mathrm{cl}(\mathrm{co}(\mathrm{ext}(\widehat{\mathcal{Q}}))) \subseteq  \mathrm{cl}(\mathrm{co}(\mathcal{Q}^\star)) \subseteq \widehat{Q}$, and hence they are all equal, concluding the proof. 
\end{proof}

Complementary to \citet*{he2024private} who provide a garbling characterization of the feasible belief \textit{distributions} under private private information, \Cref{prop:convex} provides a convex hull characterization of the feasible belief \textit{quantiles}. Unlike the feasible belief distributions, the set of feasible belief quantiles is convex and admits extreme points, which \Cref{prop:convex} characterizes.  An immediate consequence of our convex hull characterization is that the optimal private private information structure can be attained using nested bi-upset signals, whenever the objective is a quasiconvex functional of the feasible belief quantiles:  
\begin{cor}\label{cor:bi-upset}
 For any objective $\Phi(q)$ where $\Phi$ is an upper semi-continuous and quasiconvex functional, an optimal private private information structure is given by a nested bi-upset signal. 
\end{cor}
An example of a convex objective in the quantile space is maximizing the weighted welfare of multiple receivers, each of whom has a decision problem to solve. Indeed, \citet*{he2024private} show that in such cases, the optimal private private information would be on the Blackwell-Pareto frontier and can be implemented using a single up-set. \Cref{cor:bi-upset} shows that more generally, for non-welfare-maximizing objectives, even though the optimal signal structure may not be on the Blackwell-Pareto frontier, they can still be implemented by randomizing over two nested up-sets.

In the special case of $n=2$, recall that our extreme point characterization also shows that the two nested up-sets can only differ by a \textit{rectangle} $[\underline{s}_1,\overline{s}_1]\times [\underline{s}_2,\overline{s}_2]$ (\Cref{thm:rectangle}). This characterization turns out to provide an even simpler implementation of the optimal information structure. In particular, any extreme point of the feasible belief quantiles can be implemented by \textit{(i)} drawing signals uniformly from some up-set $A^\star$ if the state is $1$ and from its complement $[0,1]^2 \backslash A^\star$ if the state is $0$, and \textit{(ii)} for one of the two dimensions, pooling the signals in an interval into a single signal. For example, we may take $A^\star$ to be the ``horizontal average'' of the two nested up-sets, and pool the signals $s_2 \in [\underline{s}_2,\overline{s}_2]$ into a single signal. 

Combining with the characterization of the Blackwell-Pareto frontier from \citet*{he2024private}, our result shows that the extreme belief quantiles differ from the ones resulting from the Blackwell-Pareto undominated information structure only by a single pooling interval: 
\begin{prop}\label{prop:two-info}
 Every extreme belief quantile pair $(q_1, q_2)$ under private private information can be implemented by \textit{(i)} selecting a Blackwell-Pareto undominated information structure and \textit{(ii)} for either $i = 1$ or $i=2$, pooling the signals $s_i$ in an interval $I_i$ into a single signal. Thus, for any upper semi-continuous and quasiconvex objective $\Phi(q_1, q_2)$, it is optimal to use a Blackwell-Pareto undominated information structure with a single pooling interval. 
 \end{prop}

 \begin{proof}[Proof of \Cref{prop:two-info}]
 By the proof of \Cref{prop:convex}, any extreme belief quantile pair can be rationalized by 
 \[f=  \1_{A_1} + \lambda \1_{A_2\backslash A_1}\,, \]
where $A_1 \subset A_2$ are two nested up-sets and $\lambda \in [0, 1]$ is pinned down given $A_1, A_2$ and prior $p$. By \Cref{thm:rectangle}, we also know that $A_2 \backslash A_1$ must be a rectangle $I_1 \times I_2$ where $I_i$ is a closed interval on dimension $i$. Let $g_i$ be the lower nonincreasing boundary of $A_i$. Note that $g_1 \geq g_2$. Let $g^\star = (1-\lambda) g_1 + \lambda g_2$. Let $A^\star:= \{(s_1, s_2): s_2 \geq g^\star(s_1)\}$ be the up-set defined by $g^\star$. Note that 
\[\int \1_{A^\star}(s)\d s = \int \1_{A_1}(s) + \lambda \1_{A_2\backslash A_1}(s) \d s = p \,,\]
and hence $\1_{A^\star}$ defines a feasible private private information structure. Now consider the information structure defined by \textit{(i)} drawing signals uniformly from $A^\star$ if the state is $1$ and uniformly from its complement $[0, 1]^2 \backslash A^\star$ otherwise, and \textit{(ii)} pooling the signals $s_2$ in the interval $I_2$ into a single signal. 

Note that this information structure generates the same belief distribution in each dimension $i$. Moreover, since the pooling in dimension $2$ is only a function of $s_2$, the resulting signals are still independent of each other. Thus, it is still a private private information structure. By Theorem 1 of \citet*{he2024private}, a private private information structure is Blackwell-Pareto undominated if and only if it can be implemented using an up-set in the sense we describe above, concluding the proof. 
\end{proof}

The characterization given by \Cref{prop:two-info} can also be made \textit{sufficient} for a belief quantile pair to be an extreme point if, after selecting any Blackwell-Pareto undominated information structure, we pool an interval of signals in each dimension into a single signal in a consistent way that results in two nested up-sets differing by at most a rectangle, because Bayes plausibility constraint satisfies the sufficiency condition in \Cref{prop:marginal-sufficiency}.

To illustrate \Cref{prop:two-info}, consider the following example adapted from \citet*{he2024private}:
\begin{ex}[Information design with privacy constraint]
Consider a platform market with two buyers and one seller. The seller has two units of the same object (with cost normalized to be $0$). The buyers have unit demand and have a common value $v(\omega)$ for the object, where $\omega \in \{0, 1\}$ is a random state with a common prior $p=\P(\omega = 1)$. Suppose without loss that $v(1) > v(0)$. The platform observes the realized state $\omega$ and wants to design an information structure that sends signals $s_i$ about $\omega$ to each buyer $i$ while keeping the signals independent of each other (e.g., to prevent full surplus extraction by the seller). Suppose that each unit is sold at some price $k$, where $v(0) < k < v(1)$. The platform wants to maximize the probability of trade (because, e.g., a per-trade commission): 
\[ \sum_i \int  \1_{\{v(1)p_i +v(0)(1-p_i)\geq k\}} \mu_i (\d p_i) \,.\]
Note that $\1_{\{v(1)p_i +v(0)(1-p_i)\geq k\}}$ is not convex in $p_i$, and hence we cannot restrict attention to Blackwell-Pareto undominated information structures as in \citet*{he2024private}. But the objective is quasi-convex in the belief quantiles $(q_1, q_2)$. Thus, \Cref{prop:two-info} applies, and shows that the optimal private private information can be implemented using an up-set and applying a single interval pooling. \qed
\end{ex}

\section{Conclusion}\label{sec:conclusion}
We characterize the extreme points of multidimensional monotone functions from $[0,1]^n$ to $[0,1]$, as well as the extreme points of the set of one-dimensional marginals of these functions. These characterizations lead to new results in various mechanism design and information design problems, including public good provision with interdependent values; interim efficient bilateral trade mechanisms; asymmetric reduced form auctions; as well as private private information. 
As another application, we also uncover a mechanism anti-equivalence theorem for two-agent, two-alternative social choice problems.

\renewcommand*{\bibfont}{\footnotesize}
\setlength\bibsep{0pt}
\bibliographystyle{ecta} 
\bibliography{references}

\newpage
\noindent\title{\centering \LARGE Online Appendix to ``Multidimensional Monotonicity and Economic Applications''}

\[\text{\centering \Large Frank Yang and Kai Hao Yang}\]
\appendix
\section{Omitted Proofs}\label{app:proof}

The proofs are presented in the same order in which they are referenced in the main text, except that we prove \Cref{lem:extreme-projection} first,\footnote{We thank Elliot Lipnowski for providing a simple proof for this lemma.} since it will be used in the proofs of \Cref{thm:nested-up-sets} and \Cref{thm:extreme-are-up-sets}.

\subsection{Proof of \texorpdfstring{\Cref{lem:extreme-projection}}{}}
It suffices to show that for any $y \in \mathrm{ext}(L(X))$, $L^{-1}(y) \cap \mathrm{ext}(X) \neq \emptyset$. Clearly, $L^{-1}(y) \neq \emptyset$ since $y \in L(X)$. Moreover, $L^{-1}(y)$ is compact since $X$ is compact and $L$ is continuous. Since $L$ is affine and $X$ is convex, $L^{-1}(y)$ is also convex. By the Krein-Milman theorem, $\mathrm{ext}(L^{-1}(y)) \neq \emptyset$. Take any $z \in \mathrm{ext}(L^{-1}(y))$. We claim that  $z \in \mathrm{ext}(X)$. To see this, fix any $x,x' \in X$ and $\lambda \in (0,1)$ such that $z=\lambda x+(1-\lambda)x'$. Since $L$ is affine, we have 
\[
y=L(z)=L(\lambda x+(1-\lambda) x')=\lambda L(x)+(1-\lambda) L(x')\,.
\]
Since $y \in \mathrm{ext}(L(X))$, it must be that $L(x)=L(x')=y$, which in turn implies $x,x' \in L^{-1}(y)$. Furthermore, since $z \in \mathrm{ext}(L^{-1}(y))$, it must then be that $x=x'=z$. As a result, $z \in \mathrm{ext}(X)$. 

\subsection{Proof of \texorpdfstring{\Cref{thm:ordered-up-sets}}{}}
The ``if'' part follows immediately as monotonicity is preserved under mixtures. For the ``only if'' part,  
consider any $f \in \F$. For any $x \in [0,1]^n$, write
 \[
 f(x) = \int_0^1 \1\{f(x) \geq 1-r\} \d r\,.
 \]
For any $r \in [0,1]$, let  
\[
A_{r}:= \big\{x \in [0, 1]^n: f(x) \geq 1-r \big\}\,.
\]
Note that since $f$ is monotone, $A_r$ is an up-set for all $r \in [0,1]$. Moreover, $A_r \subseteq A_{r'}$ for all $r<r'$. Let $\mu$ be the uniform distribution on $[0, 1]$. Then, for all $x \in [0,1]^n$, we have 
\[
\int_0^1 \1_{A_{r}}(x) \mu(\d r)=\int_0^1 \1\{ f(x) \geq 1-r\} \d r=f(x)\,.
\]
To prove uniqueness, suppose that there exist a family of nested up-sets $\{A_r\}_{r \in [0,1]}$ and a probability measure $\mu \in \Delta([0,1])$  
\[
f(x)=\int_0^1 \1_{A_r}(x)\mu(\d r)
\]
for all $x \in [0,1]^n$. Let $M^{-1}(r):=\inf\{r' \in [0,1]: \mu([0,r']) \geq r\}$ for all $r \in [0,1]$. Since $\{A_r\}_{r \in [0,1]}$ is nested, $\{A_{M^{-1}(r)}\}_{r \in [0,1]}$ is also nested. Moreover, since $A_{M^{-1}(r)}$ is an up-set, there exists a nondecreasing $g:[0,1]^n \to [0,1]$ such that 
\[
A_{M^{-1}(r)}=\Big\{x: x_1 \geq g(1-r,x_{-1})\Big\}
\]
for all $r \in [0,1]$. Since $\{A_{M^{-1}(r)}\}_{r\in [0,1]}$ is nested, $g(\,\cdot\,,x_{-1})$ is nondecreasing for all $x_{-1} \in [0,1]^{n-1}$. Let $g^{-1}(q,x_{-1}):=\inf\{x_1 \in [0,1]:g(x_1,x_{-1}) \geq q\}$ for all $q \in [0,1]$ and for all $x_{-1} \in [0,1]^{n-1}$. 
As a result, for all $x \in [0,1]^n$, we have 
\begin{align*}
f(x)=\int_0^1 \1_{A_r}(x)\mu(\d r)=&\int_0^1 \1_{A_{M^{-1}(r)}}(x)\d r\\
=& \int_0^1 \1\{x_1 \geq g(1-r,x_{-1})\}\d r\\
=& \int_0^1 \1\{g^{-1}(x_1,x_{-1}) \geq 1-r\} \d r\\
=&g^{-1}(x_1,x_{-1})\,.
\end{align*}
Therefore, for all $r \in [0, 1]$, 
\[
A_{M^{-1}(r)}=\Big\{x: x_1 \geq g(1-r,x_{-1})\Big\} = \Big\{x: g^{-1}(x_1,x_{-1}) \geq 1-r \Big\}= \Big\{x: f(x) \geq 1-r\Big\}\,,
\]
equivalent to our construction. This completes the proof.

\subsection{Proof of \texorpdfstring{\Cref{thm:nested-up-sets}}{}}
Consider any extreme point $\overline{f}$ of $\overline{\F}$. By \Cref{thm:ordered-up-sets}, there exist nested up-sets $\{A_r\}_{r \in [0,1]}$
and a probability distribution $\mu \in \Delta([0,1])$ such that 
\[
\overline{f}(x)=\int_0^1 \1_{A_r}(x)\mu(\d r)\,.
\]
For each $j \in \{1,\ldots,m\}$, let
\[
\rho_j(r):=\int_{[0,1]^n} \phi^j(x)\1_{A_r}(x) \d x
\]
for all $r \in [0,1]$. Then, for each $j \in \{1,\ldots,m\}$, 
\begin{align*}
\int_0^1 \rho_j(r) \mu(\d r)=&\int_0^1 \left(\int_{[0,1]^n}\phi^j(x) \1_{A_r}(x) \d x\right) \mu(\d r)\\
=&\int_{[0,1]^n} \phi^j(x) \left(\int_0^1 \1_{A_r}(x) \mu(\d r) \right)\d x\\
=& \int_{[0,1]^n} \phi^j(x) \overline{f}(x) \d x\\
\leq & \eta^j\,.
\end{align*}

Let $\M$ be the set of probability measures $\tilde{\mu} \in \Delta([0,1])$ such that 
\[
\int_0^1 \rho_j(r) \tilde{\mu}(\d r)\leq \eta^j
\]
for all $j \in \{1,\ldots,m\}$. Note that $\M$ is a compact (under the weak-* topology) and convex set.\footnote{Convexity of $\M$ follows from the linearity of the constraints. To see that $\M$ is compact, note that since $\{A_r\}_{r \in [0,1]}$ are nested up-sets, for any sequence $\{r_k\} \to r \in [0,1]$, $\{\1_{A_{r_k}}\} \to \1_{A_r}$ pointwise almost everywhere. Therefore, by the dominated convergence theorem, $\rho_j$ is continuous for all $j$. As a result, $\tilde{\mu} \mapsto \int_0^1 \rho_j(r) \tilde{\mu}(\d r)$ is continuous under the weak-* topology. Therefore, $\M$ is a closed subset of the compact set of probability measures on $[0,1]$ under the weak-* topology, and hence is compact.} Define a continuous affine map $L:\mathcal{M} \to \F$ by 
\[
L(\tilde{\mu})[x]:=\int_0^1 \1_{A_r}(x) \tilde{\mu}(\d r)\,,
\]
for all $x \in [0,1]^n$. Then, $\mu \in \mathcal{M}$, $\overline{f}=L(\mu)$, and 
\[
\int_{[0,1]^n} L(\tilde{\mu})[x]\phi^j(x) \d x=\int_0^1 \left(\int_{[0,1]^n}\phi^j(x) \1_{A_r}(x)\d x\right)\tilde{\mu}(\d r)=\int_0^1 \rho_j(r) \tilde{\mu}(\d r) \leq \eta^j\,,
\]
for all $\tilde{\mu} \in \M$ and all $j \in \{1,\ldots,m\}$. That is, $L(\M)$ is a convex subset of $\overline{\F}$ that includes $\overline{f}$. Since $\overline{f}$ is an extreme point of $\overline{\F}$, it must also be an extreme point of $L(\M)$. By \Cref{lem:extreme-projection}, $\overline{f}=L(\mu^\star)$, for some $\mu^\star \in \M$ that is an extreme point of $\M$. Moreover, by Proposition 2.1 of \citet{Winkler1988}, $|\mathrm{supp}(\mu^\star)| \leq m+1$. Together, it then follows that $\overline{f}$ is a mixture of at most $m+1$ indicator functions defined on nested up-sets. This completes the proof.

\subsection{Proof of \texorpdfstring{\Cref{prop:necessity-joint}}{}}
Necessity follows immediately from the proof of \Cref{thm:nested-up-sets}. For sufficiency, consider any such $f \in \overline{\F}^\star$. If $k=1$, then $f$ must be an extreme point of $\overline{\F}^\star$ by \Cref{thm:choquet} as $A_1$ is an up-set. Suppose that $k \geq 2$ and that there exists $u:[0,1]^n \to [0,1]$ such that $f+u$ and $f-u$ are in $\overline{\F}^\star$. For each $j \in \{1,\ldots,k-1\}$, let $A_j^\star:=A_{j+1}\backslash A_j$. Since $\{A_j\}_{j=1}^k$ is nested, $\{A_j^\star\}_{j=1}^{k-1}$ is disjoint. Moreover, since $f=\sum_{j=1}^k \lambda_j \1_{A_j}$, $f$ is constant on $A_j^\star$ for all $j \in \{1,\ldots,k-1\}$. Since $f+u$ and $f-u$ are monotone, and since $A_{j} \subseteq A_{j+1}$ are up-sets, $u$ must also be constant on $A_j^\star$ for all $j \in \{1,\ldots,k-1\}$. Let $u_j:=u(x)$ for all $x \in A_j^\star$. Then, for each $j \in \{1,\ldots,m\}$, 
\[
\sum_{l=1}^k u_l \int_{A_l^\star} \phi^j(x) \d x=0\,.
\]
However, since 
\[
\left\{\begin{pmatrix}
\int_{A_1^\star}\phi^1(x) \d x\\
\vdots\\
\int_{A_1^\star}\phi^m(x) \d x
\end{pmatrix},\cdots,\begin{pmatrix}
\int_{A_k^\star}\phi^1(x) \d x\\
\vdots\\
\int_{A_k^\star}\phi^m(x) \d x
\end{pmatrix}\right\}
\]
are linearly independent, it must be that $u_j=0$ for all $j \in \{1,\ldots,k\}$. Moreover, since $f\equiv 0$ on $[0,1]^n \backslash A_k$ and $f\equiv 1$ on $A_1$, and since $f+u$ and $f-u$ are in $\overline{\F}^\star$, $u \equiv 0$ on  $[0,1]^n \backslash \cup_{j=1}^{k-1} A^\star_j$. Together, $u \equiv 0$. Therefore, $f$ must be an extreme point of $\overline{\F}^\star$. 

\subsection{Proof of \texorpdfstring{\Cref{prop:sym-nested-up-sets}}{}}
Note that by the proof of \Cref{thm:ordered-up-sets}, for any symmetric monotone function $f$, there exist nested, symmetric up-sets $\{A_r\}_{r\in [0, 1]}$ and a probability distribution $\mu \in \Delta([0, 1])$ such that 
\[f(x) = \int_{0}^1 \1_{A_r}(x) \mu (\d r)\,,\]
because the sets defined by 
\[A_r := \Big\{x \in [0, 1]^n : f(x) \geq 1 - r\Big\}\]
are symmetric up-sets. The result follows immediately by the same proof of \Cref{thm:nested-up-sets}.

\subsection{Proof of \texorpdfstring{\Cref{thm:extreme-are-up-sets}}{}}
To prove part $(i)$, define an affine map $L:\F \to \cQ$ by
\[
L(f):=\left(\int_{[0,1]^{n-1}} f(x) \d x_{-1}\,,\,\ldots\,,\, \int_{[0,1]^{n-1}} f(x) \d x_{-n}\right)\,.
\]
By \Cref{lem:gutmann-monotone} (Theorem 5 of \citealt{gutmann1991existence}), a tuple $q$ of monotone functions is rationalizable if and only if it can be rationalized by a monotone function $f \in \F$. Thus, $L(\F)=\cQ$.\footnote{Note that since $\F \subseteq L^1([0,1]^n)$, $f \in \F$ is identified by an equivalent class of functions $f' \equiv f$ almost everywhere. Thus, $L(\F) \subseteq \cQ$, since for any monotone function $f'$, there exists $f$ whose marginals are left-continuous and $f \equiv f'$ almost everywhere. Together with $\cQ \subseteq L(\F)$, we have $L(\F)=\cQ$.} Moreover, by the dominated convergence theorem, $L$ is continuous. Together, since $\F$ is compact and convex, part $(i)$ follows by invoking \Cref{thm:choquet} and \Cref{lem:extreme-projection}. 

To prove part $(ii)$, let 
\[
\overline{\F}:=\left\{f \in \F: \int_{[0,1]^n} f(x) \sum_{i=1}^n \psi_i^j(x_i) \d x\leq \eta^j\,\,\, \text{ for all $j$}\right\}\,.
\]
By \Cref{lem:gutmann-monotone} (Theorem 5 of \citealt{gutmann1991existence}) again, any $q \in \overline{\cQ}$ is rationalized by some $f \in \overline{\F}$, and hence $L(\overline{\F})=\overline{\cQ}$. Moreover, by the dominated convergence theorem, $L$ is continuous. Lastly, since $\overline{\F} \subseteq \F$ is closed and since $\F$ is compact, $\overline{\F}$ is compact as well. Part $(ii)$ then follows by invoking \Cref{thm:nested-up-sets} and \Cref{lem:extreme-projection}. This completes the proof.

\subsection{Proof of \texorpdfstring{\Cref{thm:rationalized-upsets}}{}}

Necessity follows immediately from \Cref{thm:extreme-are-up-sets}. For sufficiency, fix any $q \in \cQ$ that is rationalized by $\1_A$ for some up-set $A \subseteq [0,1]^2$.

We prove a stronger claim that $q$ is an exposed point of $\mathcal{Q}$ and hence an extreme point. Since $A$ is an up-set, there exists a nonincreasing function $g: [0,1] \to [0,1]$ such that $A= \{(x_1,x_2):x_2 \geq g(x_1)\}$. Let $\phi_1(x_1) := -g(x_1)$ and $\phi_2(x_2) := x_2$. Then note that 
\[\phi_1(x_1) + \phi_2(x_2) \begin{cases}
>0,&\mbox{if }(x_1, x_2) \in A \backslash \{(x_1, g(x_1))\}_{x_1 \in [0, 1]}\\
=0,&\mbox{if } (x_1, x_2) \in \{(x_1, g(x_1))\}_{x_1 \in [0, 1]}\\
<0,&\mbox{otherwise}
\end{cases}\,.
\]
We claim that the optimization problem
\begin{equation}\label{eq:unique}
\max_{f\in \mathcal{F}} \int_{[0,1]^2} (\phi_1(x_1) + \phi_2(x_2)) f(x_1, x_2) \d x\,
\end{equation}
has a unique solution given by $\1_{A}$. To see this, note that we can relax the monotonicity constraint imposed by $\mathcal{F}$, and show that $\1_A$ is still the unique maximizer. Indeed, $\1_A$ can clearly attain the upper bound given by $ \int \max\{\phi_1(x_1) + \phi_2(x_2), 0\} \d x$. Moreover, if any function $f:[0, 1]^2 \rightarrow [0, 1]$ differs from $\1_{A}$ on a positive measure of $x$, then $f$ must have a strictly lower objective by the construction of $\phi_1$ and $\phi_2$. This proves that $\1_A$ must be the unique maximizer in $\mathcal{F}$ for the above linear functional on $\mathcal{F}$.\footnote{Note that this also shows that any extreme point of $\mathcal{F}$ is exposed.}

Now consider the linear functional $\Phi:\cQ \to \R$  defined by
\[
\Phi(\tilde{q}):=\int_0^1 \phi_1(x_1) \tilde{q}_1(x_1) \d x_1 + \int_0^1 \phi_2(x_2)\tilde{q}_2(x_2) \d x_2\,.
\]
For any $\tilde{q} \in \mathcal{Q}$, let $\tilde{f} \in \mathcal{F}$ be a monotone function that rationalizes $\tilde{q}$ (recall \Cref{lem:gutmann-monotone}). Then, we have 
\begin{align*}
    \int_0^1 \phi_1(x_1) \tilde{q}_1(x_1) \d x_1 + \int_0^1 \phi(x_2)\tilde{q}_2(x_2) \d x_2 &= \int_0^1 (\phi_1(x_1) + \phi_2(x_2)) \tilde{f}(x_1, x_2) \d x \\
    &\leq \int_0^1 (\phi_1(x_1) + \phi_2(x_2)) \1_{A}(x) \d x\,\\
    &= \int_0^1 \phi_1(x_1) q_1(x_1) \d x_1 + \int_0^1 \phi(x_2)q_2(x_2) \d x_2\,,
\end{align*}
where the inequality is due to the optimality of $\1_A$ in the previous claim. Moreover, for any $\tilde{q} \neq q$, the function $\tilde{f}$ cannot equal $\1_A$, which means that the above inequality is strict since $\1_A$ is the unique maximizer for the above objective by the previous claim. This proves that $q$ is an exposed point and hence an extreme point of $\cQ$. 

For unique rationalizability, suppose that $q$ is rationalized by $\1_A$ and some other (not necessarily monotone) function $\tilde{f}$. Suppose for contradiction that $\tilde{f} \neq \1_A$ on a positive measure of $x$. Then, since $\1_A$ is the unique solution of  \eqref{eq:unique} even if we relax the monotonicity constraint imposed by $\mathcal{F}$, we must have 
\begin{align*}
\Phi(q)=&\int_0^1 \phi_1(x_1)q_1(x_1)\d x_1+\int_0^1 \phi_2(x_2)q_2(x_2)\d x_2\\
=& \int_{[0,1]^2} (\phi_1(x_1)+\phi_2(x_2))\1_A(x)\d x\\
>&\int_{[0,1]^2} (\phi_1(x_1)+\phi_2(x_2))\tilde{f}(x_1,x_2)\d x\\
=& \int_0^1 \phi_1(x_1)q_1(x_1)\d x_1+\int_0^1 \phi_2(x_2)q_2(x_2)\d x_2=\Phi(q)\,,
\end{align*}
a contradiction. Thus, it must be that $\tilde{f}\equiv \1_A$ almost everywhere. This completes the proof.

\subsection{Proof of \texorpdfstring{\Cref{thm:rectangle}}{}}

\Cref{subsec:rationalize} proves the first part of \Cref{thm:rectangle}; \Cref{subsec:unique} proves the second part of \Cref{thm:rectangle}. 

\subsubsection{Rationalizability by Nested Upsets Differing by a Rectangle}\label{subsec:rationalize}
Consider any extreme point $q$ of $\overline{\cQ}$. Since $q \in \overline{\cQ}$ is an extreme point of $\overline{\cQ}$, by \Cref{thm:extreme-are-up-sets}, there exist nested up-sets $A_1 \subseteq A_2$ such that $q$ is rationalized by $f=\lambda\1_{A_1}+(1-\lambda)\1_{A_2}$, for some $\lambda \in (0,1)$. We now show that $A_1$ and $A_2$ can be chosen to differ by at most a rectangle (and in fact they must differ by at most a rectangle from our unique rationalizability result which we prove in \Cref{subsec:unique}). For each $j \in \{1,2\}$, since $A_j$ is an up-set, there exists a nonincreasing left-continuous function $g^j:[0,1]\to [0,1]$ such that $A_j=\{(x_1,x_2) \in [0,1]^2: x_2 \geq g^j(x_1)\}$. Moreover, since $A_1 \subseteq A_2$, $g^1(z) \geq g^2(z)$ for all $z \in [0,1]$. Let $\gamma^j(z):=1-g^j(z)$ for all $z \in [0,1]$ and $j \in \{1,2\}$. Then, for all $z \in [0,1]$,
\begin{equation}\label{eq:q1nested}
q_1(z)=\int_{0}^1 [\lambda \1_{A_1}(z,x_2)+(1-\lambda) \1_{A_2}(z,x_2)]\d x_2=\lambda \gamma^1(z)+(1-\lambda) \gamma^2(z)
\end{equation}
and 
\begin{equation}\label{eq:q2nested}
q_2(z)=\int_0^1[\lambda \1_{A_1}(x_1,z)+ (1-\lambda) \1_{A_2}(x_1,z)]\d x_1=\lambda \hat{\gamma}^1(z)+(1-\lambda)\hat{\gamma}^2(z)\,,
\end{equation}
where $\hat{\gamma}^j(z):=1-(\gamma^j)^{-1}(1-z)$ is the conjugate of $\gamma^j$, for all $j \in \{1,2\}$. 

We first claim that $q_1$ must be an extreme point of the set of monotone functions that are majorized by $\hat{q}_2$ and satisfy the linear constraint (fixing $q_2$). Indeed, suppose not. Then there exist distinct monotone functions $q_1^1,q_1^2$ such that $q_1^l \preceq \hat{q}_2$ for $l \in \{1,2\}$, $\lambda' q_1^1+(1-\lambda')q_1^2=q_1$ for some $\lambda' \in (0,1)$, and that 
\[
\int_0^1 q_1^l(x_1)\psi_1(x_1)\d x_1+\int_0^1 q_2(x_2) \psi_2(x_2) \d x_2\leq \eta\,,
\]
for all $l \in \{1,2\}$. Since $q_1^l \preceq \hat{q}_2$ for $l \in \{1,2\}$, by \Cref{lem:gutmann}, there exist $f^1,f^2 \in \F$ such that $(q_1^l,q_2)$ is rationalized by $f^l$ for $l \in \{1,2\}$. Therefore, there exist distinct $q^1:=(q_1^1,q_2),\,q^2:=(q_1^2,q_2) \in \overline{\cQ}$ such that $\lambda' q^1+(1-\lambda')q^2=q$, and hence $q$ is not an extreme point of $\overline{\cQ}$. 

As a result, by Theorem 1 of \citet{nikzad2023constrained}, there exists at most a countable collection of disjoint intervals $\{(\underline{z}_k,\overline{z}_k]\}_{k=1}^\infty$ such that $q_1\equiv \hat{q}_2$ on $[0,1] \backslash \cup_{k=1}^\infty (\underline{z}_k,\overline{z}_k]$, and $q_1$ is a step function with at most two steps on $(\underline{z}_k,\overline{z}_k]$ for all $k \in \N$. Note that this collection can be empty but then our result would follow immediately since such a pair can be rationalized by a single up-set as shown in the proof of \Cref{prop:reverse-majorization}. From now on, we assume that the collection is not empty. We order the intervals from left to right indexed by $k$. Together with \eqref{eq:q1nested}, it then follows that, since $\gamma^1$ and $\gamma^2$ are nondecreasing and $q_1=\lambda \gamma^1+(1-\lambda)\gamma^2$, we must have $\gamma^1,\gamma^2$ as step functions with at most two steps on $(\underline{z}_k,\overline{z}_k]$, for all $k \in \N$. In particular, for each $k \in \N$, there exists $y_k \in (\underline{z}_k,\overline{z}_k)$ such that 
\[
\gamma^j(z)=\begin{cases}
\underline{\gamma}_k^j,&\mbox{if } z \in (\underline{z}_k,y_k]\\
\overline{\gamma}_k^j,&\mbox{if } z \in (y_k,\overline{z}_k]
\end{cases}\,,
\]
for all $j \in \{1,2\}$, where $\overline{\gamma}_k^j \geq \underline{\gamma}_k^j$ for $j \in \{1,2\}$, $\overline{\gamma}_k^1 \leq \overline{\gamma}_k^2$, and $\underline{\gamma}_k^1 \leq \underline{\gamma}_k^2$. Clearly, we have 
\[
q_1(z)=\lambda \gamma^1(z)+(1-\lambda) \gamma^2(z)=\begin{cases}
\lambda \underline{\gamma}_k^1+(1-\lambda)\underline{\gamma}_k^2,&\mbox{if } z \in (\underline{z}_k,y_k]\\
\lambda\overline{\gamma}_k^1+(1-\lambda) \overline{\gamma}_k^2,&\mbox{if } z \in (y_k,\overline{z}_k]
\end{cases}\,.
\]

Next, we show that, for each $k \in \N$, if $\overline{\gamma}^1_k<\overline{\gamma}^2_k$ and $\underline{\gamma}^1_k<\underline{\gamma}^2_k$, then $\overline{\gamma}^j_k=\underline{\gamma}^j_k$ for $j \in \{1,2\}$. We first prove this claim in the case of $\underline{\gamma}^2_k > \overline{\gamma}^1_k$. Suppose for contradiction that there exists $k \in \N$ such that $\overline{\gamma}^1_k<\overline{\gamma}^2_k$, $\underline{\gamma}^1_k<\underline{\gamma}^2_k$, and $\overline{\gamma}^j_k>\underline{\gamma}^j_k$ for some $j \in \{1,2\}$. Suppose that $\overline{\gamma}^1_k>\underline{\gamma}^1_k$.
The case of $\overline{\gamma}^2_k>\underline{\gamma}^2_k$ while $\overline{\gamma}^1_k=\underline{\gamma}^1_k$ can be treated by a similar argument. Let $B:=(\underline{z}_k,y_k] \times (1-\overline{\gamma}^1_k,1-\underline{\gamma}^1_k]$, and $C:=(y_k,\overline{z}_k] \times (1-\underline{\gamma}^2_k,1-\overline{\gamma}^1_k]$. Note that by assumption, both regions $B$ and $C$ have positive measures. By \citet{nikzad2023constrained}, we also know that $q_1(\underline{z}_k) < q_1(\underline{z}_k^+)$ and $q_1(\overline{z}_k) < q_1(\overline{z}_k^+)$. Thus, 
\[q_1(\underline{z}_k) < \lambda \underline{\gamma}_k^1+(1-\lambda)\underline{\gamma}_k^2 < \lambda\overline{\gamma}_k^1+(1-\lambda) \overline{\gamma}_k^2 <  q_1(\overline{z}_k^+)\]
is a strictly increasing sequence of four numbers. Let $\delta_1 > 0$ be the minimum of the consecutive differences. 

We claim that 
\begin{equation}\label{eq:highest-jump}
 q_2(1-\underline{\gamma}^1_k) < q_2((1-\underline{\gamma}^1_k)^{+})   
\end{equation}
\begin{equation}\label{eq:middle-jump}
q_2(1-\overline{\gamma}^1_k) < q_2((1-\overline{\gamma}^1_k)^{+}) 
\end{equation}
\begin{equation}\label{eq:lowest-jump}
q_2(1-\underline{\gamma}^2_k) < q_2((1-\underline{\gamma}^2_k)^{+}) 
\end{equation}

To see \eqref{eq:highest-jump}, note that if we let $\tilde{A}_1:= A_1 \cup B$, then $\tilde{A}_1 \subseteq A_2$ are still nested up-sets and projecting $\lambda \1_{\tilde{A}_1} + (1-\lambda) \1_{A_2}$ to dimension $2$ would result in a monotone $\tilde{q}_2$ that coincides with $q_2$ on $(1-\underline{\gamma}^1_k, 1]$ with $\tilde{q}_2(1-\underline{\gamma}^1_k) > q_2(1-\underline{\gamma}^1_k)$. Thus,
\[q_2((1-\underline{\gamma}^1_k)^+) = \tilde{q}_2((1-\underline{\gamma}^1_k)^+) \geq \tilde{q}_2(1-\underline{\gamma}^1_k) > q_2(1-\underline{\gamma}^1_k)\,.\]
To see \eqref{eq:middle-jump}, apply the same argument with $\tilde{A}_1 := A_1 \cup (y_k,1] \times (1-\underline{\gamma}^2_k,1-\overline{\gamma}^1_k]$. To see \eqref{eq:lowest-jump}, apply the same argument but now keep $A_1$ as before and let $\tilde{A}_2 := A_2 \cup (\underline{z}_k, 1] \times (0, 1 - \underline{\gamma}^2_k]$. In particular, note that $A_1 \subseteq \tilde{A}_2$ are still two nested up-sets, and projecting $\lambda \1_{A_1} + (1-\lambda) \1_{\tilde{A}_2}$ to dimension $2$ would result in a monotone $\tilde{q}_2$ that coincides with $q_2$ on $(1-\underline{\gamma}^2_k, 1]$ with $\tilde{q}_2(1-\underline{\gamma}^2_k) > q_2(1-\underline{\gamma}^2_k)$. Now, let 
\[\delta_2:=   \min \Big\{q_2((1-\underline{\gamma}^1_k)^{+}) - q_2(1-\underline{\gamma}^1_k),\,\, q_2((1-\overline{\gamma}^1_k)^{+})  - q_2(1-\overline{\gamma}^1_k),\,\,  q_2((1-\underline{\gamma}^2_k)^{+})  - q_2(1-\underline{\gamma}^2_k) \Big\}\,. \]
By \eqref{eq:highest-jump}, \eqref{eq:middle-jump}, \eqref{eq:lowest-jump}, we have $\delta_2 > 0$.

Take $\varepsilon_B,\,\varepsilon_C \in \mathbb{R}$ 
with at least one of $\varepsilon_B,\,\varepsilon_C$ being non-zero such that 
\begin{equation}\label{eq:epsilon}
\varepsilon_B \int_B (\psi_1(x_1)+\psi_2(x_2)) \d x=\varepsilon_C \int_C (\psi_1(x_1)+\psi_2(x_2)) \d x\,, 
\end{equation}
and 
\begin{align}\label{eq:small-enough}
\max\Big\{|\varepsilon_B|,\,|\varepsilon_C|\Big\}<\frac{1}{2}\cdot \min \Big\{\delta_1,\, \delta_2\Big\}\,.
\end{align}
Now, let 
\[
f'(x):=\begin{cases}
f(x)+\varepsilon_B,&\mbox{if } x \in B\\
f(x)-\varepsilon_C,&\mbox{if } x \in C\\
f(x),&\mbox{otherwise}
\end{cases}; \quad \mbox{ and } \quad 
f''(x):=\begin{cases}
f(x)-\varepsilon_B,&\mbox{if } x \in B\\
f(x)+\varepsilon_C,&\mbox{if } x \in C\\
f(x), &\mbox{otherwise}.
\end{cases}
\]
By construction, $\frac{1}{2}f'+\frac{1}{2}f''=f$, and hence $q=\frac{1}{2} q'+\frac{1}{2}q''$, where $q'$ and $q''$ are the one-dimensional marginals of $f'$ and $f''$, respectively. 

Since the rectangles $B, C \subseteq [0, 1]^2$ both have positive measures and do not overlap when projecting to either dimension $1$ or dimension $2$, we must have  $q' \neq q''$. Moreover, by \eqref{eq:epsilon} and \eqref{eq:small-enough}, $q_i'$ and $q_i''$ are nondecreasing for $i \in \{1,2\}$, and 
\begin{align*}
&\int_0^1 q_1'(x_1)\psi_1(x_1)\d x_1+\int_0^1 q_2' \psi_2(x_2) \d x_2\\
=& \int_{[0,1]^2} f'(x_1,x_2)(\psi_1(x_1)+\psi_2(x_2))\d x\\
=& \int_{[0,1]^2} f(x_1,x_2)(\psi_1(x_1)+\psi_2(x_2))\d x+ \varepsilon_B \int_B (\psi_1(x_1)+\psi_2(x_2))\d x-\varepsilon_C \int_C (\psi_1(x_1)+\psi_2(x_2))\d x\\
=&\int_{[0,1]^2} f(x_1,x_2)(\psi_1(x_1)+\psi_2(x_2))\d x\\
=& \int_0^1 q_1(x_1) \psi_1(x_1)\d x_1+\int_0^1 q_2(x_2) \psi_2(x_2)\d x_2\\
\leq & \eta\,.
\end{align*}
and 
\begin{align*}
&\int_0^1 q_1''(x_1)\psi_1(x_1)\d x_1+\int_0^1 q_2''(x_2) \psi_2(x_2) \d x_2\\
=& \int_{[0,1]^2} f''(x_1,x_2)(\psi_1(x_1)+\psi_2(x_2))\d x\\
=& \int_{[0,1]^2} f(x_1,x_2)(\psi_1(x_1)+\psi_2(x_2))\d x- \varepsilon_B \int_B (\psi_1(x_1)+\psi_2(x_2))\d x+\varepsilon_C \int_C (\psi_1(x_1)+\psi_2(x_2))\d x\\
=&\int_{[0,1]^2} f(x_1,x_2)(\psi_1(x_1)+\psi_2(x_2))\d x\\
=& \int_0^1 q_1(x_1) \psi_1(x_1)\d x_1+\int_0^1 q_2(x_2) \psi_2(x_2)\d x_2\\
\leq & \eta\,.
\end{align*}
Together, $q',q'' \in \overline{\cQ}$, $q' \neq q''$ and $\frac{1}{2}q'+\frac{1}{2}q''=q$. Therefore, $q$ is not an extreme point of $\overline{\cQ}$, a contradiction. As a result, whenever $\overline{\gamma}^1_k<\overline{\gamma}^2_k$ and $\underline{\gamma}^1_k<\underline{\gamma}^2_k$, it must be that $\overline{\gamma}^j_k=\underline{\gamma}^j_k$ for $j \in \{1,2\}$. 

Now, for the case of $\underline{\gamma}^2_k \leq  \overline{\gamma}^1_k$, apply the same argument as above, except now letting $B := (\underline{z}_k, y_k] \times (1 - \underline{\gamma}^2_k, 1-\underline{\gamma}^1_k]$ and $C := (y_k, \overline{z}_k] \times (1 - \overline{\gamma}^2_k, 1-\overline{\gamma}^1_k]$. In particular, by the same argument as before, we still have 
\[q_1(\underline{z}_k) < \lambda \underline{\gamma}_k^1+(1-\lambda)\underline{\gamma}_k^2 < \lambda\overline{\gamma}_k^1+(1-\lambda) \overline{\gamma}_k^2 <  q_1(\overline{z}_k^+)\,.\]
Moreover, by the same arguments, we also have 
\begin{equation}\label{eq:highest-jump-square}
 q_2(1-\underline{\gamma}^1_k) < q_2((1-\underline{\gamma}^1_k)^{+})   
\end{equation}
\begin{equation}\label{eq:middle-jump-square}
 q_2(1-\underline{\gamma}^2_k) < q_2((1-\underline{\gamma}^2_k)^{+})   
\end{equation}
\begin{equation}\label{eq:middle-2-jump-square}
 q_2(1-\overline{\gamma}^1_k) < q_2((1-\overline{\gamma}^1_k)^{+})   
\end{equation}
\begin{equation}\label{eq:lowest-jump-square}
 q_2(1-\overline{\gamma}^2_k) < q_2((1-\overline{\gamma}^2_k)^{+})   
\end{equation}  
Thus, we can repeat the previous perturbation argument and conclude that whenever $\overline{\gamma}^1_k<\overline{\gamma}^2_k$ and $\underline{\gamma}^1_k<\underline{\gamma}^2_k$, it must be that $\overline{\gamma}^j_k=\underline{\gamma}^j_k$ for $j \in \{1,2\}$. 

Together, we know that there exists a countable collection of intervals $\big\{(\underline{z}_k, \overline{z}_k]\big\}_k$ such that $q_1$ is a constant on each of the interval $(\underline{z}_k, \overline{z}_k]$, and $q_1 = \hat{q}_2$ otherwise. Moreover, for each interval $(\underline{z}_k, \overline{z}_k]$, we also know that the majorization inequality holds with strict inequality for all $s \in (\underline{z}_k, \overline{z}_k)$, and that it holds with equality at the two ends of the interval. In particular, the constant $\gamma_k$ is pinned down by 
\[ \gamma_k = \frac{\int_{\underline{z}_k}^{\overline{z}_k} \hat{q}_2(s) \d s}{\overline{z}_k - \underline{z}_k}\,.\]
In fact, by the same two perturbation arguments given above, we also know that the closure of the pooling intervals must be disjoint: for any $k \neq k' \in \N$, 
\[\overline{z}_k \neq \underline{z}_{k'}\,.\]
This implies that we must have\footnote{This is because, for any $k$, there exist sequences $\{\overline{z}^m_k\}$ and $\{\underline{z}^m_k\}$, along which the majorization constraint binds and $\{\overline{z}^m_k\} \downarrow \overline{z}_k$ and $\{\underline{z}^m_k\} \uparrow \underline{z}_k$.} 
\[q_1(\underline{z}_k) = \hat{q}_2(\underline{z}_k)\,,\,\,q_1(\overline{z}^+_k) = \hat{q}_2(\overline{z}^+_k)\,.\]

Now, apply the symmetric argument to $q_2$. In particular, we can equivalently write $\overline{\mathcal{Q}}$ as the set of monotone pairs $(q_1, q_2)$ that satisfy the linear constraint and satisfy 
\[q_2 \preceq \hat{q}_1\,.\]
Using this representation, we can apply exactly the symmetric argument to conclude that there exists a countable collection of disjoint intervals $\big\{(\underline{t}_j, \overline{t}_j]\big\}_j$ such that $q_2$ is a constant on each of the interval $(\underline{t}_j, \overline{t}_j]$, and $q_2 = \hat{q}_1$ otherwise.

Now, fix any pooling interval $(\underline{z}_k, \overline{z}_k]$ of $q_1$. Consider the interval $I^\star_k:=\big(1 - q_1(\overline{z}_k^+), 1 -q_1(\underline{z}_k)\big]$. Note that on this interval, $\hat{q}_1$ must be a step function with the two steps equal to $1 - \overline{z}_k$ and $1- \underline{z}_k$, and a single jump at $1 - \gamma_k$. Note that there must exist at least one pooling interval $(\underline{t}_j, \overline{t}_j]$ of $q_2$ that overlaps with $I^\star_k$, because otherwise $q_2 = \hat{q}_1$ everywhere on $I^\star_k$, which, by definition, means that for all 
$z \in (\underline{z}_k,\overline{z}_k]$, 
\begin{align*}
\hat{q}_2(z)=1-q_2^{-1}(1-z)=&1-\inf\{y \in [0,1]:q_2(y)>1-z\}\\
=&1-\inf\{y \in (1-q_1(\overline{z}_k^+),1]: q_2(y)>1-z\}\\
=&1-\inf\{y \in (1-q_1(\overline{z}_k^+),1]: \hat{q}_1(y)>1-z\}\\
=&1-(1-\gamma_k)=\gamma_k\,.
\end{align*} 
Thus, $q_1=\hat{q}_2$ everywhere on $(\underline{z}_k,\overline{z}_k]$, and hence the majorization constraint binds everywhere on $(\underline{z}_k,\overline{z}_k]$, contradicting the construction of the interval. Moreover, note that any pooling interval $(\underline{t}_j, \overline{t}_j]$ that overlaps with $I^\star_k$ must satisfy $(\underline{t}_j, \overline{t}_j] \subseteq I^\star_k$. To see this, suppose not. Then there exists some pooling interval $(\underline{t}_j, \overline{t}_j]$
such that $(\underline{t}_j, \overline{t}_j] \cap I^\star_k \neq \emptyset$ and  it includes $1 -q_1(\overline{z}_k^+)$ or $(1 -q_1(\underline{z}_k))^{+}$. Note that if it includes $1 -q_1(\overline{z}_k^+)$, then $\underline{t}_j < 1 -q_1(\overline{z}_k^+) < \overline{t}_j$
and hence there must be a jump discontinuity in $\hat{q}_2$ that jumps strictly across the value $q_1(\overline{z}_k^+)$. But then 
\[\hat{q}_2(s) \neq q_1(\overline{z}_k^+)\]
for all $s \in [0, 1]$ by monotonicity of $\hat{q}_2$, which contradicts that $q_1(\overline{z}_k^+) = \hat{q}_2(\overline{z}_k^+)$. Similarly, if it includes $(1 -q_1(\underline{z}_k))^{+}$, then $\overline{t}_j > 1 -q_1(\underline{z}_k) > \underline{t}_j$ and hence there must be a jump discontinuity in $\hat{q}_2$ that jumps strictly across the value $q_1(\underline{z}_k)$. But then 
\[\hat{q}_2(s) \neq q_1(\underline{z}_k)\]
for all $s \in [0, 1]$ by monotonicity of $\hat{q}_2$, which contradicts that $q_1(\underline{z}_k) = \hat{q}_2(\underline{z}_k)$. Therefore, we conclude that any pooling interval $(\underline{t}_j, \overline{t}_j]$ of $q_2$ that overlaps with $I^\star_k$ must satisfy $(\underline{t}_j, \overline{t}_j] \subseteq I^\star_k$. But then, because $\hat{q}_1$ is a step function on the interval $I^\star_k$ with two steps, any such pooling interval $(\underline{t}_j, \overline{t}_j]$ must contain the jump point $1-\gamma_k$ in its interior in order to have the majorization constraint hold with strict inequality in the interior of the interval. Since these intervals are disjoint, there can be at most one interval $(\underline{t}_j, \overline{t}_j]$ that contains $1-\gamma_k$. Together, we conclude that there exists a unique pooling interval $(\underline{t}_j, \overline{t}_j]$ such that $(\underline{t}_j, \overline{t}_j] \cap I^\star_k \neq \emptyset$. Moreover, the pooling interval $(\underline{t}_j, \overline{t}_j]$ is entirely contained in $I^\star_k$ and includes the jump point $1 -\gamma_k$. Let $(\underline{t}^\star_k, \overline{t}^\star_k]$ denote this unique interval associated with $I^\star_k$. On this interval, $q_2$ is a constant $\kappa_k$ that is pinned down by 
\[ \kappa_k = \frac{\int^{\overline{t}^\star_k}_{\underline{t}^\star_k} \hat{q}_1(s) \d s}{\overline{t}^\star_k -\underline{t}^\star_k}\,.\]
Note that the above conclusion also implies that on each interval $(\underline{z}_k, \overline{z}_k]$, $\hat{q}_2$ is a step function with exactly two steps, where one step is strictly below and one step is strict above the constant value of $q_1$, and a jump point at $1 - \kappa_k$.  

We also claim that any pooling interval $(\underline{t}_j, \overline{t}_j]$ of $q_2$ must be contained in the interval $I^\star_k$ for some $k$. To see this, suppose not. Then there exists a pooling interval $(\underline{t}_j, \overline{t}_j]$ such that it is not contained in  $I^\star_k$ for any $k$. But then by the previous argument, it cannot overlap with $I^\star_k$ for any $k$. Thus,  $(\underline{t}_j, \overline{t}_j] \subseteq [0, 1]\backslash \cup_k I^\star_k$. Note that by symmetric arguments, we also know that on the pooling interval $(\underline{t}_j, \overline{t}_j]$ of $q_2$, we must have $\hat{q}_1$ being a step function with exactly two steps, one strictly below and one strictly above the constant value $\beta_j$ of $q_2$, and a single jump point. Let $1 - w_j$ denote the jump point, and let $1-\overline{s}_j<\beta_j$ be the lower step, $1-\underline{s}_j>\beta_j$ be the higher step. Then $q_1(s)=1-(\hat{q}_1)^{-1}(1-s)=w_j$ for all $s \in (\underline{s}_j,\overline{s}_j]$. Note that $(\underline{s}_j, \overline{s}_j] \cap (\underline{z}_k, \overline{z}_k] = \emptyset$ for all $k$, because otherwise there exist some $k$ and some $z \in (\underline{z}_k, \overline{z}_k]$ such that 
\[1 - q_1(z) = 1-  w_j \in (\underline{t}_j, \overline{t}_j]\,,\]
contradicting that $(\underline{t}_j, \overline{t}_j] \cap I^\star_k = \emptyset$. Now, since $(\underline{s}_j, \overline{s}_j] \subseteq [0, 1]\backslash \cup_k (\underline{z}_k, \overline{z}_k]$, we know that $q_1(s) = \hat{q}_2(s)$ for all $s \in (\underline{s}_j, \overline{s}_j]$. However, since $1-\overline{s}_j<\beta_j<1-\underline{s}_j$, for all $s$ such that $1-s \in (1-\overline{s}_j,\,\beta_j)$, we have $\hat{q}_2(s)=1-q_2^{-1}(1-s) \geq 1-\underline{t}_j$. Likewise, for all $s'$ such that $1-s' \in (\beta_j,\,1-\underline{s}_j)$, we have $\hat{q}_2(s')=1-q_2^{-1}(1-s') \leq 1-\overline{t}_j$. Thus, we have $s, s' \in (\underline{s}_j,\overline{s}_j)$, such that 
\[\hat{q}_2(s) \geq 1-\underline{t}_j > 1 - \overline{t}_j \geq \hat{q}_2(s')\,.\]
But we also know that \[\hat{q}_2(s) = q_1(s) = w_j = q_1(s') = \hat{q}_2(s')\,,\] 
a contradiction.

Combining the previous two paragraphs, we can conclude that the pooling intervals $(\underline{t}_j, \overline{t}_j]$ of $q_2$ where the majorization inequality holds with strict inequality in their interiors are exactly the pooling intervals $(\underline{t}^\star_k, \overline{t}^\star_k]$ that we identify. On each interval $(\underline{t}^\star_k, \overline{t}^\star_k]$,  we know that  $q_2$ is a constant, $\hat{q}_1$ is a step function with two steps where the jump point is exactly $1 - q_1(\overline{z}_k)$. Note that the two steps are exactly $1- \overline{z}_k$ and $1 - \underline{z}_k$.

Now we construct a function $f$ that rationalizes $(q_1, q_2)$ and then argue that we can always perturb $f$ to be $f'$ and $f''$ such that their one-dimensional marginals $q'$ and $q''$ satisfy $\frac{1}{2} q' + \frac{1}{2} q'' = q$, unless the function $f$ is a mixture of two indicator functions defined on nested up-sets that differ by a rectangle. 

Let 
\[\tilde{g}^1(z) := 1 - q_1(z)\,, \quad\tilde{g}^2(z) := 1 - q_1(z) \]
for all $z \not\in \bigcup_k (\underline{z}_k, \overline{z}_k]$, and let  
\[\tilde{g}^1(z) := \overline{t}^\star_k\,, \quad \tilde{g}^2(z) := \underline{t}^\star_k\,.\]
for all $z \in (\underline{z}_k, \overline{z}_k]$ and for all $k$. For each $j \in \{1, 2\}$, let  
\[\tilde{A}_j:= \Big\{x \in [0, 1]^2: x_2 \geq \tilde{g}^j(x_1)\Big\}\,.\]
For each $k$, let $\lambda^\star_k \in (0, 1)$ be the unique solution to 
\[\lambda^\star_k (1-\overline{t}^\star_k) + (1-\lambda^\star_k)(1 -\underline{t}^\star_k) = q_1(\overline{z}_k) \,.\]
There always exists such a solution because 
\[ 1 - \overline{t}^\star_k < q_1(\overline{z}_k) <1- \underline{t}^\star_k\,,\]
where the inequalities are due to that the interval $(\overline{t}^\star_k, \underline{t}^\star_k]$ contains $1-\gamma_k$, and $\gamma_k$ is the constant value $q_1$ takes on $(\underline{z}_k, \overline{z}_k]$. For each $k$, let 
\[D_k :=  (\underline{z}_k, \overline{z}_k] \times (\underline{t}^\star_k, \overline{t}^\star_k]\,.\]
Let 
\[f:= \1_{\tilde{A}_1} + \sum_k (1-\lambda^\star_k) \1_{D_k}\,.\]

Clearly, we have $\tilde{g}^1 \geq \tilde{g}^2$ by construction. Moreover, by construction, we also have 
\[1 - q_1(\underline{z}_k) \geq \overline{t}^\star_k > \underline{t}^\star_k \geq 1 - q_1(\overline{z}^+_k)\,.\]
Combined with $q_1$ being nondecreasing, this shows that the constructed $\tilde{g}^j$ is nonincreasing for each $j\in \{1, 2\}$. Therefore, $\tilde{A}_1 \subset \tilde{A}_2$ are two nested up-sets that differ by 
\[\bigcup_k (\underline{z}_k, \overline{z}_k] \times (\underline{t}^\star_k, \overline{t}^\star_k]\,,\]
which is a countable collection of rectangles. Moreover, these rectangles $\{D_k\}_k$ satisfy that for any $k \neq k'$ and any $i \in \{1, 2\}$, we have 
\[\Big\{x_i : \int_0^1 \1_{D_k}(x_i, x_{-i}) \d x_{-i} > 0\Big\} \bigcap \Big\{x_i : \int_0^1 \1_{D_{k'}}(x_i, x_{-i}) \d x_{-i} > 0\Big\} = \emptyset\,.\]
That is, any two different rectangles have non-overlapping projections on each dimension. Combined with that 
\[\bigcup_k D_k = \tilde{A}_2\backslash \tilde{A}_1\]
is the difference between two nested up-sets, we can conclude that the function $f$ is a multidimensional monotone function for any collection of $\{\lambda^\star_k\}_k$. 

We verify that $f$ rationalizes $(q_1, q_2)$. Note that by construction $q_1$ is $f$'s one-dimensional marginal in dimension $1$. It remains to verify that $q_2$ is equal to $f$'s one-dimensional marginal in dimension $2$. By the previous claims, we know that the majorization inequality  
\[\int^1_s q_2(z) \d z \leq \int^1_s \hat{q}_1(z) \d z\]
hold with strict inequality if and only if $s \in (\underline{t}^\star_k, \overline{t}^\star_k)$ for some $k$. Since we also know that $\overline{t}^\star_k < \underline{t}^\star_{k+1}$ for all $k$, it follows that $q_2(s) = \hat{q}_1(s)$ for all $s \not \in \cup_{k} (\underline{t}^\star_k, \overline{t}^\star_k]$ and \[q_2(s) = \frac{\int^{\overline{t}^\star_k}_{\underline{t}^\star_k} \hat{q}_1(t) \d t}{\overline{t}^\star_k -\underline{t}^\star_k}\]
for all $s \in (\underline{t}^\star_k, \overline{t}^\star_k]$ and all $k$. Now, for any $s \not \in \cup_k (\underline{t}^\star_k, \overline{t}^\star_k]$, we can compute 
\[\int_0^1 f(x_1, s) \d x_1 = 1 - (g^1)^{-1}(s) =  1 - (g^2)^{-1}(s) = 1 - q_1^{-1}(1-s) = \hat{q}_1(s)\,.\]
Moreover, for any $k$ and any  $s \in  (\underline{t}^\star_k, \overline{t}^\star_k]$, we can compute 
\begin{align*}
    \int_0^1 f(x_1, s) \d x_1 &= (1 - \lambda^\star_k) (1 - (g^2)^{-1}(s)) + \lambda^\star_k (1 - (g^1)^{-1}(s))\\
    &= (1 - \lambda^\star_k) (1-\underline{z}_k) + \lambda^\star_k (1-\overline{z}_k)\,.
\end{align*}
Recall that we also have 
\[\lambda^\star_k (1-\overline{t}^\star_k) + (1-\lambda^\star_k)(1 -\underline{t}^\star_k) = q_1(\overline{z}_k) \,,\]
and hence 
\[\lambda^\star_k \overline{t}^\star_k + (1-\lambda^\star_k)\underline{t}^\star_k = 1 - q_1(\overline{z}_k) \,.\]
Let $l_k := 1 - q_1(\overline{z}_k)$, which is the jump point of $\hat{q}_1$ on the interval $(\underline{t}^\star_k, \overline{t}^\star_k]$. Note that we can compute 
\[q_2(s) = \frac{\int^{\overline{t}^\star_k}_{\underline{t}^\star_k} \hat{q}_1(t) \d t}{\overline{t}^\star_k -\underline{t}^\star_k} = \frac{\hat{q}_1((\underline{t}^\star_k)^+)\cdot (l_k - 
\underline{t}^\star_k) + \hat{q}_1(\overline{t}^\star_k)\cdot (
\overline{t}^\star_k - l_k)}{\overline{t}^\star_k - \underline{t}^\star_k} = \frac{(1 - \overline{z}_k) (l_k - \underline{t}^\star_k) + (1 - \underline{z}_k) (
\overline{t}^\star_k - l_k)}{\overline{t}^\star_k - \underline{t}^\star_k}\,.\]
Using the construction of $\lambda^\star_k$, we have
\[(1 - \overline{z}_k) (l_k - \underline{t}^\star_k) + (1 - \underline{z}_k)  (\overline{t}^\star_k - l_k)= (1-\overline{z}_k)\lambda^\star_k(\overline{t}^\star_k - \underline{t}^\star_k) +  (1-\underline{z}_k)(1-\lambda^\star_k)(\overline{t}^\star_k - \underline{t}^\star_k)\,, \]
and hence 
\[q_2(s) = \lambda^\star_k (1 - \overline{z}_k) + (1 - \lambda^\star_k) (1 - \underline{z}_k) =   \int_0^1 f(x_1, s) \d x_1 \,,\]
proving the claim. 

Therefore, $f$ rationalizes $(q_1, q_2)$. We first claim that $\lambda^\star_k$ must all coincide. Suppose for contradiction that $\lambda^\star_k \neq \lambda^\star_{k'}$ for some $k$, $k'$. Recall that $\lambda^\star_k, \lambda^\star_{k'}\in (0, 1)$. Then, by \Cref{thm:nested-up-sets}, $f$ is not an extreme point of $\overline{\mathcal{F}}$ where the constraint is the linear constraint imposed on $(q_1, q_2)$. Thus, there exists some $u: [0, 1]^2 \rightarrow \R$ such that $f' := f + u  \in \overline{\mathcal{F}}$,  $f'' := f - u \in \overline{\mathcal{F}}$ and $f = \frac{1}{2}f' + \frac{1}{2}f''$. But clearly, $u$ must be equal to $0$ on $\tilde{A}_1$ and also equal to $0$ on $[0, 1]^2 \backslash \tilde{A}_2$, since $f$ equals either $1$ or $0$ on these two regions, respectively. Thus, $u$ takes non-zero values only on $\tilde{A}_2 \backslash \tilde{A}_1 = \bigcup_j D_j$. Let $q'$ and $q''$ be the one-dimensional marginals of $f'$ and $f''$, respectively. Clearly we have $q'$ and $q''$ are both in $\overline{\mathcal{Q}}$ and $\frac{1}{2}q' + \frac{1}{2}q'' = q$.  Since any two different rectangles $D_j, D_{j'}$ have non-overlapping projections on each dimension, note that $q' \neq q$ and $q'' \neq q$. But then $q$ cannot be an extreme point. Contradiction. 

Therefore, $\lambda^\star_k$ must be a constant $\lambda^\star$ across all $k$. Now, we claim that there can be at most one such rectangle. Indeed, if there are two such rectangles, simply following the same perturbation we used in the beginning for the case of $\underline{\gamma}^2_k \leq  \overline{\gamma}^1_k$ to perturb $f$ would result in $f'$ and $f''$ such that their one-dimensional marginals $q'$ and $q''$ are both in $\overline{\mathcal{Q}}$, are different, and satisfy $\frac{1}{2}q' + \frac{1}{2}q'' = q$. Thus, it must be the case $\tilde{A}_2\backslash\tilde{A}_1$ has only one rectangle, and hence $(q_1, q_2)$ can be rationalized by a mixture of two indicator functions defined on two nested up-sets that differ by at most a rectangle.

\subsubsection{Unique Rationalizability}\label{subsec:unique}
Now we prove that every extreme point of $\overline{\cQ}$ must be uniquely rationalized among all monotone functions. Fix any $q=(q_1, q_2) \in \text{ext}(\overline{\cQ})$. We first prove that it is uniquely rationalized among all mixtures of two indicator functions defined on nested up-sets. Since $q$ is an extreme point, by \Cref{subsec:rationalize}, it is rationalized by $f^\star:= \lambda^\star \1_{A^\star_1}+(1-\lambda^\star)\1_{A^\star_2}$ for some nested up-sets $A_1^\star\subseteq A_2^\star$ that differ by at most a rectangle and some $\lambda^\star$. If $A_2^\star \backslash A_1^\star$ has measure $0$, then $q$ must be an extreme point of $\mathcal{Q}$ and uniquely rationalized among all functions by \Cref{thm:rationalized-upsets}. Now suppose $A_2^\star \backslash A_1^\star = [\underline{x}_1, \overline{x}_1] \times [\underline{x}_2, \overline{x}_2]$, for some $\underline{x}_1 < \overline{x}_1$ and $\underline{x}_2 < \overline{x}_2$. Since $A_1^\star$ and $A_2^\star$ are up-sets, there exist nonincreasing left-continuous functions $g_1^\star, g_2^\star:[0,1] \to [0,1]$ such that $g_1^\star \geq g_2^\star$, and that $A_j^\star=\{(x_1,x_2): x_2 \geq g^\star_j(x_1)\}$ for $j \in \{1,2\}$ (up to a measure zero set). Then, since $A_1^\star$ and $A_2^\star$ differ by a rectangle $[\underline{x}_1,\overline{x}_1] \times [\underline{x}_2,\overline{x}_2]$, it must be that $g_1^\star \equiv g_2^\star$ on $[0,1] \backslash (\underline{x}_1,\overline{x}_1]$, $g_1^\star \equiv \overline{x}_2$ and $g_2^\star \equiv \underline{x}_2$ on $(\underline{x}_1,\overline{x}_1]$. 

Take any $\hat{x}_2 \in (\underline{x}_2,\overline{x}_2)$, and let 
\[
\phi_1(x_1):=\begin{cases}
-g_1^\star(x_1),&\mbox{if } x_1 \notin (\underline{x}_1,\overline{x}_1]\\
-\hat{x}_2,&\mbox{if } x_1 \in (\underline{x}_1,\overline{x}_1]
\end{cases}\,,\quad \mbox{ and } \quad 
\phi_2(x_2):=\begin{cases}
x_2,&\mbox{if } x_2 \notin (\underline{x}_2,\overline{x}_2]\\
\hat{x}_2,&\mbox{if } x_2 \in (\underline{x}_2,\overline{x}_2]
\end{cases}\,,
\]
for all $(x_1,x_2) \in [0,1]^2$.
Then, note that 
\[
\phi_1(x_1)+\phi_2(x_2)=\begin{cases}
x_2-g_1^\star(x_1),&\mbox{if } x_1 \notin (\underline{x}_1,\overline{x}_1]\,, x_2 \notin (\underline{x}_2,\overline{x}_2]\\
x_2-\hat{x}_2,&\mbox{if } x_1  \in (\underline{x}_1,\overline{x}_1]\,, x_2 \notin (\underline{x}_2,\overline{x}_2]\\
\hat{x}_2-g_1^\star(x_1),&\mbox{if }  x_1  \notin (\underline{x}_1,\overline{x}_1]\,, x_2 \in (\underline{x}_2,\overline{x}_2]\\
0,& \mbox{if } x_1 \in (\underline{x}_1,\overline{x}_1]\,, x_2 \in (\underline{x}_2,\overline{x}_2]
\end{cases}
\]
In particular, since $g_1^\star \geq g_2^\star$ are nonincreasing and since $g_1^\star \equiv \overline{x}_2$ and $g_2^\star \equiv \underline{x}_2$ on $(\underline{x}_1,\overline{x}_1]$,
\begin{align*}
\phi_1(x_1)+\phi_2(x_2) &\begin{cases}
>0,&\mbox{if} (x_1,x_2) \in \{(x_1,x_2):x_2>g_1^\star(x_1)\,, x_1 \notin (\underline{x}_1,\overline{x}_1]\} \cup (\underline{x_1},\overline{x}_1] \times (\overline{x}_2,1] \\
<0,&\mbox{if } (x_1,x_2) \in \{(x_1,x_2):x_2<g_1^\star(x_1)\,, x_1 \notin (\underline{x}_1,\overline{x}_1]\} \cup (\underline{x_1},\overline{x}_1] \times [0,\underline{x}_2]\\
=0,&\mbox{if } (x_1,x_2) \in \{(x_1,x_2):x_2=g_1^\star(x_1)\,, x_1 \notin (\underline{x}_1,\overline{x}_1]\} \cup (\underline{x}_1,\overline{x}_1] \times (\underline{x}_2,\overline{x}_2]\,.
\end{cases}
\end{align*}

Let $\overline{\mathcal{F}} \subseteq \F$ be the set of monotone functions $f$ such that 
\[
\int_{[0,1]^2} (\psi_1(x_1)+\psi_2(x_2))f(x_1,x_2)\d x \leq \eta \,,
\]
and consider the linear functional $\Phi:\overline{\F} \to \R$, where 
\[
\Phi(\tilde{f}):=\int_{[0,1]^2} [\phi_1(x_1)+\phi_2(x_2)]\tilde{f}(x_1,x_2)\d x \,.
\]
Note that since $f^\star(x_1,x_2)=1$ for all $(x_1,x_2)$ such that $\phi_1(x_1)+\phi_2(x_2)>0$, and $f^\star(x_1,x_2)=0$ for all $(x_1,x_2)$ such that $\phi_1(x_1)+\phi_2(x_2)<0$, we must have that $f^\star$ solves 
\begin{equation}\label{eq:linear-square}
\max_{\tilde{f} \in \overline{\F}} \Phi(\tilde{f})\,.
\end{equation}
Moreover, we claim that for any $f \in \overline{\F}$, if \textit{(i)} $f$ is a solution of \eqref{eq:linear-square}, \textit{(ii)} $f=\lambda\1_{A_1}+(1-\lambda)\1_{A_2}$, where $A_1 \subseteq A_2$ are nested up-sets, and \textit{(iii)} $f$ rationalizes $q$, then $f \equiv f^\star$ almost everywhere. Indeed, for $f=\lambda_{A_1}+(1-\lambda)\1_{A_2}=\1_{A_1}+(1-\lambda) \1_{A_2\backslash A_1}$ to solve \eqref{eq:linear-square}, it must be that $\Phi(f)=\Phi(f^\star)$. As a result, $f(x_1,x_2)$ must equal $1$ whenever $\phi_1(x_1)+\phi_2(x_2)>0$, and must equal $0$ whenever  $\phi_1(x_1)+\phi_2(x_2)<0$. Hence, $f$ and $f^\star$ could only differ on $(\underline{x}_1,\overline{x}_1] \times (\underline{x}_2,\overline{x}_2]$. Moreover, we may assume that $\lambda \in (0, 1)$, because otherwise $f$ is an indicator function defined on an up-set, which by \Cref{thm:rationalized-upsets} implies that $f^\star \equiv f$.  Meanwhile, since $A_1 \subseteq A_2$ are up-sets, there exist  nonincreasing functions $g_1,g_2:[0,1] \to [0,1]$ with $g_1 \geq g_2$ such that $\{(x_1, x_2): x_2 \geq g_j(x_1)\} \backslash A_j$ has measure zero for each $j$. Since $f$ also rationalizes $q$, we have 
\[
\lambda g_1(x_1)+(1-\lambda) g_2(x_1)=1-q_1(x_1)=\lambda^\star g_1^\star(x_1)+(1-\lambda^\star)g_2^\star(x_1)\,,
\]
for all $x_1 \in [0,1]$, and 
\[
\lambda g_1^{-1}(x_2)+(1-\lambda) g_2^{-1}(x_2)=1-q_2(x_2)=\lambda^\star(g_1^\star)^{-1}(x_2)+(1-\lambda^\star)(g_2^\star)^{-1}(x_2)\,,
\]
for all $x_2 \in [0,1]$. Since $q_1$ and $q_2$ are constant on $(\underline{x}_1,\overline{x}_1]$ and $(\underline{x}_2,\overline{x}_2]$, respectively, it must be that $g_1$ and $g_2$ are constant on $(\underline{x}_1,\overline{x}_1]$, while $g_1^{-1}$ and $g_2^{-1}$ are constant on $(\underline{x}_2,\overline{x}_2]$. Moreover, since $f \equiv f^\star$ on $[0,1]^2 \backslash (\underline{x}_1,\overline{x}_1] \times (\underline{x}_2,\overline{x}_2]$, it must be that $\lambda=\lambda^\star$, $g_1 \equiv \overline{x}_2$, and $g_2 \equiv \underline{x}_2$ on $(\underline{x}_1,\overline{x}_1]$. Indeed, if the constant value $\kappa_1$ of $g_1$ does not coincide with $\overline{x}_2$, then sine $\lambda \in (0, 1)$, $q_2$ must take at least two distinct values on $(\underline{x}_2,\overline{x}_2]$. Thus, $\kappa_1 = \overline{x}_2$. Similarly the constant value $\kappa_2$ of $g_2$ must coincide with $\underline{x}_2$ in order for $q_2$ to be constant on $(\underline{x}_2,\overline{x}_2]$. But then the constant $\lambda$ must coincide with $\lambda^\star$ for $f$ to rationalize $(q_1, q_2)$. Together, this implies that $f \equiv f^\star$ almost everywhere, as desired.

Lastly, for any $f \in \overline{\F}$ that rationalizes $q$, since
\begin{align*}
\Phi(f^\star)=&\int_{[0,1]^2} [\phi_1(x_1)+\phi_2(x_2)] f^\star(x_1,x_2)\d x\\
=&\int_0^1 \phi_1(x_1) q_1(x_1) \d x_1 +\int_0^1 \phi_2(x_2)q_2(x_2)\d x_2\\
=& \int_{[0,1]^2} [\phi_1(x_1)+\phi_2(x_2)] f(x_1,x_2)\d x\\
=& \Phi(f)\,,
\end{align*}
we have that $f$ solves \eqref{eq:linear-square}, where the second and the third equations follow from rationalizability. As a result, if, in addition, $f=\lambda\1_{A_1}+(1-\lambda)\1_{A_2}$ for some $\lambda \in (0,1)$ and for some nested up-sets $A_1 \subseteq A_2$ then it must be that $f\equiv f^\star$ almost everywhere, as desired. 

It now remains to prove that $q$ is uniquely rationalized among all monotone functions. To see this, note that any monotone $f \in \mathcal{F}$ that rationalizes $q$ must be in $\overline{\mathcal{F}}$. Now, for any monotone $f \in \overline{\mathcal{F}}$ that rationalizes the extreme point $q \in \overline{\cQ}$, by \Cref{thm:nested-up-sets} and Choquet's integral representation theorem, 
\[f = \int f^r \mu (\d r)\]
where $r \in [0, 1]$ is an index, $f^r = \lambda^r \1_{A_1^r} + (1 - \lambda^r) \1_{A^r_2}$ where $A^r_1 \subseteq A^r_2$ are two nested up-sets, $\lambda^r \in (0,1)$, and $f^r \in \overline{\mathcal{F}}$. Let $q^r=(q^r_1, q^r_2) \in \overline{\mathcal{Q}}$ be the one-dimensional marginals of $f^r$. That is, 
\[
q_1^r(x_1):=\int_0^1 f^r(x_1,x_2)\d x_2 \quad \mbox{ and } \quad q_2^r(x_2):=\int_0^1 f^r(x_1,x_2)\d x_1\,.
\]
It follows immediately that, for $i \in \{1,2\}$,
\[q_i = \int q^r_i \mu (\d r)\,.\]
Since $(q_1, q_2)$ is an extreme point of $\overline{\mathcal{Q}}$, by \citet{bauer1961silovscher}, there exists a unique probability measure on $\overline{\mathcal{Q}}$ that represents $q$. Thus, 
\[q^r \equiv q\]
for almost all $r$ with respect to measure $\mu$. For all such $r$, because $f^r$ is a mixture of two indicator functions defined on nested up-sets that also rationalizes $q$, by the previous claim, it must be that 
\[f^r \equiv \lambda^\star \1_{A^\star_1} + (1- \lambda^\star) \1_{A^\star_2}\,,\]
almost everywhere. Therefore, we have 
\[f \equiv \lambda^\star \1_{A^\star_1} + (1- \lambda^\star) \1_{A^\star_2}\,,\]
almost everywhere. This completes the proof.

\subsection{Proof of \texorpdfstring{\Cref{prop:marginal-sufficiency}}{}}
Necessity and uniqueness follow immediately from the proof of \Cref{thm:rectangle}. For sufficiency, consider any such $q$. Let $f=\1_{A}+\lambda\1_{D}$ be a monotone function that rationalizes $q$, where $\lambda \in (0,1)$, $A \subseteq [0,1]^2$ is an up-set, and $D \subseteq [0,1]^2$ is a rectangle. If $D$ has measure zero, then $q$ is an extreme point of $\overline{\cQ}^\star$ according to \Cref{thm:rationalized-upsets}. Suppose now that $D$ has a positive measure. Then, $\mathrm{cl}(D)=[\underline{x}_1,\overline{x}_1] \times [\underline{x}_2,\overline{x}_2]$ with $\underline{x}_1<\overline{x}_1$ and $\underline{x}_2<\overline{x}_2$. In particular, $q_1(x_1)=1-(\lambda\underline{x}_2+(1-\lambda) \overline{x}_2)$ for all $x_1 \in (\underline{x}_1,\overline{x}_1]$, whereas $q_2(x_2)=1-(\lambda \underline{x}_1+(1-\lambda)\overline{x}_1)$ for all $x_2 \in (\underline{x}_2,\overline{x}_2]$. Let 
\[
g_1(z):=\begin{cases}
1-q_1(z),&\mbox{if } z \notin (\underline{x}_1,\overline{x}_1]\\
\overline{x}_2,&\mbox{if } z \in (\underline{x}_1,\overline{x}_1]
\end{cases}\,; \quad  \mbox{ and } \quad
g_2(z):=\begin{cases}
1-q_1(z),&\mbox{if } z \notin (\underline{x}_1,\overline{x}_1]\\
\underline{x}_2,&\mbox{if } z \in (\underline{x}_1,\overline{x}_1]
\end{cases}\,,
\]
for all $z \in [0,1]$. Note that $g_1,g_2$ are nonincreasing and left-continuous, $g_1 \geq g_2$, and $A=\{(x_1,x_2):x_2 \geq g_1(x_1)\}$ whereas $A\cup D=\{(x_1,x_2):x_2 \geq g_2(x_1)\}$ almost everywhere. 

In the meantime, by Choquet's representation theorem, there exist  $\{q^r\}_{r \in [0,1]} \subseteq \mathrm{ext}(\overline{\cQ}^\star)$ and a probability measure $\mu \in \Delta([0,1])$ such that 
\begin{equation}\label{eq:mixrationalize}
q_i(x_i)=\int_0^1 q^r_i(x_i) \mu(\d r)\,,
\end{equation}
for all $x_i \in [0,1]$ and $i \in \{1,2\}$. We prove that $q$ is an extreme point of $\ub{\cQ}^\star$ by showing that $q^r \equiv q$ almost everywhere, for $\mu$-almost all $r \in [0,1]$. To this end, first note that since $q_i$ is constant on $(\underline{x}_i,\overline{x}_i]$, $q_i^r$ must also be constant on $(\underline{x}_i,\overline{x}_i]$ for $\mu$-almost all $r \in [0,1]$ and for all $i \in \{1,2\}$. Furthermore, since $q^r$ is an extreme point of $\overline{\cQ}^\star$, by \Cref{thm:rectangle}, $q^r$ must be rationalized by some  $f^r:=\1_{A_r}+\lambda_r \1_{D_r}$, for all $r \in [0,1]$, where $\lambda_r \in (0,1)$, $A_r \subseteq [0,1]^2$ is an up-set and $D_r \subseteq [0,1]^2$ is a rectangle.

Similar to the proof of the uniqueness part of \Cref{thm:rectangle}, take any $\hat{x}_2 \in (\underline{x}_2,\overline{x}_2)$, and let 
\[
\phi_1(x_1):=\begin{cases}
-g_1(x_1),&\mbox{if } x_1 \notin (\underline{x}_1,\overline{x}_1]\\
-\hat{x}_2,&\mbox{if } x_1 \in (\underline{x}_1,\overline{x}_1]
\end{cases}\,,\quad \mbox{ and } \quad 
\phi_2(x_2):=\begin{cases}
x_2,&\mbox{if } x_2 \notin (\underline{x}_2,\overline{x}_2]\\
\hat{x}_2,&\mbox{if } x_2 \in (\underline{x}_2,\overline{x}_2]
\end{cases}\,,
\]
for all $(x_1,x_2) \in [0,1]^2$.
Then, as argued in the uniqueness proof of \Cref{thm:rectangle} (see \Cref{subsec:unique}),
\begin{align*}
\phi_1(x_1)+\phi_2(x_2) &\begin{cases}
>0,&\mbox{if} (x_1,x_2) \in A \\
<0,&\mbox{if } (x_1,x_2) \in [0,1]^2 \backslash A \cup D\\
=0,&\mbox{if } (x_1,x_2) \in D\,,
\end{cases}
\end{align*}
for almost all $(x_1,x_2) \in [0,1]^2$. Let $\overline{\mathcal{F}}^\star \subseteq \F$ be the set of monotone functions $\tilde{f}$ such that 
\begin{equation}\label{eq:linear-equality}
\int_{[0,1]^2} (\psi_1(x_1)+\psi_2(x_2))\tilde{f}(x_1,x_2)\d x = \eta \,.
\end{equation}
Note that, since $q^r \in \overline{\cQ}^\star$ for all $r \in [0,1]$ and $q \in \overline{\cQ}^\star$, and since $f$ rationalizes $q$ and $f^r$ rationalized $q^r$ for all $r \in [0,1]$, we have $f,f^r \in \overline{\F}^\star$ for all $r \in [0,1]$. Consider now the linear functional $\Phi:\overline{\F}^\star \to \R$, where 
\[
\Phi(\tilde{f}):=\int_{[0,1]^2} [\phi_1(x_1)+\phi_2(x_2)]\tilde{f}(x_1,x_2)\d x \,.
\]
Note that since $f(x_1,x_2)=1$ for all $(x_1,x_2)$ such that $\phi_1(x_1)+\phi_2(x_2)>0$, and $f(x_1,x_2)=0$ for all $(x_1,x_2)$ such that $\phi_1(x_1)+\phi_2(x_2)<0$, we must have that $f$ solves 
\begin{equation}\label{eq:linear-square2}
\max_{\tilde{f} \in \overline{\F}^\star} \Phi(\tilde{f})\,.
\end{equation}
Moreover, any other solution of \eqref{eq:linear-square2} can only differ from $f$ with a positive measure on $D$. 

Note that, by \eqref{eq:mixrationalize}, we have 
\begin{align*}
\Phi(f)=&\int_{[0,1]^2} (\phi_1(x_1)+\phi_2(x_2))f(x_1,x_2)\d x\\
=& \int_0^1 q_1(x_1) \phi_1(x_1)\d x_1+\int_0^1 q_2(x_2)\phi_2(x_2) \d x_2\\
=& \int_0^1 \left(\int_0^1 q_1^r(x_1)\mu(\d r)\right)\phi_1(x_1)\d x_1+\int_0^1 \left(\int_0^1 q_2^r(x_2)\mu(\d r)\right)\phi_2(x_2)\d x_2\\
=&\int_0^1 \left(\int_0^1 \int_0^1 f^r(x_1,x_2)\mu(\d r) \d x_2\right)\phi_1(x_1)\d x_1+\int_0^1 \left(\int_0^1 \int_0^1 f^r(x_1,x_2)\mu(\d r) \d x_1\right)\phi_2(x_2)\d x_2\\
=& \int_{[0,1]^2} (\phi_1(x_1)+\phi_2(x_2))\left(\int_0^1 f^r(x_1,x_2)\mu(\d r)\right) \d x\\
=&\Phi\left(\int_0^1 f^r\mu(\d r)\right)\\
=&\int_0^1 \Phi(f^r) \mu(\d r)\,,
\end{align*}
and therefore $f^r$ must also solve \eqref{eq:linear-square2} for $\mu$-almost $r \in [0,1]$. As a result, for $\mu$-almost all $r \in [0,1]$, $f^r \equiv f$ on $[0,1]^2 \backslash D$ almost everywhere. In particular, $D_r \subseteq D$ for $\mu$-almost all $r \in [0,1]$. Furthermore, since $q_1^r$ and $q_2^r$ are constant on $(\underline{x}_1,\overline{x}_1]$ and $(\underline{x}_2,\overline{x}_2]$, respectively, and since $\mathrm{cl}(D)=[\lb{x}_1,\ub{x}_1] \times [\lb{x}_2,\ub{x}_2]$, it must be that $D_r=D$ for $\mu$-almost all $r \in [0,1]$ (by the same argument as in \Cref{subsec:unique}), and thus, $q_1^r(x_1)=1-(\lambda_r \underline{x}_2+(1-\lambda_r)\overline{x}_2)$ for all $x_1 \in (\lb{x}_1,\ub{x}_1]$ and $q_2^r(x_2)=1-(\lambda_r \lb{x}_1+(1-\lambda_r)\ub{x}_1)$ for all $x_2 \in (\lb{x}_2,\ub{x}_2]$. Together, since $f^r \equiv f$ outside of $D$, \eqref{eq:mixrationalize} and \eqref{eq:linear-equality} then imply that, for $\mu$-almost $r \in [0,1]$, 
\begin{align*}
&\int_{\lb{x}_1}^{\ub{x}_1}\left(1-(\lambda \lb{x}_2+(1-\lambda)\ub{x}_2)\right)\psi_1(x_1)\d x_1+\int_{\lb{x}_2}^{\ub{x}_2}\left(1-(\lambda \lb{x}_1+(1-\lambda)\ub{x}_1)\right)\psi_2(x_2)\d x_2\\
=&\int_{\lb{x}_1}^{\ub{x}_1}\left(1-(\lambda_r \lb{x}_2+(1-\lambda_r)\ub{x}_2)\right)\psi_1(x_1)\d x_1+\int_{\lb{x}_2}^{\ub{x}_2}\left(1-(\lambda_r \lb{x}_1+(1-\lambda_r)\ub{x}_1)\right)\psi_2(x_2)\d x_2\,,
\end{align*}
which simplifies to 
\[
\lambda \int_{D}(\psi_1(x_1)+\psi_2(x_2))\d x
=\lambda_r \int_{D}(\psi_1(x_1)+\psi_2(x_2))\d x\,.
\]
Since 
\[
\int_{D}(\psi_1(x_1)+\psi_2(x_2))\d x \neq 0\,,
\]
it must be that $\lambda=\lambda_r$. As a result, it must be that, for $\mu$-almost $r \in [0,1]$, $f^r \equiv f$ almost everywhere, and hence $q^r \equiv q$ almost everywhere. This completes the proof.   

\subsection{Proof of \texorpdfstring{\Cref{prop:reverse-majorization}}{}}

By \Cref{lem:gutmann}, we have 
\[
\Big\{q: q_1 \preceq \hat{q}_2\,;\, q_1,q_2 \mbox{ are nondecreasing and left-continuous} \Big\}=\cQ\,,
\]
which is clearly convex. Now, by \Cref{thm:rationalized-upsets}, $q \in \cQ$ is an extreme point of $\cQ$ if and only if it is rationalized by $\1_A$ for some up-set $A \subseteq [0,1]^2$. 

Now fix any extreme point $q$ of $\cQ$. There exists $\1_{A}$ for some up-set $A$ that rationalizes $q$. Since $A$ is an up-set, there exists nonincreasing and left-continuous function $g:[0,1] \to [0,1]$ such that $A=\{(x_1,x_2): x_2 \geq g(x_1)\}$. Thus,
\[
q_1(x_1)=\int_0^1 \1_A(x_1,x_2)\d x_2=\int_{g(x_1)}^1 1 \d x_2=1-g(x_1)\,, 
\]
and 
\[
q_2(x_2)=\int_0^1 \1_A(x_1,x_2) \d x_1=\int_{g^{-1}(x_2)}^1 1 \d x_1=1-g^{-1}(x_2)\,,
\]
and hence, $q_1 \equiv \hat{q}_2$. 

Conversely, fix any $q$ such that $q_1$, $q_2$ are nondecreasing and left-continuous, and $q_1 \equiv \hat{q}_2$. Let 
\[
A:=\{(x_1,x_2): x_2 \geq 1-q_1(x_1)\}\,.
\]
Then, $A$ is an up-set since $q_1$ is nondecreasing. Moreover, we have 
\[
q_1(x_1)=\int_0^1 \1_A(x_1,x_2) \d x_2
\]
by definition, and 
\[
q_2(x_2)=\hat{q}_1(x_2)=\int_0^1 \1_A(x_1,x_2) \d x_1\,.
\]
That is, $q$ is rationalized by $\1_A$, and hence $q$ is an extreme point of 
\[
\cQ=\Big\{q: q_1 \preceq \hat{q}_2\,;\, q_1,q_2 \mbox{ are nondecreasing and left-continuous} \Big\}\,.
\]
This completes the proof.

\subsection{Proof of \texorpdfstring{\Cref{prop:square-majorization}}{}}
By \Cref{lem:gutmann}, we have 
\[
\begin{aligned}
\bigg\{q\,\,:\,\, &q_1 \preceq \hat{q}_2\,;\, q_1,q_2 \mbox{ are nondecreasing and left-continuous}\,;\\
&\int_{0}^1 q_1(x_1)\psi_1(x_1)\d x_1+\int_0^1 q_2(x_2)\psi_2(x_2)\d x_2 \leq \eta \bigg\}=\overline{\cQ}\,.
\end{aligned}
\]
Thus, by \Cref{thm:rectangle}, any extreme point of
\[
\begin{aligned}
\bigg\{q\,\,:\,\, &q_1 \preceq \hat{q}_2\,;\, q_1,q_2 \mbox{ are nondecreasing and left-continuous}\,;\\
&\int_{0}^1 q_1(x_1)\psi_1(x_1)\d x_1+\int_0^1 q_2(x_2)\psi_2(x_2)\d x_2 \leq \eta \bigg\}
\end{aligned}
\]
must be rationalized by $(1-\lambda)\1_{A_1}+\lambda\1_{A_2}$ for some $\lambda \in [0,1]$ and some nested up-sets $A_1 \subseteq A_2 \subseteq [0,1]^2$ that differ by at most a rectangle $[\underline{x_1},\overline{x}_1] \times [\underline{x}_2,\overline{x}_2]$. 

Furthermore, for the up-sets $A_1,A_2$, there exist nonincreasing and left-continuous functions $g_1,g_2:[0,1] \to [0,1]$ such that $A_j=\{(x_1,x_2): x_2 \geq g_j(x_1)\}$ for $j \in \{1,2\}$, where $g_1\equiv g_2$ on $[0,1] \backslash [\underline{x}_1,\overline{x}_1]$, and $g_2(x_1)=\underline{x}_2$, $g_1(x_1)=\overline{x}_2$ for all $x_1 \in (\underline{x}_1,\overline{x}_1]$. Since $q$ is rationalized by $(1-\lambda)\1_{A_1}+\lambda\1_{A_2}$, it follows that 
\[
q_1(x_1)=\int_0^1  (1-\lambda)  \1_{A_1}(x_1,x_2)+ \lambda\1_{A_2}(x_1,x_2) \d x_2=(1-\lambda)  (1-g_1(x_1))+\lambda(1-g_2(x_1))\,,
\]
and 
\[
q_2(x_2)=\int_0^1 (1-\lambda) \1_{A_1}(x_1,x_2)+ \lambda  \1_{A_2}(x_1,x_2) \d x_1= (1-\lambda)(1-g_1^{-1}(x_2))+ \lambda(1-g_2^{-1}(x_2))\,.
\]
Let $\underline{z}:=\underline{x}_1$ and $\overline{z}:=\overline{x}_1$. Let $\underline{\gamma} := 1 - \overline{x}_2$ and $\overline{\gamma} := 1 - \underline{x}_2$. This then completes the proof.

\subsection{Proof of \texorpdfstring{\Cref{prop:weak}}{}}
We first show that the set 
\[\mathcal{Q}^w := \Big\{q\,:\,  q_1 \preceq_w  \hat{q}_2\,;\, q_1,q_2 \mbox{ are nondecreasing and left-continuous}\Big\}\]
is convex. By Theorem 4.A.6 of \citet{shaked2007stochastic}, for any two nondecreasing left-continuous functions $g_1,g_2:[0,1] \to [0,1]$, $g_1 \preceq_w g_2$ if and only if 
\[
g_1 \leq \tilde{g}_1 \preceq g_2
\]
for some nondecreasing left-continuous function $\tilde{g}_1:[0,1] \to [0,1]$. 

Consider the set 
\[\mathcal{K}:=\Big\{(q_1, \tilde{q}_1, q_2)\,:\, q_1 \leq \tilde{q}_1 \preceq  \hat{q}_2\,;\, q_1, \tilde{q}_1, q_2 \mbox{ are nondecreasing and left-continuous}\Big\}\,.\]
Note that $\mathcal{K}$ is convex since for any $\lambda \in [0, 1]$ and any $(q_1, \tilde{q}_1, q_2), (q'_1, \tilde{q}'_1, q'_2) \in \mathcal{K}$, we have 
\[\lambda q_1 + (1-\lambda) q'_1 \leq \lambda \tilde{q}_1 + (1-\lambda) \tilde{q}'_1 \preceq \widehat{\lambda q_2 + (1-\lambda) q'_2}\,,\]
where the majorization relation is due to the convexity of the set \[\Big\{q\,:\, q_1 \preceq \hat{q}_2\,;\,\, q_1,q_2 \mbox{ are nondecreasing and left-continuous} \Big\}\,,\]
established in \Cref{prop:reverse-majorization}. Therefore, $\mathcal{K}$ is convex. As $\cQ^w$ is a linear projection of $\mathcal{K}$, and $\mathcal{K}$ is convex, $\cQ^w$ must also be convex.

Now, fix any extreme point $q$ of the convex set $\mathcal{Q}^w$. Since $q_1\preceq_w \hat{q}_2$, we know that $q_1$ must be an extreme point of the convex set of nondecreasing left-continuous functions that are weakly majorized by $\hat{q}_2$. By Corollary 2 of \citet*{kleiner2021extreme}, there exists $k \in [0,1]$ such that $q_1 \equiv 0$ on $[0,k]$ almost everywhere and that $q_1$ is an extreme point of the convex set of nondecreasing left-continuous functions that are majorized by $\hat{q}_2\1_{[k,1]}$. In particular, 
\[
\int_s^1 q_1(z) \d z \leq \int_s^1 \hat{q}_2\1_{[k,1]}(z) \d z= \int_s^1 \hat{q}_2(z) \d z
\]
for all $s \in [k,1]$, and 
\[
\int_k^1 q_1(z) \d z =\int_0^1 q_1(z) \d z= \int_0^1 \hat{q}_2 \1_{[k,1]}(z) \d z=\int_k^1 \hat{q}_2(z) \d z\,.
\]

Now, consider the set 
\[\cQ^k := \Big\{q\,:\, q_1(s) = 0 \text{ for all $s \leq k$}\,;\, q_1 \preceq \hat{q}_2\,;\, q_1, q_2 \text{ are nondecreasing and left-continuous} \Big\}\,.\]
We know that $\cQ^k$ is convex since it is the intersection of two convex sets. We claim that the extreme point $q$ of $\mathcal{Q}^w$ must satisfy 
\[\Big(q_1, \min\{q_2, 1 - k\}\Big) \in \text{ext}(\mathcal{Q}^k)\,.\]
Indeed, note that for all $s \leq k$, 
\[
\mathrm{conjugate}[\min\{q_2, 1 - k\}](s) = 0\,, 
\]
and  for all $s > k$, 
\[
\mathrm{conjugate}[\min\{q_2, 1 - k\}](s) = \hat{q}_2(s)\,,
\]
by definition, where $\mathrm{conjugate}(q)$ denotes the conjugate of a monotone function $q$. Thus, by the construction of $k$ and previous observation, we have 
\[q_1 \preceq \mathrm{conjugate}[\min\{q_2, 1 - k\}]
\]
and $q_1(s) = 0$ for all $s \leq k$. Therefore, $(q_1, \min\{q_2, 1 - k\}) \in \mathcal{Q}^k$. Moreover, note that by \Cref{lem:gutmann}, we have 
\[\mathcal{Q}^k = \mathcal{Q} \cap \{(q_1, q_2)\,:\,\text{$q_1(s) = 0$ for all $s \leq k$}\}\,. \]
where $\mathcal{Q}$ is the set of rationalizable monotone pairs. Then, any $(\tilde{q}_1, \tilde{q}_2) \in \mathcal{Q}^k$ must be rationalized by some function $f$ such that $f(x_1, x_2) = 0$ for all $x_1 \leq k$. Thus, $\tilde{q}_2(s) \leq 1 - k$. Suppose for contradiction that $(q_1, \min\{q_2, 1 - k\})$ is not an extreme point of $\mathcal{Q}^k$. Then there exist $u_1, u_2: [0, 1] \rightarrow \R$ such that at least one of $u_1$ and $u_2$ is not identically zero and  
\[(q_1 + u_1, \min\{q_2, 1- k\} + u_2) \in \mathcal{Q}^k\,,\]
and
\[(q_1 - u_1, \min\{q_2, 1- k\} - u_2) \in \mathcal{Q}^k\,.\]
We claim that 
\[(q_1 + u_1, q_2 + u_2) \in \mathcal{Q}^w\,.\]
\[(q_1 - u_1, q_2 - u_2) \in \mathcal{Q}^w\,.\]
Indeed, note that $u_2(s) = 0$ for all $s$ such that 
\[q_2(s) \geq 1 - k\,,\]
since otherwise it would imply that at such $s$ we have either
\[\min\{q_2(s), 1-k\}+ u_2(s) = 1-k + u_2(s) > 1 - k\]
or 
\[\min\{q_2(s), 1-k\} - u_2(s) = 1-k - u_2(s) > 1 - k\,,\]
but as argued before, any  $(\tilde{q}_1, \tilde{q}_2) \in \mathcal{Q}^k$ must satisfy $\tilde{q}_2(s) \leq 1 - k$ for all $s \in [0, 1]$. Now, for all $s$ such that $q_2(s) < 1 - k$, we can write 
\[\min\big\{q_2(s), 1- k\big\} \pm u_2(s) = q_2(s) \pm u_2(s)\,.\]
Together, for all $s\in [0, 1]$, we have 
\[\min\big\{q_2(s), 1- k\big\} \pm u_2(s) = q_2(s) \pm u_2(s)\,.\]
Therefore, $ q_2 \pm u_2$ are both monotone. Clearly, $q_1 \pm u_1$ are both monotone. Moreover,
\[\int_s^1 \big(q_1(z) + u_1(z)\big) \d z \leq \int_s^1 \text{conjugate}[\min\{q_2, 1 -k\} + u_2](z) \d z  =  \int_s^1 \text{conjugate}[q_2 + u_2](z)\d z\,, \]
for all $s\in [0,1]$. Similarly,  
\[\int_s^1 \big(q_1(z) - u_1(z)\big) \d z \leq \int_s^1 \text{conjugate}[\min\{q_2, 1 -k\} - u_2](z) \d z  =  \int_s^1 \text{conjugate}[q_2 - u_2](z)\d z\,, \]
for all $s \in [0, 1]$. Therefore, 
\[(q_1 + u_1, q_2 + u_2) \in \mathcal{Q}^w\,.\]
\[(q_1 - u_1, q_2 - u_2) \in \mathcal{Q}^w\,.\]
But then, since at least one of $u_1$ and $u_2$ is not identically $0$, the pair $(q_1, q_2)$ cannot be an extreme point of $\mathcal{Q}^w$, a contradiction. 

Therefore, the extreme point $q$ of $\mathcal{Q}^w$ must satisfy 
\[\Big(q_1, \min\{q_2, 1 - k\}\Big) \in \text{ext}(\mathcal{Q}^k)\,.\]
Note that by \Cref{lem:gutmann-monotone}, the set 
\[\mathcal{Q}^k = \mathcal{Q} \cap \{(q_1, q_2)\,:\,\text{$q_1(s) = 0$ for all $s \leq k$}\}\, \]
is exactly the set of one-dimensional marginals of the following set 
\[\mathcal{F}^k := \Big\{f \in \mathcal{F}: f(x_1, x_2) = 0 \text{ for all $x_1 \leq k$ and all $x_2$}\Big\}\,,\]
where $\mathcal{F}$ is the set of monotone functions from $[0, 1]^2$ to $[0, 1]$. Clearly, $\mathcal{F}^k$ can be identified by the set of monotone functions from $[k, 1] \times [0, 1]$ to $[0, 1]$. Therefore, by \Cref{lem:extreme-projection} and the proof of \Cref{thm:extreme-are-up-sets}, we immediately have that any extreme point of $\mathcal{Q}^k$ must be rationalized by $\1_{A}$ for some up-set $A \subseteq [k, 1]\times[0, 1]$, which is also an up-set in $[0, 1]^2$. Since $\big(q_1, \min\{q_2, 1 - k\}\big)$ is an extreme point of $\mathcal{Q}^k$ and hence is rationalized by such an up-set function $\1_A$, by the proof of \Cref{prop:reverse-majorization}, we have 
\[q_1 = \mathrm{conjugate}[\min\{q_2, 1 - k\}] = \hat{q}_2 \1_{[k, 1]}\,,\]
proving the result. 

\subsection{Proof of \texorpdfstring{\Cref{prop:example}}{}}
First, note that, given an externality-refund mechanism, and given any realized pair of signals $(s_1,s_2)$, each agent paying the price $k$ if and only if their own signal $s_i$ is weakly above $k$ forms an ex-post Nash equilibrium. Thus, such a mechanism can implement the allocation rule $\alpha(s_1,s_2)=\1\{\max\{s_1,s_2\} \geq k\}$ for any $k \geq 0$. 

Next, we show that there must exist an optimal mechanism that is a randomization of two externality-refund mechanisms. To see this, note that since the environment is symmetric, there must exist a symmetric optimal mechanism. By strong duality, any symmetric optimal mechanism $\alpha$ must solve
\begin{align}\label{eq:dual}
\max_{\alpha \in \F_{\mathrm{sym}}} &\Bigg[\E\Bigg[ \alpha(s) \Big(\sum_i v_i(s) - c \Big)\Bigg]\notag\\
&+\lambda \Bigg(\mathbb{E}\Bigg[\sum_i \alpha(s) v_i(s)  - \int^{s_i}_0 \alpha(z, s_{-i}) \frac{\partial}{\partial s_i} v_i(z, s_{-i}) \d z - c \cdot \alpha (s)\Bigg]\Bigg)\Bigg]\,,
\end{align}
for some $\lambda \geq 0$. Note that for $i \in \{1,2\}$, 
\[
\frac{\partial}{\partial s_i}v_i(s_1,s_2)=\1\{s_i \geq s_{-i}\}\,,
\]
and thus 
\[
\frac{1-G(s_1|s_2)}{g(s_1|s_2)}\frac{\partial}{\partial s_1}v_1(s_1,s_2)=\frac{1-G(\max\{s_1,s_2\}|\min\{s_1,s_2\})}{g(\max\{s_1,s_2\}|\min\{s_1,s_2\})}\,,
\]
if $s_1 \geq s_2$, and 
\[
\frac{1-G(s_2|s_1)}{g(s_2|s_1)} \frac{\partial}{\partial s_2}v_2(s_1,s_2)=\frac{1-G(\max\{s_1,s_2\}|\min\{s_1,s_2\})}{g(\max\{s_1,s_2\}|\min\{s_1,s_2\})}\,,
\]
if $s_2 \geq s_1$, for all $(s_1,s_2) \in [0,1]^2$. Moreover, for all $(s_1,s_2)\in [0,1]^2$, we have 
\begin{align*}
&v_1(s_1,s_2)+v_2(s_1,s_2)=|s_1-s_2|=\max\{s_1,s_2\}-\min\{s_1,s_2\}\,.
\end{align*}
Therefore, by changing variables and letting $x=\max\{s_1,s_2\}$ and $y=\min\{s_1,s_2\}$,  \eqref{eq:dual} can be written as 
\begin{equation}\label{eq:changed-dual}
\max_{\tilde{\alpha} \in \F}\int_{[0,1]^2} \tilde{\alpha}(x,y) \left[(1+\lambda)(x-y-c)-\lambda\frac{1-G(x|y)}{g(x|y)}\right]k(x,y)\d x \d y\,,
\end{equation}
where 
\[
k(x,y):=2\cdot \1\{x \geq y \}  g(x,y) 
\]
is the joint density of $\max\{s_1,s_2\}$ and $\min\{s_1,s_2\}$.\footnote{Note that for any symmetric monotone $\alpha(s_1, s_2)$, it induces some monotone $\tilde{\alpha}(x, y)$ defined on $\{(x, y): x  \geq y\}$ and vice versa. For notational convenience, we define $\tilde{\alpha}$ also on $[0, 1]^2$ and write $\tilde{\alpha} \in \mathcal{F}$ (any monotone extension of $\tilde{\alpha}$ from $\{(x, y): x  \geq y\}$ to the whole domain works for our purpose).} 

Moreover, by \Cref{prop:public}, there exist two nested up-sets $A_1 \subseteq A_2 \subseteq [0,1]^2$ and some $p \in [0, 1]$ such that 
\[
\tilde{\alpha}^\star(x,y)=p \1\{(x,y) \in A_1\}+(1-p) \1\{(x,y) \in A_2\}\,.
\]
Since $A_1$ and $A_2$ are nested up-sets, there exist nonincreasing, left-continuous functions $\gamma_1 \geq \gamma_2$ such that  
\[
A_j=\{(x,y):x \geq \gamma_j(y)\}
\]
almost everywhere. Therefore, 
\begin{align*}
&\max_{\tilde{\alpha} \in \F}\int_{[0,1]^2} \tilde{\alpha}(x,y) \left[(1+\lambda)(x-y-c)-\lambda\frac{1-G(x|y)}{g(x|y)}\right]k(x,y)\d x \d y\\
=&p\int_{[0,1]^2} \1\{x \geq \gamma_1(y)\}\left[(1+\lambda)(x-y-c)-\lambda\frac{1-G(x|y)}{g(x|y)}\right]k(x,y)\d x \d y\\
&+(1-p)\int_{[0,1]^2} \1\{x \geq \gamma_2(y)\}\left[(1+\lambda)(x-y-c)-\lambda\frac{1-G(x|y)}{g(x|y)}\right]k(x,y)\d x \d y\,.
\end{align*}
We claim that such $\gamma_1,\gamma_2$ must be constant on the relevant domain. Indeed, since $h(x|y)$ is strictly increasing in $x$ and decreasing in $y$ and since $\lambda \geq 0$, the function defined by 
\[
(x,y) \mapsto (1+\lambda) (x-y-c)-\lambda\frac{1-G(x|y)}{g(x|y)}
\]
must be strictly increasing in $x$ and decreasing in $y$. Then, combined with that $\tilde{\alpha}^*$ is optimal, this implies that for each $j \in \{1,2\}$, since the density $g$ is continuous and since $c \in (0,1)$, there exists some $(x_j^\star,y^\star_j)$, where $x_j^\star \geq y^\star_j$,  such that 
\begin{equation}\label{eq:zero}
(1+\lambda) (x^\star_j-y^\star_j-c)-\lambda\frac{1-G(x_j^\star|y_j^\star)}{g(x_j^\star|y_j^\star)}=0
\end{equation}
and that 
\[
x_j^\star \in [\gamma_j(y_j^{\star+}),\gamma_j(y_j^\star)]\,.
\]
Then, since
\[
(x,y) \mapsto (1+\lambda) (x-y-c)-\lambda\frac{1-G(x|y)}{g(x|y)}
\]
is strictly increasing in $x$ and decreasing in $y$, and since $\gamma_j$ and $\gamma_j^{-1}$ are nonincreasing, we have
\begin{align*}
&\int_{[0,1]^2} \1\{x \geq x_j^\star\} \left[(1+\lambda)(x-y-c)-\lambda\frac{1-G(x|y)}{g(x|y)}\right]k(x,y)\d x \d y\\
&-\int_{[0,1]^2} \1\{x \geq \gamma_j(y)\} \left[(1+\lambda)(x-y-c)-\lambda\frac{1-G(x|y)}{g(x|y)}\right]k(x,y)\d x \d y\\
=&\int_{x_j^\star}^1 \int_{0}^{\gamma_j^{-1}(x)} \left[(1+\lambda)(x-y-c)-\lambda\frac{1-G(x|y)}{g(x|y)}\right]k(x,y) \d y \d x\\
&-\int_{0}^{x_j^\star}\int_{\gamma_j^{-1}(x)}^1 \left[(1+\lambda)(x-y-c)-\lambda\frac{1-G(x|y)}{g(x|y)}\right]k(x,y) \d y \d x\\
\geq & \int_{x_j^\star}^1 \int_0^{\gamma_j^{-1}(x)} \left[(1+\lambda)(x_j^\star-y_j^\star-c)-\lambda\frac{1-G(x_j^\star|y_j^\star)}{g(x_j^\star|y_j^\star)}\right]k(x,y) \d y \d x\\
&-\int_{0}^{x_j^\star} \int_{\gamma_j^{-1}(x)}^1 \left[(1+\lambda)(x_j^\star-y_j^\star-c)-\lambda\frac{1-G(x_j^\star|y_j^\star)}{g(x_j^\star|y_j^\star)}\right]k(x,y) \d y \d x\\
=&0\,,
\end{align*}
where the last equality follows from \eqref{eq:zero}, and the inequality must be strict if $\gamma_j$ is not constant on $[0, x^\star_j]$. 

As a result, it must be that $A_1 \cap \{(x, y): x \geq y\}=\{(x,y): x \geq y, x \geq k_1\}$ and $A_2 \cap \{(x, y): x \geq y\}=\{(x,y): x \geq y, x \geq k_2\}$, for some $k_1, k_2 \in [0, 1]$. Changing variables back to $(s_1,s_2)$ using $x=\max\{s_1,s_2\}$ and $y=\min\{s_1,s_2\}$ (i.e., extending $\tilde{\alpha}^\star$ from $\{(x, y): x \geq y\}$ to the whole domain by symmetry), it then follows that any symmetric optimal allocation rule in the original space must satisfy 
\[
\alpha^\star(s_1,s_2) =p\1\{(s_1,s_2): \max\{s_1,s_2\} \geq k_1\}+(1-p)\1\{(s_1,s_2): \max\{s_1,s_2\} \geq k_2\}\,.
\]
The proof is then completed by noting that $\alpha^\star$ is a mixture of two externality-refund mechanisms.

\section{Additional Results}\label{app:additional}
\subsection{Exposed Points of \texorpdfstring{$\mathcal{F}$}{}}
An extreme point $f$ of $\F$ is \textit{\textbf{exposed}} if there exists an essentially bounded $\phi:[0,1]^n \to \R$ such that 
\[
\int_{[0,1]^n} f(x)\phi(x) \d x >\int_{[0,1]^n} \hat{f}(x) \phi(x) \d x\,,
\]
for all $\hat{f} \in \F \backslash \{f\}$. 

\begin{taggedtheorem}{B.1}\label{thm:exposedF}
   \textit{ $f \in \mathcal{F}$ is an exposed point of $\mathcal{F}$ if and only if $f= \1_A$ for some up-set $A$. }
\end{taggedtheorem}

\begin{proof}[Proof of \Cref{thm:exposedF}]
For any exposed point $f \in \F$, $f$ must be an extreme point. Therefore, $f=\1_A$ for some up-set $A$ by \Cref{thm:choquet}. Now, consider any up-set $A \subseteq [0,1]^n$. Let $\phi:[0,1]^n \to \R$ be defined as 
\[
\phi(x):=\begin{cases}
1,&\mbox{if } x \in A\\
-1,&\mbox{if } x \notin A
\end{cases}\,.
\]
Therefore, for any $f \in \F$, 
\[
\int_{[0,1]^n} \phi(x)f(x) \d x=\int_A f(x) \d x -\int_{[0,1]^n \backslash A} f(x) \d x \leq \int_A 1 \d x=\int_{[0,1]^n} \phi(x)\1_A(x) \d x\,,
\]
and the inequality is strict whenever $f \neq \1_A$. Thus, $\1_A$ is an exposed point of $\F$ as it uniquely maximizes 
\[
f \mapsto \int_{[0,1]^n} \phi(x) f(x) \d x\,,
\]
concluding the proof. 
\end{proof}

\subsection{Exposed Points of \texorpdfstring{$\cQ$}{}}
An extreme point $q$ of $\cQ$ is said to be an \textbf{\textit{exposed point}} if there exist essentially bounded functions $\{\psi_{i}\}_{i=1}^n$ such that 
\[
\sum_{i=1}^n \int_0^1 q_i(z) \psi_i(z) \d z>\sum_{i=1}^n \int_0^1 \hat{q}_i(z) \psi_i(z) \d z 
\]
for all $\hat{q} \in \cQ\backslash \{q\}$. 

An up-set $A \subseteq [0,1]^n$ is an \textbf{\textit{additive set}} if there exist nondecreasing functions $\{\phi_i\}_{i=1}^n$ such that 
\[
A=\left\{x \in [0, 1]^n : \sum_{i=1}^n \phi_i(x_i) \geq 0\right\}\,.
\]
Moreover, we say that an additive set $A \subseteq [0,1]^n$ is \textbf{\textit{strongly additive}} if at least one of the nondecreasing functions $\{\phi_i\}_{i=1}^n$ can be chosen to be strictly increasing.

\begin{taggedtheorem}{B.2}\label{thm:exposedQ}
\textit{Every exposed point of $\cQ$ is rationalized by $\1_A$ for an additive set $A$. Conversely, if $q \in \cQ$ is rationalized by  $\1_A$ for a strongly additive set $A$, then $q$ is an exposed point of $\cQ$.}   
\end{taggedtheorem}

\begin{proof}[Proof of \Cref{thm:exposedQ}]
For sufficiency, fix any strongly additive set $A$ defined by  nondecreasing functions $\{\phi_i\}_{i=1}^n$ where at least one of them is strictly increasing. Let $q$ be the projection of $f:=\1_A$. Then, for any $\hat{q} \in \cQ \backslash\{q\}$, $\hat{q}$ is rationalized by some $\hat{f}\neq f$, which implies 
\begin{align*}
\sum_{i=1}^n \int_0^1 q_i(z)\phi_i(z)\d z=& \sum_{i=1}^n \int_{[0,1]^n} \phi_i(x_i)f(x) \d x\\
=& \int_{[0,1]^n} \sum_{i=1}^n \phi_i(x_i) \1\left\{x: \sum_{i=1}^n \phi_i(x_i) \geq 0 \right\} \d x\\
>& \int_{[0,1]^n} \sum_{i=1}^n \phi_i(x_i) \hat{f}(x) \d x\\
=& \sum_{i=1}^n \int_{[0,1]^n} \phi_i(x_i) \hat{f}(x) \d x \\
=& \sum_{i=1}^n \int_0^1 \phi_i(z) \hat{q}_i(z) \d z\,,
\end{align*}
where the strict inequality holds since the set 
\[\Bigg\{x \in [0, 1]^n: \sum_{i=1}^n \phi_i(x_i) = 0\Bigg\}\]
has measure zero given that at least one of the $\phi_i$'s is strictly increasing. Thus, by taking $\psi_i:=\phi_i$ for all $i$, it then follows that $q$ is exposed. 

For necessity, suppose that $q$ is an exposed point of $\cQ$. Then there exist essentially bounded functions $\{\psi_i\}_{i=1}^n$ such that 
\[
\sum_{i=1}^n \int_0^1 q_i(z) \psi_i(z) \d z > \sum_{i=1}^n \int_0^1 \hat{q}_i(z) \psi_i(z) \d z 
\]
for all $\hat{q} \neq q \in \cQ$. Define an additive set $A$ as 
\[
A:=\left\{x \in [0, 1]^n: \sum_{i=1}^n \overline{\psi}_i(x_i) \geq 0\right\}\,,
\]
where $\overline{\psi}_i$ is the ironed version of $\psi_i$, i.e., $\overline{\psi}_i$ is the (left-continuous) subgradient of the convex closure of $z \mapsto \int_0^z \psi_i(t) \d t$. Note that, by construction, for all $i$ and for all $\hat{q} \in \cQ$,
\begin{equation}\label{eq:ironingeq}
\int_0^1 \hat{q}_i(z)\psi_i(z) \d z \leq \int_0^1 \hat{q}_i(z) \overline{\psi}_i(z) \d z\,,
\end{equation}
and the equality holds whenever $\hat{q}_i$ is constant on all intervals in which $\overline{\psi}_i$ is constant. 

Let $q^\star$ be the projection of $\1_A$. We claim that $q\equiv q^\star$. To see this, first note that since
\[
q_i^\star(x_i)=\int_{[0,1]^{n-1}} \1\left\{x: \overline{\psi}_i(x_i) \geq -\sum_{j\neq i} \overline{\psi}_j(x_j)\right\}\d x_{-i}\,,
\]
for all $x_i \in [0,1]$, $q_i^\star$ is constant whenever $\overline{\psi}_i$ is constant. Therefore, by \eqref{eq:ironingeq}, 
\[
\sum_{i=1}^n \int_0^1 q_i^\star(z) \overline{\psi}_i(z) \d z=\sum_{i=1}^n \int_0^1 q_i^\star(z) \psi_i(z) \d z\,.
\]
Now let $f \in \F$ be any function that rationalizes $q$. Then, 
\begin{align*}
\sum_{i=1}^n \int_0^1 q_i(z) \psi_i(z) \d z \leq& \sum_{i=1}^n \int_0^1 q_i(z) \overline{\psi}_i(z) \d z \tag{by \eqref{eq:ironingeq}}\\
=& \int_{[0,1]^n} f(x) \sum_{i=1}^n \overline{\psi}_i(x_i) \d x \tag{$f$ rationalizes $q$}\\
\leq & \int_{[0,1]^n } \1_A(x) \sum_{i=1}^n \overline{\psi}_i(x_i) \d x\tag{by definition of $A$}\\
=& \sum_{i=1}^n \int_0^1 q_i^\star(z) \overline{\psi}_i(z) \d z \tag{$\1_A$ rationalizes $q^\star$} \\
=& \sum_{i=1}^n \int_0^1 q_i^\star(z) \psi_i(z) \d z \tag{by \eqref{eq:ironingeq}}\,.
\end{align*}
However, since $q$ is the unique maximizer of 
\[
\tilde{q} \mapsto \sum_{i=1}^n \int_0^1 \tilde{q}_i(z) \psi_i(z) \d z 
\]
over $\cQ$ and since $q^\star \in \cQ$, it must be that $q \equiv q^\star$, as desired. 
\end{proof}

\subsection{Projections under General Product Measures}\label{subsec:measure}
In \Cref{rmk:projection}, we note that \Cref{thm:rationalized-upsets} and \Cref{thm:rectangle} can be extended to projections that are defined under different measures. Specifically, fix any pair of CDFs $G_1,G_2$ that are continuous and have full supports on some compact intervals $[\underline{x}_1,\overline{x}_1]$ and $[\underline{x}_2,\overline{x}_2]$ respectively.  We say that a pair of nondecreasing left-continuous functions $q=(q_1,q_2)$ is $(G_1,G_2)$\textit{\textbf{-rationalized}} by a function $f:[\underline{x}_1,\overline{x}_1] \times [\underline{x}_2,\overline{x}_2] \rightarrow [0,1]$ if 
\[
q_1(x_1)=\int_{\underline{x}_2}^{\overline{x}_2} f(x_1,x_2)G_2(\d x_2)\,, \quad \mbox{ and } \quad q_2(x_2)=\int_{\underline{x}_1}^{\overline{x}_1} f(x_1,x_2)G_1(\d x_1)\,,
\]
for all $x_1\in [\underline{x}_1,\overline{x}_1]$ and $x_2 \in [\underline{x}_2,\overline{x}_2]$. Let $\cQ^{(G_1,G_2)}$ be the set of pairs of nondecreasing functions that are $(G_1,G_2)$-rationalizable. Furthermore, for any essentially bounded functions $\psi_1,\psi_2$, and for any $\eta \in \R$, let 
\[
\overline{\cQ}^{(G_1,G_2)}:=\left\{q \in \cQ^{(G_1,G_2)}: \int_{\underline{x}_1}^{\overline{x}_1} q_1(x_1)\psi_1(x_1)G_1(\d x_1)+\int_{\underline{x}_2}^{\overline{x}_2} q_2(x_2)\psi_2(x_2)G_2(\d x_2) \leq \eta\right\}\,.
\]

By simply changing variables and letting $t_1:=G_1(x_1)$ and $t_2:=G_2(x_2)$, it follows immediately that, for any $f:[\underline{x}_1,\overline{x}_1] \times [\underline{x}_2,\overline{x}_2] \rightarrow [0,1]$, 
\[
\int_{\underline{x}_2}^{\overline{x}_2} f(x_1,x_2) G_2(\d x_2)=\int_0^1 f(x_1,G_2^{-1}(t_2))\d t_2
\]
and 
\[
\int_{\underline{x}_1}^{\overline{x}_1} f(x_1,x_2) G_1(\d x_1)=\int_0^1 f(G_1^{-1}(t_1),x_2)\d t_1\,.
\]
Therefore, any $f:[\underline{x}_1,\overline{x}_1] \times [\underline{x}_2,\overline{x}_2] \rightarrow [0,1]$ uniquely corresponds to some $\tilde{f}:[0,1]^2 \rightarrow [0,1]$, where 
\[
\tilde{f}(t_1,t_2)=f(G_1^{-1}(t_1),G_2^{-1}(t_2))\,,
\]
and any $q \in \cQ^{(G_1,G_2)}$ uniquely corresponds to some $\tilde{q} \in \cQ$, where 
\[
\tilde{q}_1(t_1)=q_1(G_1^{-1}(t_1))\,,\quad \mbox{ and } \quad \tilde{q}_2(t_2)=q_2(G^{-1}_2(t_2))\,,
\]
for all $(t_1,t_2)\in [0,1]^2$. By \Cref{thm:rationalized-upsets} and \Cref{thm:rectangle}, we then immediately have the following corollaries. 

\begin{taggedcorollary}{B.1}\label{cor:rationalized-upsets}
\textit{$q=(q_1,q_2)$ is an extreme point of $\cQ^{(G_1,G_2)}$ if and only if $q$ is $(G_1, G_2)$-rationalized by $\1_A$ for some up-set $A \subseteq [\underline{x}_1,\overline{x}_1] \times [\underline{x}_2,\overline{x}_2]$. Moreover, every extreme point of $\cQ^{(G_1,G_2)}$ is uniquely $(G_1,G_2)$-rationalized. }
\end{taggedcorollary}

\begin{taggedcorollary}{B.2}\label{cor:rectangle}
\textit{Every extreme point of $\overline{\cQ}^{(G_1, G_2)}$ is a mixture of $\1_A$ and $\1_{A'}$, where $A' \subseteq A \subseteq [\underline{x}_1,\overline{x}_1] \times [\underline{x}_2,\overline{x}_2]$ are nested up-sets that differ by at most a rectangle. Moreover, every extreme point of $\overline{\cQ}^{(G_1,G_2)}$ is uniquely $(G_1,G_2)$-rationalized among all monotone functions.}
\end{taggedcorollary}

\end{document}